\documentclass[11pt,a4paper]{article}
\usepackage{geometry}
\usepackage{graphicx}
\usepackage{bm,array}
\usepackage{comment}
\usepackage{caption}
\usepackage{subcaption}
\usepackage{tablefootnote}
 \usepackage{makecell}
\usepackage{booktabs}
 \usepackage{multirow}
 \usepackage{float}
 \usepackage{appendix}
 \usepackage{apacite}
\usepackage{tikz}
\usetikzlibrary{arrows.meta, positioning, calc}
\usepackage[normalem]{ulem}

\usepackage[pdfpagemode={UseOutlines},bookmarks=true,bookmarksopen=true,
bookmarksopenlevel=0,bookmarksnumbered=true,hypertexnames=false,
colorlinks,linkcolor={blue},citecolor={blue},urlcolor={red},
pdfstartview={FitV},unicode,breaklinks=true]{hyperref}
\hypersetup{urlcolor=blue, colorlinks=true}

\usepackage{setspace}
\usepackage{vmargin}
\setmarginsrb{ 1.0in}  % left margin
{ 0.6in}  % top margin
{ 1.0in}  % right margin
{ 0.8in}  % bottom margin
{  20pt}  % head height
{0.25in}  % head sep
{9pt}  % foot height
{ 0.3in}  % foot sep

\usepackage{authblk}
\usepackage{amsmath}
\allowdisplaybreaks
\usepackage{amsfonts}
\usepackage{amssymb}
\usepackage[round, sort,comma,authoryear]{natbib}

\newtheorem{theorem}{Theorem}

\newtheorem{example}[theorem]{Example}

\newtheorem{proposition}{Proposition}

\newenvironment{proof}[1][Proof]{\noindent\textbf{#1.} }{\ \rule{0.5em}{0.5em}}
\usepackage{mdframed}

\usepackage{multirow}% http://ctan.org/pkg/multirow
\usepackage{hhline}% http://ctan.org/pkg/hhline
\usepackage{footnote}
\usepackage[ruled,vlined]{algorithm2e}
\usepackage{algorithmic}

\newcommand{\cX}{{\mathcal{X}}}

\newcommand{\cY}{{\mathcal{Y}}}

\newcommand{\cPA}{{\mathcal{Z}}}

\newcommand{\bB}{\textbf{B}}
\newcommand{\bC}{\textbf{C}}

\newcommand{\bx}{\textbf{x}}
\newcommand{\by}{\textbf{y}}
\newcommand{\bs}{\textbf{s}}

\newcommand{\bq}{\textbf{q}}
\newcommand{\bu}{\textbf{u}}

\newcommand{\bU}{\textbf{U}}

\newcommand{\bA}{\textbf{A}}

\newcommand{\ba}{\textbf{a}}

\newcommand{\bz}{\textbf{z}}
\newcommand{\bt}{\textbf{t}}

\newcommand{\bb}{\textbf{b}}
\newcommand{\bh}{\textbf{h}}
\newcommand{\bk}{\textbf{k}}

\newcommand{\bW}{\textbf{W}}

\usepackage{mathtools}

\usepackage{mathtools}

\newif\ifnotes\notestrue
\def\htien#1{}

\begin{document}
%%%%%%%%%%%%%%%%

%\maketitle
\newcolumntype{C}{>{\centering\arraybackslash}p{4em}}

\title{\textbf{Constrained Assortment and  Price Optimization under Generalized Nested Logit Models}}
\author[1]{Hoang Giang Pham}
\author[1]{Tien Mai}
\affil[1]{\it\small
School of Computing and Information Systems, Singapore Management University}
%\affil[4]{\it\small Corresponding author}

\date{}
\maketitle

\begin{abstract}
We study assortment and price optimization under the generalized nested logit (GNL) model, one of the most general and flexible modeling frameworks in discrete choice modeling. Despite its modeling advantages, optimization under GNL is highly challenging: even the pure assortment problem is NP-hard, and existing approaches rely on approximation schemes or are limited to simple cardinality constraints. In this paper, we develop the first exact  and near-exact algorithms for constrained assortment and joint assortment--pricing optimization (JAP) under GNL. Our approach reformulates the problem into bilinear and exponential-cone convex programs and exploits convexity, concavity, and submodularity properties to generate strong cutting planes within a Branch-and-Cut framework (B\&C). We further extend this framework to the mixed GNL (MGNL) model, capturing heterogeneous customer segments, and to JAP with discrete prices. For the continuous pricing case, we propose a near-exact algorithm based on piecewise-linear approximation (PWLA) that achieves arbitrarily high precision under general linear constraints. Extensive computational experiments demonstrate that our methods substantially outperform state-of-the-art approximation approaches in both solution quality and scalability. In particular, we are able to solve large-scale instances with up to 1000 products and 20 nests, and to obtain near-optimal solutions for continuous pricing problems with negligible optimality gaps. To the best of our knowledge, this work resolves several open problems in assortment and price optimization under GNL.
\end{abstract}

{\bf Keywords} :Assortment Optimization, Price Optimization, Generalized Nested Logit, Branch-and-Cut

\noindent
\textbf{Notation:}
Boldface characters represent matrices (or vectors), and $a_i$ denotes the $i$-th element of vector $\ba$. We use $[m]$, for any $m\in \mathbb{N}$, to denote the set $\{1,\ldots,m\}$.

%%%%%%%%%%%%%%%%%%%%%%%%%%%%%%%%%%%%%%%%%%%%%%%%%%%%%%%%%%%%%%%%%%%%%%

\section{Introduction}

Assortment and price optimization are central problems in revenue management, with broad applications in retailing, e-commerce, and service industries \citep{TalluriVanRyzin2004MS}. Firms must determine which products to offer (assortment decisions) and at what prices (pricing decisions) in order to maximize expected revenues, while respecting operational and business constraints. These problems have attracted significant attention due to their practical impact on demand forecasting, personalized recommendations, and dynamic pricing strategies \citep{LiHuh2011, gallego2014multiproduct}.

A standard approach to assortment and price optimization relies on discrete choice modeling, which captures consumer purchase behavior in response to available alternatives. Models such as the multinomial logit (MNL) and nested logit (NL) have been widely used for this purpose because of their analytical tractability and interpretability \citep{rusmevichientong2010dynamic,rusmevichientong2012robust,davis2014assortment,gallego2014constrained}. By linking product attributes to customer preferences, these models enable firms to evaluate expected revenues under different assortment and pricing decisions, thereby forming the basis of optimization strategies.

Among discrete choice models, the generalized nested logit (GNL) model, also known as the cross-nested logit (CNL), stands out as one of the most general and flexible frameworks \citep{train2009discrete, ,fosgerau2013choice,vovsha1997application}. Unlike the MNL or NL, the GNL model allows products to belong to multiple nests simultaneously, with fractional membership weights capturing heterogeneous substitution patterns across product categories \citep{ fosgerau2013choice,vovsha1997application}. This flexibility makes GNL highly effective in modeling complex customer choice behavior, and, as a result, assortment and price optimization under GNL are both practically relevant and theoretically important.

Despite its modeling advantages, optimization under the GNL model is computationally challenging. Even the pure assortment optimization problem is NP-hard, including the special case with only two nests, and is NP-hard to approximate within any constant factor in general \citep{Cuong2024}. Incorporating pricing decisions further increases the computational complexity of the problem.
Existing approaches are either limited to approximation schemes with no guarantee of global optimality \citep{Cuong2024}, or can only handle very restricted forms of constraints such as simple cardinality limits \citep{AlfandariHassanzadehLjubic2021}. To the best of our knowledge, no prior method is capable of solving assortment and price optimization under GNL to exact or near-optimality under general linear constraints.

In this paper, we address these challenges by developing exact and near-exact solution methods for assortment and price optimization under GNL. Specifically, we propose the first exact cutting-plane framework for assortment optimization and JAP with discrete prices. For continuous pricing, we design a near-optimal method with arbitrary precision guarantees. Our formulations can incorporate any general linear constraints on assortments and prices, substantially extending the scope of solvable problems beyond existing literature. \textit{To the best of our knowledge, this is the first work to provide exact  and near-exact algorithms for assortment and price optimization under GNL, thereby resolving several long-standing open problems in this domain.}

\paragraph{Paper Contributions.} 
We summarize our main contributions as follows:
\begin{itemize}
    \item \textbf{Exact solution framework for assortment optimization under GNL.}  
    We develop the first exact algorithm tailored for assortment optimization under the GNL. Our approach reformulates the problem into bilinear or exponential-cone convex programs, enabling the use of cutting-plane techniques to certify optimality. We further derive strong validity cuts by exploiting convexity, concavity, and submodularity properties of different model components, thereby strengthening the optimization process. Importantly, our framework is general: it can handle assortment optimization under any form of GNL model and accommodate arbitrary linear constraints on the assortment. Prior work, such as \citep{Cuong2024}, addressed constrained assortment optimization under GNL but only provided approximation methods whose model size grows with the desired solution precision, in contrast to our exact approach.

    \item \textbf{Extension to MGNL.}  
    We extend our exact framework to multi-segment customer choice models, where heterogeneous customers are represented by distinct GNL structures. Our solution relies on either a convex  reformulation or a bilinear-convex formulation, in which key components are shown to be convex. We also introduce valid cuts based on convexity and submodularity, which are integrated into a B\&C procedure.
    To the best of our knowledge, this is the first optimal method for assortment optimization under the generalized MGNL model.

    \item \textbf{JAP with discrete prices.}  
    Building on our framework for assortment optimization under GNL, we formulate and solve the JAP problem when products are associated with discrete price levels. By leveraging our bilinear and exponential-cone reformulations together with cutting-plane techniques, we provide the first exact algorithm for this important extension.

    \item \textbf{JAP with continuous prices.}  
    We further study the JAP problem with continuous prices, which aligns more naturally with discrete choice principles than the discrete-price setting. This formulation involves both discrete and continuous decision variables and is particularly challenging under general linear constraints on assortments and prices. We propose a near-exact solution method based on PWLA, which achieves arbitrarily high precision. To the best of our knowledge, this is the first near-optimal approach for continuous pricing under GNL that is capable of handling arbitrary model structures and general linear constraints.

     \item \textbf{Comprehensive numerical experiments.}  
    We conduct extensive computational experiments across a wide range of datasets, including GNL, MGNL, and JAP problems with both discrete and continuous prices. Our results show that the proposed B\&C algorithms significantly outperform state-of-the-art approximation and exact methods  in terms of scalability, solution quality, and robustness. In particular, our exact methods successfully solve large-scale instances with up to 1000 products and 20 nests, while our continuous pricing approach achieves near-optimal solutions with negligible optimality gaps. These results highlight both the practical relevance and computational efficiency of our proposed methods in handling such challenging  optimization problems.

   In addition to comparing the proposed algorithms with state-of-the-art baselines, our results provide strong evidence that explicitly accounting for customer heterogeneity is economically significant. When markets are heterogeneous, decisions based on the MGNL model outperform those derived from an aggregated GNL model by up to 33\%, with average revenue improvements of approximately 26\%, underscoring the substantial cost of ignoring segment-level differences.

Moreover, in the context of the JAP problem—where firms must simultaneously determine product offerings and continuous prices under a highly nonconvex structure—our results show that the proposed algorithm can achieve revenue gains of up to 20\% relative to a state-of-the-art general-purpose MINLP solver (SCIP). These findings highlight the effectiveness of the proposed approach in tackling complex, constrained pricing and assortment decisions that are beyond the reach of generic optimization tools.

\end{itemize}

\paragraph{Paper outline:} 
The remainder of the paper is organized as follows. 
Section~\ref{sec:review} reviews the related literature. 
Section~\ref{sec:Assort-GNL} introduces the assortment optimization problem under the GNL model. 
Section~\ref{sec:convexification_gnl} develops several exact reformulations, including the bilinear--convex formulation and convex reformulations based on logarithmic transformations, and discusses their structural properties. 
Section~\ref{sec:JAP} studies the JAP problem. 
Section~\ref{sec:experiments} reports numerical experiments evaluating the computational performance and solution quality of the proposed methods. 
Finally, Section~\ref{sec:concl} concludes the paper and outlines directions for future research.

\section{Literature Review}\label{sec:review}
Assortment and pricing optimization are central topics in operations research and management science, with wide-ranging applications in retailing and revenue management. Over the past decades, these problems have received sustained and growing attention, driven by the increasing availability of demand data and the practical need to optimize product offerings and prices under business constraints \citep{TalluriVanRyzin2004MS,TalluriVanRyzin2004Book}.

A key modeling approach in this literature is the use of discrete choice models to capture customer substitution behavior. Due to their relatively simple structure and favorable computational properties, the MNL and NL models have been the most extensively studied \citep{train2009discrete}. Under the MNL model, there exists a rich body of work on assortment optimization. Early studies integrate estimation and optimization under operational constraints, including joint learning and assortment optimization with capacity constraints \citep{rusmevichientong2010dynamic}, as well as robust formulations that account for uncertainty in choice parameters \citep{rusmevichientong2012robust}. The NL model extends the MNL framework by allowing correlation among products within nests, but at the cost of increased computational complexity. \citet{davis2014assortment} show that the unconstrained assortment optimization problem under the NL model can be solved in polynomial time when the dissimilarity parameters of all nests are less than one and customers always make a purchase after choosing a nest; relaxing either condition renders the problem NP-hard. Consequently, much of the literature focuses on approximation algorithms for more challenging settings. For the unconstrained assortment problem with $m$ products and $n$ nests, \citet{li2014greedy} propose a greedy algorithm with complexity $\mathcal{O}(mn\log n)$. \citet{gallego2014constrained} study assortment optimization under per-nest capacity constraints and develop an approximation algorithm with a $0.5$ performance guarantee under certain conditions. For problems with a single global capacity constraint, \citet{rusmevichientong2009ptas} and \citet{feldman2015capacity} propose constant-factor approximation algorithms, while \citet{segev2022approximation} develop an approximation scheme that can achieve arbitrarily tight solutions.

To capture richer substitution patterns and customer heterogeneity, several studies extend assortment optimization to a mixture of MNL models (MMNL). \citet{bront2009column} and \citet{mendez2014branch} propose mixed-integer reformulations and B\&C algorithms for this setting. From a complexity perspective, \citet{rusmevichientong2014assortment} show that assortment optimization is NP-hard even when the choice model is a mixture of only two MNL models. These hardness results have motivated the development of scalable approximation and exact methods, including an FPTAS for capacitated assortment problems \citep{Desiretal2022}, conic quadratic mixed-integer formulations \citep{sen2018conic}, and more recent outer-approximation–based approaches designed for large-scale instances \citep{PhamMai2025MixedLogitOA}. Relatedly, \citet{abdallah2021demand} study joint assortment optimization and customization, highlighting the value of explicitly accounting for customer heterogeneity in decision making.

Both the MNL and NL models suffer from structural limitations in capturing complex substitution patterns observed in practice. 
The MNL model’s independence of irrelevant alternatives (IIA) property implies identical cross-elasticities across products, which is often unrealistic in assortment and pricing applications \citep{train2009discrete,gallego2014constrained}. 
Although the NL model relaxes the IIA property by allowing correlation within predefined nests, its rigid hierarchical structure still limits modeling flexibility \citep{Cuong2024}.
The GNL model effectively addresses these limitations by extending the NL framework to allow each product to belong to multiple nests. This added flexibility enables the GNL model to capture rich and overlapping substitution patterns that cannot be represented by standard MNL or NL models. From a theoretical perspective, the GNL model is highly expressive and can approximate a broad class of parametric and nonparametric discrete choice models \citep{fosgerau2013choice,Cuong2024}.
 Beyond its theoretical appeal, the GNL model been successfully applied in a wide range of transportation applications. These include mode choice \citep{yang2013cross}, departure time choice \citep{ding2015cross}, route choice and revenue management \citep{lai2015specification,mai2016method,mai2017dynamic}, air travel management \citep{drabas2013modelling}, and, more recently, location choice in international migration \citep{beine2021new}. Empirical evidence in these studies consistently shows that the CNL model outperforms simpler alternatives such as the MNL and NL models in terms of predictive accuracy.

Despite its success in transportation and travel-demand modeling, the potential of the GNL model in revenue management and assortment optimization remains largely unexplored. Only a limited number of recent studies have begun to address this gap. For instance, \citet{Cuong2024} study assortment optimization under the GNL model with general linear constraints but rely on approximation methods that do not guarantee global optimality. \citet{Zhang2024} propose a B\&C  and heuristic algorithms for assortment optimization under the CNL model; however, their approach is restricted to a simple cardinality constraint and cannot be extended to JAP decisions. Related work by \citet{zhang2020assortment} and \citet{ghuge2022constrained} considers assortment optimization under the paired combinatorial logit model, which constitutes a special case of the GNL framework. 

Pricing is equally central to revenue management and has been extensively studied in the context of choice-based demand models. When prices are restricted to a discrete set, pricing problems can be reformulated as assortment optimization problems over an expanded set of product--price alternatives. This modeling approach is widely adopted in the literature, as it allows pricing decisions to be handled using assortment optimization techniques developed for discrete choice models \citep{Cuong2024, gallego2014multiproduct}. 

When prices are treated as continuous decision variables—which is more consistent with the foundations of discrete choice theory—the resulting optimization problems become substantially more challenging. A stream of work focuses on unconstrained pricing problems under restrictive structural assumptions, such as homogeneous price sensitivity across products. Under these conditions, the pricing problem admits tractable solutions and optimal prices can be computed efficiently \citep{LiHuh2011,MaiJaillet2019GEVpricing,ZhangRusmevichientongTopaloglu2018}. Other studies relax these assumptions by allowing product-specific price sensitivities \citep{gallego2014multiproduct} or by imposing constraints on purchasing probabilities, and show that, after appropriate variable transformations, the resulting problems exhibit convexity properties that can be exploited computationally \citep{LiHuh2011,ZhangRusmevichientongTopaloglu2018,VanDeGeerDenBoer2023}.

Despite these advances, most of the existing literature does not explicitly incorporate general constraints on prices, even though such constraints are common in practice. Only a limited number of studies address constrained pricing problems, and typically under restrictive forms of constraints. For instance, \citet{Rayfield2015} develop approximation schemes for pricing under the NL model with price bounds, but their approach is limited to simple box constraints and does not extend to more general linear pricing restrictions.

Our work is also related to the literature on JAP optimization with continuous prices. Under the MNL model, \citet{wang2012capacitated} and \citet{chen2020capacitated} study capacitated JAP problems and show that the joint decision problem can be reformulated as a fixed-point problem, enabling efficient solution methods. In contrast, much less is known for richer choice models and more general constraint structures, motivating the need for new models and optimization frameworks.

 % Classic results show that profit maximization under discrete choice demand can be nonconcave in prices, complicating global optimization \citep{HansonMartin1996}. A substantial stream of work studies pricing and JAP under MNL/NL-type models, often leveraging additional structure to recover tractability, e.g., concavity in market shares or restricted price-sensitivity patterns \citep{LiHuh2011,MaiJaillet2019GEVpricing,ZhangRusmevichientongTopaloglu2018}. The pricing optimization or JAP under the NL model has also been studied under specific hierarchical structures and modeling assumptions \citep{gallego2014constrained,Rayfield2015}. However, much of the existing literature focuses on unconstrained pricing or joint decisions with limited operational constraints, despite the practical importance of feasibility restrictions such as price bounds, linear side constraints, and policy constraints.

By developing exact and scalable algorithms for assortment and JAP under the GNL model with general linear constraints, this paper addresses several open questions at the intersection of expressive choice modeling and exact optimization. In particular, our framework bridges the gap between flexible substitution modeling (GNL/MGNL) and the algorithmic requirements of large-scale constrained optimization, while also advancing the study of JAP when continuous prices and realistic constraints are incorporated.

\noindent
\textbf{Notation:}
Boldface characters represent matrices (or vectors), and $a_i$ denotes the $i$-th element of vector $\ba$. We use $[m]$, for any $m\in \mathbb{N}$, to denote the set $\{1,\ldots,m\}$.
\section{Assortment Optimization under GNL Model}\label{sec:Assort-GNL}
Discrete choice models provide a probabilistic framework for describing how consumers select among a finite set of alternatives. Among these, the GNL is one of the most flexible extensions of the MMNL model. The GNL model captures complex substitution patterns by allowing products to belong simultaneously to multiple nests, with fractional membership weights.

Consider a set of $m$ products indexed by $[m] := \{1,2,\ldots,m\}$. Each product $i \in [m]$ is associated with a deterministic utility component $V_{in} > 0$, which reflects its attractiveness to consumers. The $m$ products are grouped into $N$ overlapping nests, denoted by $\{S_1, S_2, \ldots, S_N\}$, where each nest represents a category of alternatives with correlated preferences. Unlike the standard NL, in which nests are disjoint, the GNL model allows each product to appear in multiple nests. The membership of product~$i$ in nest~$n$ is specified by a nonnegative parameter $\alpha_{in}$. If product~$i$ does not belong to nest~$n$, then $\alpha_{in}=0$. It is common to assume that
$\sum_{n=1}^N \alpha_{in} = 1, \forall i \in [m],$
so that the membership parameters can be interpreted as fractional allocations of product~$i$ across nests. However, this normalization is not required for the GNL model to be consistent with the random utility maximization framework.

Each nest $n$ is characterized by a dissimilarity parameter $\sigma_n \in (0,1]$, which governs the correlation among alternatives within the nest, and by an outside option with utility weight $V_{0n} > 0$, representing the possibility that the consumer leaves the nest without making a purchase. Given an offered product set $\bx = (x_1,\ldots,x_m) \in \{0,1\}^m$, the total preference weight in nest $n$ is defined as:
$
W_n(\bx) = V_{0n} + \sum_{i \in [m]} \alpha_{in} x_i V_{in}
$, 
where $V_{in} = \exp(v_{in}/\sigma_n)$, for $i=0,\ldots,n$. The GNL model assumes a two-stage choice process:
\begin{enumerate}
    \item \textbf{Nest selection.} The consumer first selects a nest $n \in [N]$ with probability
    \[
    P(S_n \mid \bx) = \frac{W_n(\bx)^{\sigma_n}}{\sum_{n'=1}^N W_{n'}(\bx)^{\sigma_{n'}}}.
    \]
    \item \textbf{Product selection within nest.} Conditional on nest $n$ being chosen, the consumer selects product $i \in S_n$ with probability
    \[
    P(i \mid S_n, \bx) = \frac{\alpha_{in} x_i V_{in}}{W_n(\bx)}.
    \]
\end{enumerate}
Combining these two stages, the unconditional probability that product $i$ is purchased is
\[
P(i \mid \bx) = \sum_{n=1}^N \frac{W_n(\bx)^{\sigma_n - 1} \, \alpha_{in} x_i V_{in}}{\sum_{n'=1}^N W_{n'}(\bx)^{\sigma_{n'}}}.
\]
The GNL model generalizes both the MNL and the NL. If each product belongs exclusively to one nest, i.e., $\alpha_{in} \in \{0,1\}$ with $\sum_n \alpha_{in}=1$, the model reduces to the NL. If there is only one nest, it collapses to the MNL. By allowing partial membership across multiple nests, GNL captures richer substitution patterns and correlation structures, making it a powerful tool in applications where products span overlapping categories or features.

% Using this, the expected revenue obtained from a single customer can be expressed as
% \begin{align}
%     F(\bx)  = \frac{\sum_{i \in [m]} \sum_{n \in [N]} W_n^{\sigma_n - 1} \, (\alpha_{in} x_i r_i V_{in})}{\sum_{n \in [N]} W_{n}^{\sigma_{n}}}, \nonumber
% \end{align}
% where $W_n = V_{0n} + \sum_{i \in [m]} \alpha_{in} x_i V_{in}$ denotes the total preference weight of nest $n$ and $r_i$ denotes the revenue of product $i\in [m]$. 

Let us shift to the formulation of the assortment optimization under the GNL model. Given an offered assortment $\bx = (x_1, \ldots, x_m) \in \{0,1\}^m$, the probability that a product $i \in [m]$ is purchased under the GNL model has been derived in the previous section. Moreover, in practice, assortment decisions are often subject to various business rules and operational requirements. These may include capacity restrictions, category or shelf-space limits, or more general resource allocation constraints. To capture such requirements, we let
$\cX = \{ \bx \in \{0,1\}^m \mid \bA \bx \leq \bb \},$
where $\bA \in \mathbb{R}^{p \times m}$ and $\bb \in \mathbb{R}^p$ encode the set of feasible assortments. This formulation is flexible enough to represent a wide range of linear constraints.

With these definitions, the constrained assortment optimization problem under the GNL model can be formulated as
\begin{align}
    \max_{\bx \in \cX} \quad 
    &F(\bx) = \frac{\sum_{i \in [m]} \sum_{n \in [N]} W_n^{\sigma_n - 1} \, (\alpha_{in} x_i r_i V_{in})}{\sum_{n \in [N]} W_{n}^{\sigma_{n}}} 
    \label{prob:GNL-assort}\tag{\sf Assort-1}\\
    \text{s.t.} \quad 
    & W_n = V_{0n} + \sum_{i \in [m]} \alpha_{in} x_i V_{in}, \quad \forall n \in [N]. \nonumber
\end{align}
The objective function $F(\bx)$ is highly nonlinear and, in general, not concave in $\bx$. This significantly complicates the design of exact algorithms: classical approaches for mixed-integer nonlinear programs (MINLP), such as mixed-integer linear programming (MILP) reformulations, conic optimization techniques~\citep{SenAtamturkKaminsky2018}, or outer-approximation methods~\citep{duran1986outer,PhamMai2025MixedLogitOA}, are not directly applicable. To the best of our knowledge, apart from the study by \citet{Cuong2024}, no existing method is capable of solving problem~\eqref{prob:GNL-assort} to global or near-global optimality under such general linear constraints.
This motivates the development of new formulations and exact algorithms, which we introduce in the following sections.

\section{Solution  Methods for Assortment Optimization under GNL Model}\label{sec:convexification_gnl}
We now present exact solution methods for the assortment optimization problem introduced in the previous section. The key idea is to reformulate the original non-convex problem into one (or a sequence) of mixed-integer convex programs (MICP), i.e., optimization problems with convex objectives and constraints, for which exact solution techniques can be applied. To enable this reformulation, we first recast the maximization problem into an equivalent minimization problem, which allows us to identify and exploit convex components of the objective function.

\begin{proposition}[Minimization reformulation]
Let $\beta>\max_{i\in[m]} \{r_i\}$ and define $r'_i:=\beta-r_i>0$ for all $i\in[m]$,
%Assume that for every feasible $\bx\in\cX$ the denominator $\sum_{n\in[N]} W_n^{\sigma_n}$ is strictly positive.
then problem~\eqref{prob:GNL-assort} is equivalent (optimization-wise) to the minimization problem
\begin{align}
\min_{\bx\in\cX}\quad 
&\frac{\sum_{n\in[N]} W_n^{\sigma_n-1}\Big(\beta V_{0n}+\sum_{i\in[m]}\alpha_{in}x_i r'_i V_{in}\Big)}
{\sum_{n\in[N]} W_n^{\sigma_n}}
\tag{\sf Assort-min}\label{prob:GNL-assort-min}\\
\text{s.t.}\quad 
& W_n= V_{0n}+\sum_{i\in[m]}\alpha_{in}x_i V_{in},\qquad \forall n\in[N]. \nonumber
\end{align}
\end{proposition}

\subsection{Bisection}
% The minimization form of the assortment optimization problem under the GNL model can be written as
% \begin{align}
% \min_{\bx\in\cX}\quad 
% &\frac{\sum_{n\in[N]} W_n^{\sigma_n-1}\Big(\beta V_{0n}+\sum_{i\in[m]}\alpha_{in}x_i r'_i V_{in}\Big)}
% {\sum_{n\in[N]} W_n^{\sigma_n}}
% \tag{\sf Assort-min}\label{prob:GNL-assort-min}\\
% \text{s.t.}\quad 
% & W_n= V_{0n}+\sum_{i\in[m]}\alpha_{in}x_i V_{in},\qquad \forall n\in[N]. \nonumber
% \end{align}

The objective of~\eqref{prob:GNL-assort-min} is a ratio of two nonnegative functions. A common technique in fractional programming is to convert such problems into a parametric feasibility form using a bisection search (or Dinkelbach-type approach) \citep{dinkelbach1967nonlinear}. The idea is to guess a candidate value $\delta$ for the optimal ratio, and then check whether there exists a feasible $\bx$ such that the numerator is at most $\delta$ times the denominator. By repeatedly updating $\delta$ through binary search, we can converge to the optimal value within arbitrary precision. This approach reduces the original nonlinear optimization to a sequence of feasibility checks (or equivalent subproblems), which are often easier to tackle.

Following this principle, problem~\eqref{prob:GNL-assort-min} can be equivalently reformulated as
\[
\min \Big\{\delta > 0 \;\Big|\; \exists \bx \in \cX \text{ such that } 
\sum_{i \in [m]} \sum_{n \in [N]} W_n^{\sigma_n - 1} \big(\alpha_{in} x_i r_i V_{in}\big) 
\leq \delta \sum_{n \in [N]} W_n^{\sigma_n} \Big\}.
\]
This formulation suggests that the problem can be viewed as a search over the scalar variable $\delta$, 
aiming to find the smallest value of $\delta$ for which there exists a feasible assortment vector $\bx \in \cX$ 
satisfying the inequality constraint. In other words, $\delta$ represents an upper bound on the normalized 
weighted utility contributions of all products and nests. 

Since $\delta$ is a one-dimensional scalar variable, the search over $\delta$ can be efficiently performed 
using a simple \emph{bisection procedure}. To operationalize this idea, we define for any given $\delta$ 
the auxiliary (or subproblem) function:
\[
G(\delta, \bx) 
= \sum_{n \in [N]} W_n^{\sigma_n - 1} 
\Big(\beta V_{0n} + \sum_{i \in [m]} \alpha_{in} x_i r'_i V_{in}\Big) 
- \delta \sum_{n \in [N]} W_n^{\sigma_n}.
\]
Then, the original problem can be equivalently expressed as
\[
\min \big\{\delta > 0 \;\big|\; \min_{\bx \in \cX} G(\delta, \bx) \leq 0 \big\}.
\]
This reformulation reveals that the optimization can be solved iteratively via binary search on $\delta$: 
at each iteration, we solve the subproblem $\min_{\bx \in \cX} G(\delta, \bx)$ to test feasibility. 
If $\min_{\bx \in \cX} G(\delta, \bx) \leq 0$, then the current $\delta$ is sufficiently large 
(i.e., it constitutes a valid upper bound for the optimal value). 
Otherwise, $\delta$ is too small and must be increased. 
This monotonic property of $G(\delta, \bx)$ with respect to $\delta$ guarantees the correctness 
and convergence of the bisection search.

\paragraph{Convexity and submodularity.}
We now discuss how to solve the subproblem efficiently. The subproblem objective involves both exponential and bilinear terms, rendering it highly nonconvex in general. Interestingly, as we show below, certain components of this objective function exhibit useful convexity and concavity properties, which allow for an exact solution via a cutting-plane approach.

For notational convenience, let us define:
$
H_n(\bx) = W_n^{\sigma_n - 1}, 
$ and $
K_n(\bx) = W_n^{\sigma_n},
$
where $W_n(\bx) = V_{0n} + \sum_{i \in [m]} \alpha_{in} x_i V_{in}$ represents the inclusive value associated with nest $n$. With these definitions, the objective of the subproblem can be written compactly as:
\[
G(\bx) 
= \sum_{n \in [N]} H_n(\bx)
\left(\beta V_{0n} + \sum_{i \in [m]} \alpha_{in} x_i r'_i V_{in}\right) 
- \delta \sum_{n \in [N]} K_n(\bx).
\]
By introducing the auxiliary variables $h_n$ and $k_n$ to separately represent the nonlinear components $H_n(\bx)$ and $K_n(\bx)$, the subproblem can be reformulated as the following constrained optimization problem:
\begin{align}
    \min_{\bx,\,\bh,\,\bk}\quad &
    \sum_{n \in [N]} 
    h_n\!\left(\beta V_{0n} + \sum_{i \in [m]} \alpha_{in} x_i r'_i V_{in}\right)
    - \delta \sum_{n \in [N]} k_n, \nonumber\\
    \text{s.t.}\quad 
    & h_n \ge H_n(\bx), \qquad n \in [N], \label{con:h}\\
    & k_n \le K_n(\bx), \qquad n \in [N], \label{con:k}\\
    & \bx \in \mathcal{X} \nonumber.
\end{align}
This reformulation is valid because both coefficients 
$\left(\beta V_{0n} + \sum_{i \in [m]} \alpha_{in} x_i r'_i V_{in}\right)$ 
and $\delta$ are strictly positive for all $n$. 
Hence, the inequalities in \eqref{con:h}–\eqref{con:k} preserve the optimal solution: 
decreasing $h_n$ or increasing $k_n$ within their feasible bounds 
always decreases the objective, ensuring that the constraints bind at optimality, 
i.e., $h_n = H_n(\bx)$ and $k_n = K_n(\bx)$ for all $n$.

An important observation here is that the functions $H_n(\bx)$ and $K_n(\bx)$ possess opposite curvature properties that can be exploited in our solution approach. Specifically, $H_n(\bx)$ is convex while $K_n(\bx)$ is concave when the nest parameters satisfy $\sigma_n \in (0,1]$. We formalize this in the following proposition.
\begin{proposition}\label{pro:convex_H_K}
For $\sigma_n \in (0,1]$, the function $H_n(\bx) = W_n(\bx)^{\sigma_n - 1}$ is convex in $\bx$, whereas $K_n(\bx) = W_n(\bx)^{\sigma_n}$ is concave in $\bx$.
\end{proposition}

The convexity of $H_n(\bx)$ and the concavity of $K_n(\bx)$ have important algorithmic implications. 
As shown below, these curvature properties enable us to approximate both functions using 
supporting hyperplanes (subgradient cuts) without eliminating any feasible solutions. 
By iteratively adding such cuts, we obtain a sequence of increasingly tighter outer approximations 
of the feasible region. This leads naturally to a cutting-plane or B\&C procedure 
that converges to the global optimal solution of the subproblem \citep{duran1986outer,Ljubic2018outer}. 
In each iteration, the convex components are linearized through subgradient inequalities, 
while the concave parts are conservatively approximated from above, ensuring that 
no feasible point is ever excluded. The resulting framework combines theoretical soundness 
with computational tractability, as further discussed in the next section.

It is also noteworthy that these curvature properties hold only when $\sigma_n \leq 1$.  
This restriction aligns with the theoretical requirements of the GNL and NL models, ensuring consistency with the random utility maximization (RUM) framework.  
For $\sigma_n > 1$, the model violates the RUM regularity conditions and the convexity of $H_n(\bx)$ no longer holds, making the optimization problem substantially more challenging.

% Another important observation is that the two functions $H_n(\bx)$ and $K_n(\bx)$, 
% when interpreted as set functions, also exhibit submodularity-related properties. 
% This structural property plays a crucial role in strengthening the cutting-plane algorithm: 
% it enables the incorporation of \emph{submodular cuts} in addition to conventional subgradient cuts, 
% which further tighten the outer-approximation of the feasible region and improve convergence.

Specifically, consider $H_n(\cdot)$ and $K_n(\cdot)$ as set functions defined on subsets of the product set $[m]$.  
For any subset $S \subseteq [m]$, define
\[
H_n(S) = \left(V_{0n} + \sum_{i \in S} \alpha_{in} V_{in}\right)^{\sigma_n - 1}, 
\qquad 
K_n(S) = \left(V_{0n} + \sum_{i \in S} \alpha_{in} V_{in}\right)^{\sigma_n}.
\]
Intuitively, $H_n(S)$ and $K_n(S)$ represent, respectively, the contributions of the chosen items 
in set $S$ to the nest-level utilities raised to powers $\sigma_n - 1$ and $\sigma_n$.  
The following proposition establishes their monotonicity and curvature properties.

\begin{proposition}
For $\sigma_n \in (0,1]$, the set function $H_n(S)$ is monotonically decreasing and supermodular in $S$, 
while $K_n(S)$ is monotonically increasing and submodular in $S$.
\end{proposition}

The submodularity and supermodularity of $K_n(S)$ and $H_n(S)$, respectively, 
offer an additional layer of structure that can be leveraged computationally.  
In particular, submodular functions admit tight linear underestimators known as \emph{submodular cuts}.  
By combining such cuts with the subgradient-based convex approximations derived earlier, 
we can construct a stronger outer-approximation of the feasible region.  
This hybrid approach—integrating convexity-based and submodularity-based cuts—enhances both the 
tightness of the relaxation and the overall efficiency of the cutting-plane algorithm.

\paragraph{A tractable bisection subproblem and convergence.}
Putting all components together, we can now express the subproblem used at each step of the bisection procedure in a tractable MILP form. 
Specifically, given a fixed value of~$\delta$, each iteration of the binary search requires solving the following problem:
\begin{align}
\min_{\bx,\bs,\bh,\bk,\bW}\quad 
& \sum_{n \in [N]} \beta V_{0n} h_n 
  + \sum_{n \in [N]}\sum_{i \in [m]} \alpha_{in} r'_i V_{in} s_{ni} 
  - \delta \sum_{n \in [N]} k_n 
  \tag{\sf GNL-Bis}\label{prob:Bisection}\\
\text{s.t.}\quad 
& W_n = V_{0n} + \sum_{i\in[m]} \alpha_{in}x_i V_{in}, 
  && \forall n \in [N], \nonumber\\[2pt]
& s_{ni} = h_n x_i, 
  && \forall n \in [N],\, i \in [m], \label{const:hn-xi}\\[2pt]
&  h_n \ge {H}_n(\bx), 
  && \forall n \in [N], \label{const:OA-hn}\\[2pt]
&  k_n \le {K}_n(\bx),
  && \forall n \in [N], \label{const:OA-kn}\\[2pt]
& A\bx + B \le C, \nonumber\\
& \bx \in \{0,1\}^m,\;
  \bs \in \mathbb{R}^{N\times m},\;
  \bW, \bh, \bk \in \mathbb{R}_+^N. \nonumber
\end{align}
Problem~\eqref{prob:Bisection} can be solved by iteratively generating \emph{outer-approximation (OA) cuts} \citep{duran1986outer} and \emph{submodular cuts (SCs)} \citep{Nemhauser1978analysis} to approximate the nonlinear constraints~\eqref{const:OA-hn} and~\eqref{const:OA-kn}.
 In particular, the OA constraints are constructed from subgradients of the convex and concave envelopes 
of $H_n$ and $K_n$, while the SC constraints arise from 
the  submodular inequalities 
that capture the submodular or supermodular structure of these functions. 
Together, these linear inequalities progressively tighten the relaxation of the nonlinear subproblem. Here we note that
the bilinear relations $s_{ni}=h_nx_i$ in~\eqref{const:hn-xi} are linearized 
exactly using the \emph{McCormick inequalities} \citep{mccormick1976computability}:
\[
\begin{aligned}
    s_{ni} &\ge h_n^{\mathrm{L}} x_i, &
    s_{ni} &\le h_n^{\mathrm{U}} x_i,\\
    s_{ni} &\ge h_n - (1-x_i)h_n^{\mathrm{U}}, &
    s_{ni} &\le h_n - (1-x_i)h_n^{\mathrm{L}},
\end{aligned}
\]
where $h_n^{\mathrm{L}}$ and $h_n^{\mathrm{U}}$ denote known lower and upper bounds on~$h_n$.  
Since these inequalities are tight for binary~$x_i$, the linearization is exact.

Problem~\eqref{prob:Bisection} can therefore be solved efficiently using a 
B\&C procedure (for each fixed $\delta$). A master problem is solved in which the nonlinear constraints 
\eqref{const:OA-hn}–\eqref{const:OA-kn} 
are replaced by their current linear OA cuts and SCs. 
If the candidate solution obtained is feasible for the original nonlinear problem 
(up to a tolerance~$\epsilon$), the algorithm terminates. 
Otherwise, violated OA or SC constraints are identified and new corresponding OA cuts and SCs added to the master problem 
to refine the feasible region. 
Because the feasible set of the OA/SC relaxation is convex and the cuts are valid globally, 
this iterative refinement process monotonically improves the approximation of the true feasible region.
Convergence of the cutting-plane scheme is therefore guaranteed in a finite number of iterations, 
and the overall bisection procedure converges to the optimal value of~$\delta^\star$ 
within a prescribed tolerance.

It is important to note that problem~\eqref{prob:Bisection} can be solved directly without resorting to an explicit bisection procedure by formulating it as a bilinear program. 
Specifically, when using a solver that supports bilinear optimization (such as \texttt{GUROBI}), the variable $\delta$ can be treated as a decision variable and problem~\eqref{prob:Bisection} can be solved directly via a B\&C framework. 
At each iteration, the nonlinear constraints~\eqref{const:OA-hn}--\eqref{const:OA-kn} are approximated by OA cuts and SCs, while the resulting master problem—formulated as a bilinear program—is solved by the solver. 
We refer to this approach as a \emph{bilinear--convex reformulation}.

\subsection{Convexification via Log-Transformation Reformulation}

The bisection scheme described above enjoys rapid convergence due to the exponential convergence rate of binary search and is highly effective when handling a single fractional structure. 
However, it cannot be directly extended to problems involving multiple coupled fractional components or more general nonlinear constraints. 
To overcome this limitation, we show below that the structure of the GNL model admits a reformulation that transforms the assortment optimization problem into a \emph{convex} MINLP. 
This convex reformulation can then be solved directly using standard B\&C techniques, without the need for iterative bisection.

We begin by rewriting the objective as
\[
f(\bx) 
= 
\frac{\sum_{n\in [N]} W_n^{\sigma_n-1} \beta V_{0n}}
      {\sum_{n\in [N]}W_n^{\sigma_n}}
+
\frac{\sum_{i\in [m]}x_i\!\left(\sum_{n\in [N]} W_n^{\sigma_n-1}\alpha_{in}r'_iV_{in}\right)}
     {\sum_{n\in [N]}W_n^{\sigma_n}},
\]
where $W_n = V_{0n} + \sum_{i\in[m]} \alpha_{in} V_{in} x_i$.
The nonconvexity of $f(\bx)$ arises from the ratio of sums of nonlinear terms.
To convexify this structure, we introduce logarithmic auxiliary variables
\[
y_n = \log\!\big(W_n^{\sigma_n-1}\big)
\quad\text{and}\quad 
z = \log\!\left(\sum_{n\in [N]} W_n^{\sigma_n}\right).
\]
Substituting these definitions into the objective and replacing the exponential terms $e^{y_n - z}$ by new variables 
$t_n \ge e^{y_n - z}$, we obtain the following equivalent problem
\begin{align}
\min_{\bx,\,\by,\,\bz,\,\bt}\quad &
\sum_{n\in [N]} \beta V_{0n} t_n 
+ \sum_{n\in [N]}\sum_{i\in [m]} \alpha_{in} r'_i V_{in} x_i t_n 
\tag{\sf Assort-Convex}\label{prob:GNL-assort-convex}\\
\text{s.t.}\quad 
& W_n = V_{0n} + \sum_{i\in [m]} \alpha_{in} V_{in} x_i, 
\quad \forall n \in [N], \nonumber\\
& t_n \ge e^{y_n - z}, 
\quad \forall n \in [N], \nonumber\\
& y_n \ge (\sigma_n - 1)\log W_n, 
\quad \forall n \in [N], \nonumber\\
& z \le \log\!\Big(\sum_{n\in [N]} W_n^{\sigma_n}\Big), \nonumber\\
& \bx \in \{0,1\}^m,\; \by,\bt,\bW \in \mathbb{R}_+^N. \nonumber
\end{align}

\paragraph{Convexity and feasibility.}
We briefly verify the convexity of the constraints in Problem~\eqref{prob:GNL-assort-2} and the feasibility of the reformulation. 
The convexity follows directly from standard properties of exponential and logarithmic functions. The constraint $t_n \ge e^{y_n-z}$ is convex since the exponential function is convex and nondecreasing and its argument $y_n-z$ is affine. 
The constraint $y_n \ge (\sigma_n-1)\log W_n$ is also convex because $\log(W_n)$ is concave for $W_n>0$ and $(\sigma_n-1)\le 0$ when $\sigma_n\in(0,1]$. Finally, the constraint 
$z \le \log\!\Big(\sum_{n\in[N]} W_n^{\sigma_n}\Big)$
defines the hypograph of a concave function. Indeed, for $\sigma_n\in(0,1]$, each mapping $W_n\mapsto W_n^{\sigma_n}$ is concave and nondecreasing, their sum is concave, and composition with the concave, nondecreasing logarithm preserves concavity. Hence, this constraint defines a convex feasible set in $(\bW,z)$. 
Moreover, feasibility is guaranteed because $W_n=V_{0n}+\sum_i \alpha_{in}V_{in}x_i$ is strictly positive for all feasible~$\bx$, ensuring that all logarithmic expressions are well defined.

Combining these properties, the overall feasible region of~\eqref{prob:GNL-assort-2} is convex in the continuous variables $(\bx,\bW,\by,\bz,\bt)$ for any fixed integer vector $\bx\in\{0,1\}^m$. 
Hence, the problem is a MICP, and the relaxation obtained by allowing $\bx\in[0,1]^m$ is a globally convex program. Moreover, feasibility of the reformulated model is guaranteed by the domain of the logarithmic and exponential functions. 
Since each $W_n = V_{0n} + \sum_i \alpha_{in}V_{in} x_i$ represents the inclusive value of nest~$n$, it is strictly positive whenever $V_{0n}>0$ and $\alpha_{in},V_{in} \ge 0$. 
This ensures that $\log W_n$ and $\log(\sum_n W_n^{\sigma_n})$ are well-defined for all feasible $\bx$. 
Additionally, because the mapping $t_n \mapsto e^{y_n-z}$ is continuous and monotonically increasing in $y_n$ and decreasing in $z$, there always exists a feasible triple $(t_n,y_n,z)$ satisfying the exponential constraint for any admissible $\bW$.

The only remaining nonconvex components in problem~\eqref{prob:GNL-assort-2} are the bilinear terms 
$x_i t_n$ that appear in the objective function. 
These terms couple the binary decision variables~$x_i$ with the continuous auxiliary variables~$t_n$, and can be linearized exactly using the \emph{McCormick inequalities}. 

\paragraph{Submodularity.}
In line with the bisection-based analysis, we investigate submodularity properties of the nonlinear
components in the reformulation \eqref{prob:GNL-assort-2}. Establishing (super)submodularity for
appropriate set functions allows us to generate valid SCs that strengthen the
optimization.

\begin{proposition}\label{prop:submod-yn-z}
Let $S\subseteq [m]$ denote an offered set of products. For each nest $n\in[N]$, define
\[
W_n(S)\;=\;V_{0n}+\sum_{i\in S}\alpha_{in}V_{in} \quad\text{and}\quad
Y_n(S)\;=\;(\sigma_n-1)\log\!\big(W_n(S)\big).
\]
Moreover, define
$Z(S)\;=\;\log\!\Big(\sum_{n=1}^N W_n(S)^{\sigma_n}\Big).$
Then the following hold:
\begin{enumerate}
\item[(i)] $Y_n(S)$ is \emph{monotonically decreasing} in $S$ and \emph{supermodular}.
\item [(ii)] $Z(S)$ is \emph{monotonically increasing} in $S$ and \emph{submodular}.
\end{enumerate}
\end{proposition}

\paragraph{Cutting Plane method.}
Integrating the above results, define 
\[
Y_n(\bx) = (\sigma_n - 1)\log\!\big(W_n(\bx)\big),
\quad \text{where } W_n(\bx) = V_{0n} + \sum_{i\in [m]} \alpha_{in} V_{in} x_i,
\]
and 
\[
Z(\bx) = \log\!\Big(\sum_{n\in [N]} W_n(\bx)^{\sigma_n}\Big).
\]
As shown above, when viewed as set functions, Proposition~\ref{prop:submod-yn-z} establishes that $Y_n(S)$ is supermodular and $Z(S)$ is submodular. These structural properties enable the construction of valid \emph{OA cuts} and \emph{SCs} that progressively tighten the feasible region.

Consequently, problem~\eqref{prob:GNL-assort-convex} can be represented as a master problem and solved through a B\&C procedure, in which OA cuts and SCs are iteratively generated and incorporated on the fly until convergence to the global optimum, as follows:
\begin{align}
\min_{\bx \in \cX}\quad 
    &\left\{
        \sum_{n\in [N]} \beta V_{0n} t_n 
        + \sum_{n\in [N]} \sum_{i\in [m]} \alpha_{in} r'_i V_{in} s_{ni}
    \right\}
    \tag{\sf GNL-Convex}
    \label{prob:Expcone}\\[3pt]
\text{s.t.}\quad 
    & W_n = V_{0n} + \sum_{i\in [m]} \alpha_{in} V_{in} x_i, 
    && \forall n \in [N], \nonumber\\
    & s_{ni} = t_n x_i, 
    && \forall n \in [N],\, i \in [m], \nonumber\\[3pt]
    &  \text{\texttt{[OA]}:}\quad t_n \ge e^{\,y_n - z}, 
    && \forall n \in [N], \nonumber\\
    & \text{\texttt{[OA]}:}\quad y_n \ge (\sigma_n - 1)\log(W_n), 
    && \forall n \in [N], \nonumber\\
    & \text{\texttt{[SC]}:}\quad y_n \ge Y_n(\bx), 
    && \forall n \in [N], \nonumber\\
    &\text{\texttt{[OA]}:}\quad z \le \log(\sum_{n\in [N]} W_n^{\sigma_n}), \nonumber\\
    & \text{\texttt{[SC]}:}\quad z \le Z(\bx), \nonumber\\
    & \bA \bx + \bB \le \bC, \nonumber\\
    & \bx \in \{0,1\}^{[m]},\;
      \bs \in \mathbb{R}_+^{[N]\times[m]},\;
      \bW, \bt, \by \in \mathbb{R}_+^{[N]},\;
      z \in \mathbb{R}_+.\nonumber
\end{align}
where the terms \texttt{[OA]} and \texttt{[SC]} indicate that OA cuts or SCs, respectively, are used to approximate the nonlinear constraints.
The OA cuts provide locally valid linear approximations of convex or concave nonlinear functions, whereas the SCs exploit the (super)submodularity properties of the set functions $Y_n(S)$ and $Z(S)$ established in Proposition~\ref{prop:submod-yn-z}.  Given the convexity of the continuous relaxations and the submodularity structure of the underlying set functions, all OA cuts and SCs are valid and do not remove any feasible point from the convex hull of the original problem. Therefore, the B\&C procedure is guaranteed to converge to the global optimum after a finite number of iterations.

A final note is that the following inequality provides a valid strengthening of the convex reformulation and can be directly added to the optimization problem to further tighten the feasible region.
\begin{proposition}\label{prop:expcone-lowerbound}
   If $\sigma_n<1$ for all $n\in [N]$, then the following convex constraints are valid for the reformulation:
   $ z \;\geq\; 
    \log\!\left(\sum_{n\in [N]} e^{\frac{\sigma_n}{\sigma_n - 1}y_n}\right).$
\end{proposition}

In summary, assortment optimization under the GNL model can be solved to optimality using three exact approaches developed above: (i) a bisection method that exploits the fractional structure of the objective, (ii) a bilinear--convex reformulation,  and (iii) a convex reformulation based on logarithmic transformations. Each approach reduces the original problem to the solution of one or more MINLPs, which can be efficiently handled within a B\&C framework using OA cuts and/or SCs.

While the MGNL model constitutes a practically important extension that captures customer heterogeneity, not all of the above approaches extend naturally to this setting. In particular, the bisection method cannot be directly generalized due to the loss of the single-ratio structure in the objective. In contrast, both the convex reformulation and the bilinear--convex reformulation admit natural extensions to the MGNL model, enabling exact solution methods for assortment optimization under mixed GNL demand. A detailed discussion of these extensions  is provided in Appendix \ref{ec:mixed_gnl_model}.

% \subsection{An practical alterantive bilinear reformualtin for teh bisection subproblem}

% Currente SOTA mixed-itergert solver such as CPLEX or gurobi can support bilinear terms. Thus we can formualte the the following practical subproblem with bilinear  and convex terms as 
% 	\begin{align}
% 		\min\qquad &\delta\label{prob:Bilinear}\tag{\sf GNL-Bi} \\
%         \text{subject to} \quad & \sum_{n \in [N]} h_n(\beta V_{0n} + \sum_{i \in [m]} \alpha_{in} x_i r'_i V_{in}) = \delta (\sum_{n \in [N]} k_n)\nonumber\\
%         &W_n = V_{0n} + \sum_{i\in  [m]}\alpha_{in}x_iV_{in}, \forall{n \in [N]} \nonumber \\
%         & k_n = W_n h_n, \forall{n \in [N]} \nonumber \\
%         & \texttt{[Lazy cut]}\quad h_n \geq \texttt{[OA: $H_n(\bx)$]}, \forall{n \in [N]}\nonumber \\
%          & \texttt{[Lazy cut]}\quad h_n \geq \texttt{[SC: $H_n(S)$]}, \forall{n \in [N]}\nonumber \\
%         & \texttt{[Lazy cut]}\quad k_n \leq \texttt{[OA: $K_n(\bx)$]}, \forall{n \in [N]}\nonumber \\
%          & \texttt{[Lazy cut]}\quad k_n \leq \texttt{[SC: $K_n(S)$]}, \forall{n \in [N]}\nonumber \\
%         & \bA \bx + \bB \leq \bC \nonumber\\
%         &\bx \in \{0,1\}^{[m]}, \bW , \bh , \bk \in \mathbb{R}_+^{[N]}, \delta \in \mathbb{R}_+.\nonumber
% 	\end{align}
% Then a cutting plane or B\&C approaches can be used.

\section{Joint Assortment and Pricing Optimization}\label{sec:JAP}
While the previous section focuses on assortment optimization with fixed prices, many practical applications require JAP decisions to maximize revenue or profit. This problem is highly relevant in settings such as retailing, transportation, and online platforms, where pricing and availability interact to shape customer choice. Under the GNL framework, the JAP problem is particularly challenging because it combines the discrete, combinatorial nature of assortment decisions with the continuous and nonlinear dependence of choice probabilities on prices, resulting in a highly nonconvex optimization problem.

In this section, we focus on the JAP problem with continuous price decisions. The discrete-price variant—which can be reformulated as an assortment optimization problem by extending the product set—is deferred to Appendix~\ref{ec:JAP_DP}. Although discretizing prices offers computational convenience, it can be restrictive when feasible price ranges are large or when optimal prices lie between predefined grid points. Allowing prices to vary continuously over intervals $[L_i,U_i]$ provides a more flexible and realistic modeling framework. However, continuous pricing substantially increases computational complexity, as prices enter the choice probabilities nonlinearly through exponential utility terms, destroying the separability exploited in the discrete setting. In what follows, we present the continuous-price JAP formulation and show how the convexification and bilinear--convex techniques developed earlier can be extended to address this more challenging problem.

\subsection{Problem Formulation}
We now consider the JAP problem under the GNL framework when product prices are treated as continuous decision variables.  
Unlike the discrete-price case, where each product--price combination can be enumerated explicitly, the continuous-price setting allows each product $i \in [m]$ to take any price value within a continuous range $[L_i, U_i]$. In addition to the binary assortment variables $x_i \in \{0,1\}$, we introduce continuous decision variables $y_i \in [L_i, U_i]$ representing the selling prices of the products.  
Let $\mathcal{Y}$ denote the feasible set of price vectors.  
Besides the bound constraints, $\mathcal{Y}$ can incorporate additional linear pricing relations such as budget, ordering, or promotional constraints.  We define the deterministic utility of product~$i$ as
$v_i(y_i) = -\eta_i y_i + \kappa_i,$
where $\eta_i > 0$ captures the price sensitivity of product $i$, and $\kappa_i$ represents the intrinsic attractiveness of the product.  
Under the GNL model, the inclusive value of nest $n \in [N]$ is given by
\[
W_n(\bx, \by) 
= V_{0n} + \sum_{i \in [m]} \alpha_{in} x_i 
\exp\!\left(\frac{v_i(y_i)}{\sigma_n}\right)
= V_{0n} + \sum_{i \in [m]} \alpha_{in} x_i 
\exp\!\left(\frac{-\eta_i y_i + \kappa_i}{\sigma_n}\right),
\]
% where $\alpha_{in}$ is the allocation parameter linking product~$i$ to nest~$n$, $\sigma_n \in (0,1]$ is the within-nest correlation parameter, and $V_{0n}$ denotes the utility of the outside option associated with nest~$n$. Now, given an assortment $\bx$ and price vector $\by$, the expected revenue of the firm is then
% \[
% F(\bx, \by)
% = 
% \frac{
% \sum_{i\in[m]} \sum_{n\in[N]} 
% \alpha_{in} x_i y_i 
% \exp\!\left(\frac{v_i(y_i)}{\sigma_n}\right)
% W_n(\bx, \by)^{\sigma_n - 1}
% }{
% \sum_{n\in[N]} W_n(\bx, \by)^{\sigma_n}
% }.
% \]
The continuous-price JAP problem can be formulated as
\begin{align}
\max_{\bx \in \mathcal{X},\, \by \in \mathcal{Y}} \quad &
F(\bx, \by)
=
\frac{
\sum_{i\in[m]} \sum_{n\in[N]} 
\alpha_{in} x_i y_i 
\exp\!\left(\frac{-\eta_i y_i + \kappa_i}{\sigma_n}\right)
W_n^{\sigma_n - 1}
}{
\sum_{n\in[N]} W_n^{\sigma_n}
}
\label{prob:JAP-Continuous}\\
\text{s.t.} \quad &
W_n = V_{0n} + \sum_{i \in [m]} \alpha_{in} x_i 
\exp\!\left(\frac{-\eta_i y_i + \kappa_i}{\sigma_n}\right),
\quad \forall n \in [N], \nonumber\\
& x_i \in \{0,1\}, \quad \forall i \in [m], \nonumber\\
& y_i \in [L_i, U_i], \quad \forall i \in [m]. \nonumber
\end{align}

% This formulation generalizes the discrete-price JAP model by allowing the prices $y_i$ to vary continuously over compact intervals.  
% The objective in~\eqref{prob:JAP-Continuous} involves both binary assortment decisions and continuous pricing variables in a highly nonlinear and nonconvex manner, due to the multiplicative coupling between $x_i$ and $\exp(-\eta_i y_i)$, as well as the ratio structure of the GNL model.  
% Consequently, the continuous-price formulation is substantially more challenging than its discrete counterpart: the objective function is nonconvex in $(\bx, \by)$, the denominator introduces cross-nest dependencies, and the joint presence of discrete and continuous variables precludes direct application of standard convex optimization methods. 

Before presenting our solution approach for the above JAP problem, 
we first provide a simple example illustrating that the pricing problem under the GNL model 
(with a fixed assortment) can be \emph{non-unimodal}. 
In particular, even with only \textit{one product} and two nests, 
the revenue function may exhibit multiple local maxima. 
This implies that standard gradient-based optimization methods 
are prone to convergence to suboptimal local solutions.

\begin{example}[Non-unimodality with One Product and Two Nests]
We illustrate the non-unimodal behavior of the expected revenue function 
in a setting with a single product offered through two nests that differ significantly 
in their price sensitivities and allocation weights.
The model parameters are chosen as:
$V_{01}=V_{02}=1.0,\quad 
\alpha_{11}=0.1,\; \alpha_{12}=5.0,\quad 
\sigma_1=0.1,\; \sigma_2=0.9,\quad 
\eta=1.0,\; \kappa=2.0.$
The deterministic utility of the product is $v(y)=-\eta y+\kappa$, and the expected revenue as a function of price is
$F(y)
=
y
\frac{
\sum_{n=1}^2
\alpha_{1n}
e^{v(y)/\sigma_n}
W_n(y)^{\sigma_n-1}
}{
\sum_{n=1}^2
W_n(y)^{\sigma_n}
}.$
Figure~\ref{fig:nonunimodal_gnl} shows the resulting revenue curve.
Two local maxima can be clearly observed:
the first (left peak) corresponds to the low-correlation nest~1, which is highly price-sensitive,
while the second (right peak) corresponds to the high-correlation nest~2.
This simple example confirms that even a single-product GNL structure with heterogeneous nests
can yield a nonconvex, non-unimodal pricing landscape.
\begin{figure}[!h]
    \centering \includegraphics[width=0.5\textwidth]{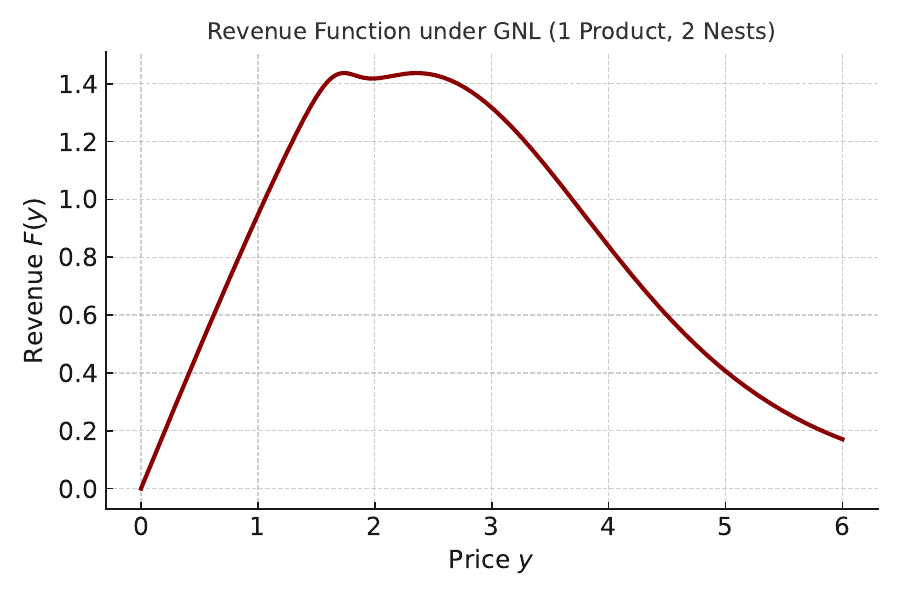}
    \caption{Example of the non-unimodal revenue function $F(y)$ under the GNL model with fixed assortment.}
    \label{fig:nonunimodal_gnl}
\end{figure}

We note that when there is only one product and one nest, the revenue function is \emph{unimodal}.
To see this, observe that in this case the GNL model reduces to the standard logit formulation.
The derivative $F'(y)$ changes sign exactly once, implying that $F(y)$ is strictly increasing
up to a unique stationary point and strictly decreasing thereafter.
Hence, the expected revenue function is \emph{unimodal}, attaining a single global maximum at an interior price level.
\end{example}

\subsection{Near-optimal Solution via PWLA}
To address the nonconvexity of the continuous-price JAP problem, we propose a practical approximation approach based on \emph{PWLA} of the exponential terms that depend on prices.  
Specifically, the nonlinear function $\exp((-y_i \eta_i + \kappa_i)/\sigma_n)$ appearing in both the numerator and denominator of the objective can be approximated by a set of linear segments over the bounded price range $[L_i, U_i]$.  
By introducing auxiliary variables to represent these linear segments and applying standard linearization or convexification techniques, the original nonconvex program can be transformed into an equivalent MICP (or bilinear-convex) formulation.  
This allows us to directly leverage the solution techniques developed in the previous sections—such as the log-transformation and bilinear-convex reformulations—within an approximate yet tractable optimization framework.

%The main advantage of the PWLA approach is that the approximation accuracy can be controlled to an arbitrarily small level by adjusting the number of breakpoints.  Hence, the resulting approximation can achieve any desired balance between computational tractability and solution precision.  In the following, we first present the construction of the PWLA model for the exponential term and then analyze the approximation error it introduces.  We subsequently show that the overall JAP formulation with the approximated exponential functions admits a convex or bilinear-convex reformulations that can be solved efficiently to optimality.

To facilitate convexification and enable the application of the transformation techniques 
developed in Section~\ref{sec:convexification_gnl}, 
we first reformulate the maximization problem into an equivalent minimization problem. 
The key idea is to introduce a constant $\beta > \max_i\{U_i\}$, 
This transformation preserves optimality but reverses the objective’s direction, 
allowing us to exploit convexity and submodularity properties in subsequent reformulations. 
Specifically, for each product $i$, we define a transformed revenue term 
$r_i' = \beta - y_i > 0$. 
By substituting $y_i = \beta - r_i'$ into the original fractional objective, 
we obtain an equivalent minimization form:
\begin{align}
    \min_{\bx \in \cX, \by \in \cY} \quad 
    &\left\{F(\bx,\by) = 
    \frac{
    \sum_{n \in [N]} 
    W_n^{\sigma_n - 1}
    \left(
        \beta V_{0n}
        + 
        \sum_{i \in [m]} 
        \alpha_{in} x_i (\beta - y_i)
        \exp(y_i \eta_i + \kappa_i)
    \right)
    }{
    \sum_{n \in [N]} W_n^{\sigma_n}
    }
    \right\}.\nonumber
\end{align}
The reformulated objective remains fractional but is now expressed in a minimization form 
that aligns with our convexification framework. 
Intuitively, maximizing revenue is equivalent to minimizing the ``revenue gap'' with respect to 
the upper bound $\beta$, which provides a convenient structure for constructing 
convex reformulation. 

To reformulate the continuous-price JAP problem into a tractable mathematical program, 
we introduce auxiliary variables that separate nonlinear and fractional components 
while preserving equivalence to the original objective.
The resulting formulation explicitly exposes bilinear and convex–concave terms 
that can either be exactly linearized (when involving binary variables)
or handled directly by modern solvers that support bilinear programming.
The problem can thus be expressed as the following MINLP:
\begin{align}
    \min \quad & \delta 
    \label{prob:CP-Bilinear}\tag{\sf CP-GNL-Bi}\\
    \text{s.t.} \quad 
    & \sum_{n \in [N]} 
      h_n\!\left(\beta V_{0n} + \sum_{i \in [m]}\alpha_{in} s_i\right)
      \leq \delta \left(\sum_{n \in [N]} k_n\right),
      && \forall n \in [N], \label{ctr:ratio}\\[4pt]
    & s_i = x_i u_i, 
      && \forall i \in [m], \label{ctr:si}\\[4pt]
    & u_i = (\beta - y_i)t_i, 
      && \forall i \in [m], \label{ctr:ui}\\[4pt]
    & w_{in} = (y_i\eta_i+\kappa_i)/\sigma_n, 
      && \forall i \in [m], \label{ctr:wi}\\[4pt]  
    & t_{in} = \exp(w_{in}), 
      && \forall i \in [m], \label{ctr:ti}\\[4pt]
    & W_n = V_{0n} + \sum_{i \in [m]} \alpha_{in} x_i t_i, 
      && \forall n \in [N], \label{ctr:Wn}\\[4pt]
    & h_n \geq H_n(\bW),
      && \forall n \in [N], \label{ctr:hn}\\[4pt]
    & k_n \leq K_n(\bW),
      && \forall n \in [N], \label{ctr:kn}\\[4pt]
    & k_n = W_n h_n, 
      && \forall n \in [N], \label{ctr:bilinear}\\[4pt]
    & \bx\in \cX,\by\in \cY,
      \by, \bt, \bu, \bs \in \mathbb{R}_+^{[m]}, \;
      \bW, \bh, \bk \in \mathbb{R}_+^{[N]}, \;
      \delta \in \mathbb{R}_+. \nonumber
\end{align}
where $H_n(\bW) = W_n^{\sigma_n - 1}$ and $K_n(\bW) = W_n^{\sigma_n}$.
In the above formulation, 
$t_{in}$ represents the exponential utility term 
$\exp\!\left((y_i\eta_i + \kappa_i)/\sigma_n\right)$,
while $u_i = (\beta - y_i)t_i$ captures the interaction between the transformed price 
and the exponential component of the product utility.  
Constraint~\eqref{ctr:si} introduces the bilinear term $s_i = x_i u_i$, 
which can be \emph{exactly linearized} using standard McCormick envelopes 
since one of the variables ($x_i$) is binary.  

The auxiliary variables $(h_n, k_n)$ represent the nonlinear nest-level functions 
$H_n(\bW) = W_n^{\sigma_n - 1}$ and $K_n(\bW) = W_n^{\sigma_n}$, respectively.
Hence, constraints~\eqref{ctr:hn} and~\eqref{ctr:kn} 
serve as convex relaxations that safely replace the equalities 
    $h_n = H_n(\bW)$ and $k_n = K_n(\bW)$ while maintaining validity.  It can be observed that $H_n(\bW)$ is convex and $K_n(\bW)$ is concave in $\bW$. 
Therefore, constraints~\eqref{ctr:hn} and~\eqref{ctr:kn} can be safely approximated 
by valid OA cuts within a B\&C procedure, 
analogous to the approach discussed earlier for the assortment optimization problem. The bilinear equality~\eqref{ctr:bilinear}, $k_n = W_n h_n$, 
links these two nonlinear components.  
Although redundant from a purely representational standpoint, 
this constraint strengthens the relaxation during the B\&C procedure, 
as it allows additional linear cuts to be generated 
to better approximate the nonlinear relationships among variables.   Finally, constraint~\eqref{ctr:ratio} enforces the fractional objective structure, 
where $\delta$ acts as an upper bound on the normalized ratio of the transformed objective.  
Minimizing $\delta$ is therefore equivalent to minimizing the fractional objective 
of the original JAP problem.

The reformulation in~\eqref{prob:CP-Bilinear} consists primarily of convex or bilinear constraints, 
except for the nonlinear exponential relations $t_{in} = \exp(w_{in})$. To handle these constraints, we will approximate them using a PWLA scheme.  
The main idea is to discretize the function domain 
$w_{in} \in [(-L_i\eta_i + \kappa_i)/\sigma_n,\; (-U_i\eta_i + \kappa_i)/\sigma_n]$
into a finite set of breakpoints and construct a piecewise linear function 
that accurately approximates the exponential term $\exp(w_{in})$.  
This allows the nonlinear constraint $t_{in} = \exp(w_{in})$ 
to be replaced by a set of linear inequalities augmented with additional binary variables, 
thereby converting the nonlinear component into a MILP form.  
In the following, we discuss an efficient strategy to construct the PWLA of 
$t_{in} = \exp(w_{in})$ that achieves a desired approximation accuracy 
with a minimal number of breakpoints.

\paragraph{Constructing PWLA for  $t_{in} = \exp(w_{in})$: }
We now describe how to construct a PWLA 
for the exponential function: $t_{in} = \exp(w_{in})$, 
where $w_{in} = (-y_i\eta_i + \kappa_i)/\sigma_n$
and $y_i \in [L_i, U_i]$. For notational convenience, we denote
the PWLA of $\exp(w_{in})$ by: 
$\cPA(\exp(w_{in}) \,|\, \underline{w}_{in}, \overline{w}_{in}, \bq_{in})$ over the interval $[\underline{w}_{in}, \overline{w}_{in}]$, 
where $\bq_{in}$ is a vector of breakpoints used to construct the approximation:
$\underline{w}_{in} = ({-U_i\eta_i + \kappa_i})/{\sigma_n},$
 and 
$\overline{w}_{in} = ({-L_i\eta_i + \kappa_i})/{\sigma_n}$.

To construct a PWLA  for $\exp(w_{in})$, we partition the interval 
$[\underline{w}_{in}, \overline{w}_{in}]$ 
into $H+1$ subintervals using breakpoints
$\bq_{in} = (q_1, q_2, \dots, q_{H+1})$, \footnote{Each component of $\bq_{in}$ and the number of segments $H$ depend on the indices $i$ and $n$, but these subscripts are omitted here for notational simplicity.} such that 
$\underline{w}_{in} = q_1 < q_2 < \dots < q_{H+1} = \overline{w}_{in}$.  
By introducing auxiliary continuous variables $\nu_h \in [0,1]$ 
and binary variables $v_h \in \{0,1\}$ for all $h \in [H]$, 
we can represent a PWLA of $\exp(w_{in})$ as:
\[
\begin{cases}
    \cPA(\exp(w_{in}) \,|\, \underline{w}_{in}, \overline{w}_{in}, \bq_{in})
    = e^{q_1} + \sum_{h \in [H]} \delta_h (q_{h+1} - q_h) v_h,\\[4pt]
    w_{in} = q_1 + \sum_{h \in [H]} (q_{h+1} - q_h) v_h,\\[4pt]
    v_h \ge v_{h+1}, 
    \nu_h \ge v_h, 
    \nu_{h+1} \le v_h, & \forall h \in [H],
\end{cases}
\]
where: 
$
\delta_h = 
\frac{e^{q_{h+1}} - e^{q_h}}{q_{h+1} - q_h},
\forall h \in [H],
$ denotes the slope of the linear segment approximating $\exp(w_{in})$ 
in the interval $[q_h, q_{h+1}]$.

To achieve a desired approximation accuracy, i.e. $ |\cPA(\exp(w_{in}) \,|\, \underline{w}_{in}, \overline{w}_{in}, \bq_{in})-\exp(w_{in})| \leq \epsilon$ for any feasible $w_{in}$, the quality of the PLWA depends critically on the placement of the breakpoints that partition the interval of interest. Rather than using uniformly spaced breakpoints—which may result in large local approximation errors in regions where the exponential function exhibits high curvature—we adopt an adaptive construction procedure that allocates more breakpoints in highly nonlinear regions and fewer where the function is nearly linear. Specifically, starting from an initial breakpoint $q_h$, we iteratively determine the next breakpoint $q_{h+1}$ so that the maximum deviation between the exponential function $e^{w_{in}}$ and its linear approximation over the interval $[q_h, q_{h+1}]$ does not exceed a prescribed tolerance~$\epsilon$. This adaptive refinement ensures that the approximation error is uniformly controlled across the entire domain while achieving an efficient representation with a minimal number of segments. Moreover, the adaptive PWLA enjoys a favorable complexity property: the number of breakpoints required to attain approximation accuracy~$\epsilon$ grows on the order of $O(1/\sqrt{\epsilon})$. A detailed description of the breakpoint construction procedure and its theoretical justification is provided in Appendix~\ref{ec:PWLA}.

\paragraph{Bilinear–Convex Approximation and Performance Guarantee.}
Bringing all components together, 
the JAP problem with continuous prices can be approximated 
by the following bilinear–convex formulation:
\begin{align}
    \min \quad & \delta 
    \label{prob:CP-Bilinear}\tag{\sf JAP-CP-BiConv}\\
    \text{s.t.} \quad 
    &\texttt{Constraints \eqref{ctr:ratio}, \eqref{ctr:si}, \eqref{ctr:ui}, \eqref{ctr:wi}}, \nonumber\\[4pt]
    &\texttt{Constraints \eqref{ctr:Wn}, \eqref{ctr:hn}, \eqref{ctr:kn}, \eqref{ctr:bilinear}}, \nonumber\\[4pt]
    & t_{in} =  
      \cPA\!\left(\exp(w_{in}) \,\big|\, 
      \underline{w}_{in}, \overline{w}_{in}, \bq_{in}\right), 
      && \forall i \in [m],~ n \in [N], \label{ctr:ti-pwla}\\[4pt]
    & \bx \in \cX,\;
      \by \in \cY,\;
      \by, \bt, \bu, \bs \in \mathbb{R}_+^{[m]}, \;
      \bW, \bh, \bk \in \mathbb{R}_+^{[N]}, \;
      \delta \in \mathbb{R}_+. \nonumber
\end{align}

In this formulation, the nonlinear exponential constraints \eqref{ctr:ti}
are replaced by its PWLA \eqref{ctr:ti-pwla},
constructed using the adaptive breakpoint procedure described earlier. 
The resulting problem remains a MICP 
with bilinear components, which can be solved efficiently 
via a B\&C scheme incorporating OA cuts. 
As shown in Proposition~\ref{prop:breakpoint} (Appendix \ref{ec:PWLA}), the number of breakpoints required to achieve 
an approximation accuracy $\epsilon$ grows at most on the order of 
$\mathcal{O}(1/\!\sqrt{\epsilon})$, 
thus providing a controllable trade-off between model fidelity and computational cost.
Note that  this reformulation can be straightforwardly extended 
to the MGNL case discussed in Section~\ref{ec:mixed_gnl_model}.

In the approximate formulation~\eqref{prob:CP-Bilinear}, 
the exponential terms associated with prices, $\exp((-y_i\eta_i + \kappa_i)/\sigma_n)$, 
are replaced by their PWLA approximations. 
A natural question, therefore, concerns the performance guarantee of solving 
the approximate problem instead of the original JAP model. 
We analyze this approximation gap below. 
Specifically, the objective function of~\eqref{prob:CP-Bilinear} can be written as
\[
  \widehat{F}(\bx, \by)
  =
  \frac{
    \sum_{n \in [N]} 
    \widehat{W}_n^{\sigma_n - 1}
    \left(
        \beta V_{0n}
        + 
        \sum_{i \in [m]} 
        \alpha_{in} x_i (\beta - y_i)
        \widehat{t}_{in}
    \right)
  }{
    \sum_{n \in [N]} \widehat{W}_n^{\sigma_n}
  },
\]
where
\[
\widehat{t}_{in}
=
\cPA\!\left(\exp\!\left(\frac{-y_i \eta_i + \kappa_i}{\sigma_n}\right)\right),
\qquad
\widehat{W}_n
=
V_{0n}
+
\sum_{i \in [m]} 
\alpha_{in} x_i \widehat{t}_{in}.
\]
The PWLA $\cPA(\cdot)$ is constructed to guarantee the uniform approximation bound
\[
\big|
\cPA(\exp((-y_i \eta_i + \kappa_i)/\sigma_n))
- 
\exp((-y_i \eta_i + \kappa_i)/\sigma_n)
\big|
\le \epsilon,
\quad \forall i \in [m],~ n \in [N].
\]
We now establish an explicit bound on the resulting objective approximation error.
\begin{proposition}
Assume that $W_n(\bx,\by)\ge \underline W_n>0$ for all $n$, and define: 
$\underline W:=\min_n \underline W_n,
\sigma_{\min}:=\min_n\sigma_n.$
Then the approximation gap between the original and PWLA-based objectives satisfies
$|F(\bx,\by)-\widehat F(\bx,\by)|
\;\le\;
C\,\epsilon
\;=\;
\mathcal{O}(\epsilon),$
where $C$ is a constant depending only on model parameters
$(\alpha_{in},V_{0n},\beta,L_i,\sigma_n)$ and the lower bound $\underline W$,
but not on $\epsilon$.
\end{proposition}

The results above show that the PWLA-based bilinear–convex reformulation 
provides an $\mathcal{O}(\epsilon)$-accurate approximation of the original continuous-price JAP problem.  
That is, the approximation error in the objective function decreases linearly with the prescribed tolerance $\epsilon$.  
From Proposition~\ref{prop:breakpoint}, the number of breakpoints required to achieve accuracy $\epsilon$ grows at most on the order of $\mathcal{O}(1/\!\sqrt{\epsilon})$ per exponential term.  
Consequently, the overall formulation introduces at most 
$\mathcal{O}(mN/\!\sqrt{\epsilon})$ additional binary variables, 
where $m$ and $N$ denote the number of products and nests, respectively.  
In practice, however, due to the rapidly decaying curvature of the exponential function, 
a relatively small number of breakpoints already yields high accuracy.  
Hence, one can balance computational efficiency and precision by using only a modest number of binary variables 
while maintaining a negligible loss in optimality.

%%%%%%%%%%%%%%%%%%%%%%%%%%%%%%%%%%
%
%
%
%%%%%%%%%%%%%%%%%%%%%%%%%%%%%%%%%%
\section{Numerical Experiments}\label{sec:experiments}
This section presents a comprehensive numerical study to evaluate the computational performance and solution quality of the proposed frameworks for the assortment (and pricing) optimization problem under the GNL model. All experiments are conducted on synthetically generated instances that systematically vary in problem size, nest structure, and substitution patterns, allowing for a controlled assessment of scalability and robustness across diverse market environments. In all cases, the instances incorporate general linear constraints on the decision variables, reflecting realistic operational and capacity restrictions. %We benchmark our approach against several baseline solution strategies, including direct MINLP solvers and existing MILP–based approximation formulations \citep{Cuong2024} when applicable.

This section is organized as follows. We first describe the experimental setup, including instance generation, parameter settings, and implementation details. We then report the main comparative results between the proposed approaches and baseline methods for the assortment optimization problem under the GNL model, as well as for the JAP problem with continuous prices under GNL. Additional experiments—including results for assortment optimization under the MGNL model, JAP with discrete prices, and a comparison among the three proposed approaches for assortment optimization under GNL—are deferred to Appendix~\ref{ec:gnl_mgnl}--\ref{ec:cardinality_cross_rate}.

\subsection{Experimental Settings}
\paragraph{Instance Generation. }
Our test instances are generated following the procedure in \cite{Cuong2024}. The number of products $m$ varies from 10 to 1000, and the number of nests $N$ is chosen from $\{2, 5, 10, 20\}$. The dissimilarity parameter $\sigma_n$ of each nest $S_n$ is uniformly generated from $[0.25, 1]$. To control the overlap among nests, the average number of nests, which a product belongs to, is set to $\gamma = 1.2$. Specifically, $\lceil \gamma m \rceil$ product entries (sampled with replacement) are randomly assigned to the $N$ nests, allowing each product to appear in multiple nests, but at most once in each nest. For each product $i$ and nest $n$, the allocation parameter $\alpha_{in}$ is drawn from a uniform distribution over $[0,1]$ and normalized such that $\alpha_{in} = 0$ if $i \notin S_n$, and $\sum_{n \in [N]} \alpha_{in} = 1$ for all $i \in [m]$. Two types of cardinality constraints are imposed: a general capacity constraint: $\sum_{i \in [m]} x_i \leq c$ where $c = \lceil 0.5 m \rceil$, and nest-capacity constraints $\sum_{i \in S_n} x_i \leq c_n$ where $c_n = \lceil 0.8 |S_n| \rceil$.

All the experiments are implemented using C++ and conducted on a PC equipped with an Intel(R) Core(TM) i7-9700 CPU @ 3.00GHz, 16 GB of RAM, and the Windows 11 operating system. The formulations are carried out by \texttt{GUROBI 12.0.1} with default settings, except for an 8-thread limit, and \texttt{SCIP 9.1.1} wherein the minimum number of threads used is set to 8. The CPU time limit for each instance is set to 3600 seconds, i.e., the algorithms/solver are forced to stop if they exceed the time budget and report the best solutions found.

\paragraph{Baselines. }
% We evaluate the proposed approaches on a collection of synthetically generated instances covering a broad range of problem settings, including assortment optimization under the GNL/MGNL model and the JAP problem. In all cases, the instances incorporate general linear constraints on the decision variables, reflecting realistic operational and capacity restrictions.
For the assortment optimization problem under the GNL and MGNL models, as well as the JAP problem with \textit{discrete price} levels, we benchmark our methods against the approach of \cite{Cuong2024}, as, to the best of our knowledge, this is the only existing method capable of handling general linear constraints in these settings. For the JAP problem with \textit{continuous pricing decisions}, where fewer exact benchmarks are available, we consider two  baselines. The first baseline discretizes the continuous price domain into a finite set of candidate prices and solves the resulting problem using the corresponding discrete-price formulation. The second baseline directly applies the commercial solver \texttt{SCIP}~\citep{SCIP2024} to the resulting MINLP formulation. These comparisons allow us to assess both the benefits of avoiding price discretization and the computational advantages of the proposed reformulations over generic MINLP solvers.
%Across all experiments, we report solution quality and computational performance under identical instance settings and solver environments to ensure a fair and consistent comparison.
\subsection{Assortment Optimization Under GNL Model}
In this section, we examine the computational performance of our exact methods for the assortment optimization problem under the GNL model. Among the three approaches described in Sections~\ref{sec:convexification_gnl} and Appendix~\ref{ec:mixed_gnl_model}—namely, the bisection method, the convex reformulation via logarithmic transformation, and the bilinear--convex reformulation—we select the bilinear--convex reformulation for detailed comparison, as it achieves the best overall computational performance. A comprehensive comparison among these three approaches is reported in Appendix~\ref{ec:assort-GNL}. We refer to this approach as \texttt{B\&C-BiCo}. For clarity, we report results for three variants of \texttt{B\&C-BiCo}: one using only OA cuts, one using only SCs, and one incorporating both types of cuts (OA+SC).
We compare \texttt{B\&C-BiCo} against a direct baseline from the literature, namely the MILP-based approximation proposed by \citet{Cuong2024}. Among the methods introduced by \citet{Cuong2024}, we consider their two strongest variants: the bisection-based approach (BIS) and the MILP approximation scheme, both implemented with the accuracy level set to $99\%$.

We generate 15 datasets with $N \in \{5,10,20\}, \text{ } m \in \{50,100,200,500,1000\}$ and each dataset contains 5 different instances. For each instance, the preference $V_{in}$ and the revenue $r_i$ of the product $i$ are generated such that products
with higher revenues (i.e., high costs) are more likely to be less preferred. To do this, we randomly generate $u_i \in (0,1]$, $X_i \in [0.1,10]$, $Y_i \in [0.1,10]$ and calculate $V_{in}=(1-u_i)Y_i$ and $r_i=u_i^2X_i$. The no-purchase preference $V_{0n}$ is set to 1 for all $n \in [N]$.
\begin{table}[htb]\footnotesize
\centering
\resizebox{\textwidth}{!}{%
\begin{tabular}{rlrrrlrrrlrrlrrlrr}
\toprule
\multicolumn{1}{l}{} &  & \multicolumn{7}{c}{MILP-based approximation}        &  & \multicolumn{8}{c}{\texttt{B\&C-BiCo} (ours)}                                \\ \cmidrule{3-9} \cmidrule{11-18} 
\multicolumn{1}{l}{} &
   &
  \multicolumn{3}{c}{BIS (99\%)} &
   &
  \multicolumn{3}{c}{MILP (99\%)} &
   &
  \multicolumn{2}{c}{OA} &
   &
  \multicolumn{2}{c}{SC} &
   &
  \multicolumn{2}{c}{OA + SC} \\ \cmidrule{3-5} \cmidrule{7-9} \cmidrule{11-12} \cmidrule{14-15} \cmidrule{17-18} 
$(N,m)$ &
   &
  \#O/\#S &
  Time(s) &
  Gap(\%) &
   &
  \#O/\#S &
  Time(s) &
  Gap(\%) &
   &
  \#O &
  Time(s) &
   &
  \#O &
  Time(s) &
   &
  \#O &
  Time(s) \\ \midrule
(5,50)               &  & 3/5 & 5.65    & -0.58 &  & 3/5 & 2.36    & -0.58 &  & 5 & 0.11           &  & 5 & \textbf{0.07} &  & 5 & \textbf{0.07} \\
(5,100)              &  & 1/5 & 37.54   & -0.28 &  & 1/5 & 9.88    & -0.48 &  & 5 & \textbf{0.18}  &  & 5 & 0.26          &  & 5 & 0.19          \\
(5,200)              &  & 0/5 & 288.08  & -0.25 &  & 0/5 & 53.62   & -0.24 &  & 5 & 0.43           &  & 5 & 1.16          &  & 5 & \textbf{0.70} \\
(5,500)              &  & 0/3 & 2282.58 & -0.47 &  & 0/5 & 812.46  & -0.33 &  & 5 & \textbf{4.86}  &  & 5 & 12.09         &  & 5 & 7.55          \\
(5,1000)             &  & 0/0 & -       & -     &  & 0/2 & 2898.10 & -0.66 &  & 5 & \textbf{10.74} &  & 5 & 26.03         &  & 5 & 14.85         \\ \midrule
(10,50)              &  & 4/5 & 1.68    & 0.00     &  & 3/5 & 12.12   & -0.03 &  & 5 & 0.21           &  & 5 & 0.15          &  & 5 & \textbf{0.14} \\
(10,100)             &  & 3/5 & 19.76   & -0.07 &  & 4/5 & 112.71  & -0.01 &  & 5 & \textbf{0.14}  &  & 5 & 0.21          &  & 5 & 0.18          \\
(10,200)             &  & 1/5 & 509.72  & -0.12 &  & 1/2 & 2391.69 & -0.10  &  & 5 & \textbf{0.27}  &  & 5 & 0.44          &  & 5 & 0.41          \\
(10,500)             &  & 2/2 & 2282.58 & 0.00     &  & 0/1 & 3402.09 & -0.08 &  & 5 & 4.44           &  & 5 & 9.14          &  & 5 & \textbf{3.11} \\
(10,1000)            &  & 0/0 & -       & -     &  & 0/0 & -       & -     &  & 5 & \textbf{20.85} &  & 5 & 26.84         &  & 5 & 30.52         \\ \midrule
(20,50)              &  & 4/5 & 0.50    & -0.02 &  & 4/5 & 3.62    & -0.02 &  & 5 & 0.19           &  & 5 & \textbf{0.15} &  & 5 & \textbf{0.15} \\
(20,100)             &  & 1/5 & 6.59    & -0.03 &  & 2/4 & 2085.58 & -0.02 &  & 5 & 0.73           &  & 5 & 0.50          &  & 5 & \textbf{0.42} \\
(20,200)             &  & 4/5 & 250.14  & -0.01 &  & 0/0 & -       & -     &  & 5 & 0.65           &  & 5 & \textbf{0.58} &  & 5 & 0.60          \\
(20,500)             &  & 0/0 & -       & -     &  & 0/0 & -       & -     &  & 5 & \textbf{2.35}  &  & 5 & 4.18          &  & 5 & 3.55          \\
(20,1000)            &  & 0/0 & -       & -     &  & 0/0 & -       & -     &  & 5 & \textbf{16.68} &  & 5 & 56.19         &  & 5 & 22.46         \\ \midrule 
Summary: &  & 23/50  &  &  &  & 18/44 & & &  & \textbf{75} &  &  & \textbf{75} &  &  & \textbf{75} & \\ \bottomrule
\end{tabular}%
}
\caption{Comparison results on GNL instances.}
\label{tab:GNL}
\end{table}

Table~\ref{tab:GNL} reports a comprehensive comparison of solution quality and computational efficiency between the \texttt{B\&C-BiCo} and the approximation methods of \citet{Cuong2024} across a range of GNL instances. 
We note that the MINLP formulation solvable to optimality by \texttt{SCIP} but we  exclude from this comparison, as it exhibits substantially higher computational times than our proposed exact methods (see Appendix~\ref{ec:assort-GNL} for details).
%—based on maximizing a single variable~$\delta$ subject to the nonlinear constraint
% \begin{equation}
% \sum_{n \in [N]} W_n^{\sigma_n - 1} \left( \sum_{i \in [m]} \alpha_{in} x_i r_i V_{in} \right)
% = \delta \left( \sum_{n \in [N]} W_n^{\sigma_n} \right),\nonumber
% \end{equation}
% together with cardinality constraints

The first column of the table lists the datasets considered. Columns ``\#O/\#S'' report the number of instances whose objective values are certified as optimal by exact methods~(\#O) over the number of instances solved within the prescribed time limit~(\#S). The columns labeled ``Time (s)'' present the average runtime in seconds over five instances for each dataset, while the columns ``Gap (\%)'' report the average optimality gap relative to the optimal values obtained by exact methods. A dash ``--'' indicates that the corresponding instances either timed out or failed to produce a feasible solution from which a gap could be computed.

The results clearly demonstrate the superiority of \texttt{B\&C-BiCo}, which successfully solves all 75 instances with consistently low computation times. In contrast, both the BIS and MILP approximation methods show limited scalability, solving substantially fewer instances as problem size increases. Among the B\&C variants, the one using OA cuts achieves the best runtime performance in most cases, offering a strong balance between computational efficiency and robustness. Overall, these findings highlight the practical advantages of the proposed exact method over approximation-based approaches, particularly for large-scale GNL instances.

\subsection{Assortment  and Price Optimization with Continuous Prices (A\&CP)}
This section reports comparative results for the JAP problem with continuous price decisions under the GNL model. To the best of our knowledge, this is the first study to consider and solve this problem exactly. As a result, we benchmark our proposed approach—based on the PWLA framework described in Section~\ref{sec:JAP}—against two direct baseline methods.

The first baseline directly solves the resulting MINLP problem using the general-purpose solver \texttt{SCIP}. The second baseline discretizes the continuous price domain into a finite set of candidate prices and solves the resulting approximation, which corresponds to the JAP problem with discrete prices. We denote this discretization-based approach as \texttt{B\&C-BiCo-DP}, standing for a B\&C method based on the bilinear--convex reformulation applied to the discretized-price JAP problem.

In this experimental setting, the number of nests is set to $N \in \{2,5\}$, the number of products varies over $m \in \{10,20,30,40,50\}$. For each product~$i$, the lower price bound is fixed at $L_i = 0.5$, while the upper bound is defined as $U_i = L_i \gamma_i + 0.5$, where $\gamma_i \sim (0,1)$ is independently sampled. The \texttt{B\&C-BiCo-DP} approach is evaluated under two discretization schemes. In the first scheme, prices are generated uniformly as
$p_{il} = l \gamma_i + 0.5, \quad \forall l \in [L],$
and in the second scheme,
$p_{il} = 0.5\, l \gamma_i + 0.5, \quad \forall l \in [2L],$
where $L \in \{5,10\}$ controls the discretization granularity.

For the proposed PWLA-based method, the approximation error for $\mathcal{P}_i\!\left(\exp(y_i \eta_i + \kappa_i)\right)$ in formulation~\eqref{prob:CP-Bilinear} is set to $10^{-3}$. We denote this approach as \texttt{B\&C-BiCo-PWLA}, which reformulates the continuous-price JAP problem as a bilinear--convex program using PWLA with adaptively and optimally selected breakpoints, and solves it via a B\&C framework with OA cuts implemented in \texttt{GUROBI}.

\begin{table}[htb]\footnotesize
\centering
\resizebox{0.9\textwidth}{!}{%
\begin{tabular}{rrrrrrrrrrrrrrrr}
\toprule
 &
   &
  \multicolumn{2}{c}{\multirow{2}{*}{\texttt{SCIP}}} &
   &
  \multicolumn{7}{c}{\texttt{B\&C-BiCo-DP}} &
   &
  \multicolumn{3}{c}{\multirow{2}{*}{\texttt{B\&C-BiCo-PWLA}}} \\ \cmidrule{6-12}
 &
   &
  \multicolumn{2}{c}{} &
   &
  \multicolumn{3}{c}{$\#Price = L$} &
   &
  \multicolumn{3}{c}{$\#Price = 2L$} &
   &
  \multicolumn{3}{c}{} \\ \cmidrule{3-4} \cmidrule{6-8} \cmidrule{10-12} \cmidrule{14-16} 
$(N,m,L)$ &
   &
  \#O &
  Time(s) &
   &
  \#B &
  Time(s) &
  Gap(\%) &
   &
  \#B &
  Time(s) &
  Gap(\%) &
   &
  \#B &
  Time(s) &
  Gap(\%) \\ \cmidrule(r){1-12} \cmidrule(l){14-16} 
(2,10,5)  &  & \textbf{5} & 4.11    &  & 0 & 0.05   & -0.10 &  & 1          & 0.1    & -0.02 &  & \textbf{5}  & 0.14    & 0.00  \\
(2,20,5)  &  & 2          & 2937.57 &  & 1 & 0.27   & -0.02 &  & 1          & 0.6    & 0.06  &  & \textbf{5}  & 15.23   & 0.09  \\
(2,30,5)  &  & 0          & -       &  & 0 & 0.38   & 2.01  &  & 0          & 2.42   & 2.09  &  & \textbf{5}  & 195.27  & 2.11  \\
(2,40,5)  &  & 0          & -       &  & 1 & 0.37   & 5.03  &  & 3          & 0.95   & 5.13  &  & \textbf{4}  & 2219.31 & 5.13  \\
(2,50,5)  &  & 0          & -       &  & 0 & 0.92   & 8.31  &  & 1          & 544.34 & 8.37  &  & \textbf{5}  & 2999.52 & 8.38  \\ \midrule
(2,10,10) &  & \textbf{5} & 9.53    &  & 1 & 0.13   & -0.06 &  & 2          & 0.34   & -0.01 &  & \textbf{5}  & 0.21    & 0.00  \\
(2,20,10) &  & 1          & 3441.09 &  & 1 & 0.43   & 0.69  &  & 1          & 1.15   & 0.72  &  & \textbf{5}  & 2268.22 & 0.74  \\
(2,30,10) &  & 0          & -       &  & 0 & 0.61   & 1.88  &  & 2          & 1.41   & 1.90  &  & \textbf{5}  & -       & 1.91  \\
(2,40,10) &  & 0          & -       &  & 0 & 1.43   & 2.01  &  & 3          & 2.45   & 2.04  &  & \textbf{5}  & -       & 2.05  \\
(2,50,10) &  & 0          & -       &  & 1 & 723.87 & 10.58 &  & \textbf{3} & 1206.9 & 10.62 &  & 2           & -       & 10.60 \\ \midrule
(5,10,5)  &  & \textbf{5} & 3.71    &  & 2 & 0.05   & -0.04 &  & 2          & 0.12   & -0.01 &  & \textbf{5}  & 0.14    & 0.00  \\
(5,20,5)  &  & 3          & 2212.17 &  & 1 & 0.15   & 0.20  &  & 1          & 0.43   & 0.31  &  & \textbf{5}  & 77.9    & 0.34  \\
(5,30,5)  &  & 1          & -       &  & 0 & 0.29   & 1.75  &  & 2          & 1.50   & 1.84  &  & \textbf{5}  & 1506.08 & 1.85  \\
(5,40,5)  &  & 0          & -       &  & 0 & 0.40   & 2.55  &  & 2          & 1.53   & 2.66  &  & \textbf{3}  & 2379.64 & 2.67  \\
(5,50,5)  &  & 0          & -       &  & 0 & 0.53   & 2.3   &  & 1          & 3.53   & 2.39  &  & \textbf{3}  & -       & 2.39  \\ \midrule
(5,10,10) &  & \textbf{5} & 6.93    &  & 0 & 0.13   & -0.07 &  & 0          & 0.26   & -0.03 &  & \textbf{5}  & 0.31    & 0.00  \\
(5,20,10) &  & 1          & 3535.19 &  & 0 & 0.39   & 1.55  &  & 0          & 0.88   & 1.63  &  & \textbf{5}  & 2074.77 & 1.66  \\
(5,30,10) &  & 0          & -       &  & 0 & 0.8    & 5.51  &  & 0          & 1.70   & 5.55  &  & \textbf{5}  & -       & 5.57  \\
(5,40,10) &  & 0          & -       &  & 0 & 2.86   & 5.38  &  & \textbf{3} & 21.15  & 5.44  &  & \textbf{3}  & -       & 5.42  \\
(5,50,10) &  & 0          & -       &  & 0 & 1.32   & 3.50  &  & \textbf{4} & 3.50   & 3.55  &  & 3           & -       & 3.52  \\ \midrule
Summary:  &  & 28         &         &  & 8 &        &       &  & 32         &        &       &  & \textbf{88} &         &       \\ \bottomrule
\end{tabular}
}
\caption{Comparison results on A\&CP instances.}
\label{tab:ACP}
\end{table}

Table~\ref{tab:ACP} reports the performance comparison across three approaches: \texttt{SCIP}, \texttt{B\&C-BiCo-DP} (with varying number of discrete prices), and \texttt{B\&C-BiCo-PWLA}. Columns ``\#B" report the number of best-found solutions obtained by our approximation methods, whereas column ``\#O" presents the number of instances solved to optimality by \texttt{SCIP}. Columns ``Gap(\%)" show the average gap relative to the objective values obtained by \texttt{SCIP}. For each instance and B\&C-based method, the gap is calculated as: $\frac{Objective_{\texttt{B\&C-BiCo}}-Objective_{\texttt{SCIP}}}{Objective_{\texttt{SCIP}}}\times 100(\%)$.

The \texttt{B\&C-BiCo-PWLA} approach clearly demonstrates superior performance in solution quality. It achieves the highest number of best-found objective values (88 out of 100 instances), significantly outperforming \texttt{SCIP} (28/100) and both \texttt{B\&C-BiCo-DP} variants (9/100 and 32/100). For the gap, \texttt{B\&C-BiCo-PWLA} tends to achieve slightly higher gaps in most instance sets, but the improvement over \texttt{B\&C-BiCo-DP} is limited, suggesting comparable solution quality between the two approaches (see Appendix~\ref{ec:obj_jap_cp} for more details). The table further illustrates the inherent trade-off between computational time and solution quality across the three B\&C variants. Although \texttt{B\&C-BiCo-PWLA} yields the highest number of high-quality solutions, it frequently fails to prove the optimality during the time limit when $m$ increases. In contrast, using smaller price sets in \texttt{B\&C-BiCo-DP} significantly reduces computation time, but at the cost of worse average optimality gaps, particularly for instances with smaller values of \(m\), while employing larger price sets achieves more balanced performances in terms of both runtime and solution quality. 

\section{Conclusion}\label{sec:concl}

In this paper, we study assortment and JAP problems under GNL choice models with general linear constraints. We develop a unified exact solution framework based on bilinear--convex reformulations and B\&C algorithms enhanced with OA cuts and SCs. The framework applies to GNL and MGNL models, as well as JAP problems with both discrete and continuous price decisions; for the latter, we introduce a PWLA scheme that enables efficient exact optimization without sacrificing modeling fidelity. Extensive numerical experiments show that the proposed methods substantially outperform existing MILP-based approximations in both scalability and solution quality. In particular, our approach consistently solves large-scale instances to optimality or near-optimality, while approximation methods degrade as problem size and customer heterogeneity increase. The results also demonstrate the economic value of explicitly modeling heterogeneity through MGNL and of avoiding price discretization in continuous-price JAP settings.

Several directions for future research are promising, including extensions to multi-level GNL models, dynamic assortment and pricing with learning and inventory considerations, and further enhancements to approximation and cut-generation strategies for large-scale problems.

\bibliographystyle{plainnat_custom}
\bibliography{ref,refs-POM}

\clearpage

\appendix
\begin{center}
    {\Large Appendix }
\end{center}

This appendix provides supplementary technical material that supports and extends the results presented in the main paper. 
Section~\ref{ec:proofs} presents detailed derivations and proofs of the theoretical properties underlying the GNL-based formulation, including proofs of all propositions stated in the main text. 
Section~\ref{ec:mixed_gnl_model} provides detailed reformulations for the assortment optimization problem under the MGNL model. 
Section~\ref{ec:JAP_DP} presents the formulation and solution methods for the JAP problem with discrete prices. 
Section~\ref{ec:PWLA} discusses in detail the construction of optimal breakpoints for the PWLA in the JAP problem with continuous prices. 
Section~\ref{ec:zero_opt_out} studies reformulations and solution methods for the case in which the opt-out (non-purchase) option may take zero values. 
Finally, Sections~\ref{ec:assort-GNL}--\ref{ec:cardinality_cross_rate} report additional computational experiments, including comparisons of the different approaches developed for the assortment optimization and JAP problems presented in the main paper, comparative analyses between the GNL and MGNL models, comparisons between JAP with continuous and discrete prices, and investigations into the impact of cardinality constraints and cross-rate parameters in the GNL model.

\section{Missing Proofs}\label{ec:proofs}
This section provides proofs that are omitted from the main paper.

\subsection{Proof of Proposition 1:}
\begin{proof}
Start from the objective of \eqref{prob:GNL-assort}:
\[
F(\bx)\;=\;\frac{\sum_{i\in[m]}\sum_{n\in[N]} W_n^{\sigma_n-1}\,\alpha_{in}x_i r_i V_{in}}
{\sum_{n\in[N]} W_n^{\sigma_n}}\;=\;\frac{\sum_{n\in[N]} W_n^{\sigma_n-1}\Big(\sum_{i\in[m]}\alpha_{in}x_i r_i V_{in}\Big)}
{\sum_{n\in[N]} W_n^{\sigma_n}}.
\]
Fix any $\beta>\max_i \{r_i\}$ and note the algebraic identity
$F(\bx)\;=\;\beta-\big(\beta-F(\bx)\big).$
Because the denominator is strictly positive by assumption, we can write
\begin{align*}
\beta-F(\bx)
&= \frac{\beta\sum_{n\in[N]} W_n^{\sigma_n}}{\sum_{n\in[N]} W_n^{\sigma_n}}
   - \frac{\sum_{n\in[N]} W_n^{\sigma_n-1}\Big(\sum_{i\in[m]}\alpha_{in}x_i r_i V_{in}\Big)}
          {\sum_{n\in[N]} W_n^{\sigma_n}}\\
&= \frac{\sum_{n\in[N]} W_n^{\sigma_n-1}\Big(\beta W_n-\sum_{i\in[m]}\alpha_{in}x_i r_i V_{in}\Big)}
         {\sum_{n\in[N]} W_n^{\sigma_n}}\\
&= \frac{\sum_{n\in[N]} W_n^{\sigma_n-1}\Big(\beta\,(V_{0n}+\sum_{i}\alpha_{in}x_i V_{in})-\sum_{i}\alpha_{in}x_i r_i V_{in}\Big)}
         {\sum_{n\in[N]} W_n^{\sigma_n}}\\
&= \frac{\sum_{n\in[N]} W_n^{\sigma_n-1}\Big(\beta V_{0n}+\sum_{i}\alpha_{in}x_i(\beta-r_i) V_{in}\Big)}
         {\sum_{n\in[N]} W_n^{\sigma_n}}\\
&= \frac{\sum_{n\in[N]} W_n^{\sigma_n-1}\Big(\beta V_{0n}+\sum_{i}\alpha_{in}x_i r'_i V_{in}\Big)}
         {\sum_{n\in[N]} W_n^{\sigma_n}}.
\end{align*}
Hence,
\[
F(\bx)\;=\;\beta-\frac{\sum_{n\in[N]} W_n^{\sigma_n-1}\Big(\beta V_{0n}+\sum_{i}\alpha_{in}x_i r'_i V_{in}\Big)}
{\sum_{n\in[N]} W_n^{\sigma_n}}.
\]
Since $\beta$ is a constant independent of $\bx$, maximizing $F(\bx)$ over $\bx\in\cX$ is equivalent to minimizing the fraction on the right-hand side over the same feasible set. The constraints defining $W_n$ are identical to those in \eqref{prob:GNL-assort}, so the feasible regions coincide in $(\bx,\bW)$. Therefore, \eqref{prob:GNL-assort} and \eqref{prob:GNL-assort-min} have the same optimal assortments and their optimal values are related by
\[
\text{OPT}(\text{\sf Assort-1}) \;=\; \beta - \text{OPT}(\text{\sf Assort-2}).
\]
This proves the claimed max--min equivalence.
\end{proof}

\subsection{Proof of Proposition 2:}
\begin{proof}
Since $W_n(\bx) = V_{0n} + \sum_{i} \alpha_{in} x_i V_{in}$ is an affine (and hence both convex and concave) function of $\bx$, the convexity and concavity of $H_n$ and $K_n$ follow directly from standard composition rules of convex analysis. In particular:
\begin{itemize}
    \item For $\sigma_n \in (0,1]$, the function $t \mapsto t^{\sigma_n}$ is concave and nondecreasing on $\mathbb{R}_+$. Hence, $K_n(\bx) = W_n(\bx)^{\sigma_n}$ is concave in $\bx$.
    \item For $\sigma_n \in (0,1)$, the function $t \mapsto t^{\sigma_n - 1}$ is convex and nonincreasing on $\mathbb{R}_+$. Therefore, $H_n(\bx) = W_n(\bx)^{\sigma_n - 1}$ is convex in $\bx$. When $\sigma_n = 1$, we have $H_n(\bx) = W_n^0 = 1$, which is trivially convex.
\end{itemize}
\end{proof}

\subsection{Proof of Proposition 3:}
\begin{proof}
Let $W_n(S) = V_{0n} + \sum_{i \in S} \alpha_{in} V_{in}$.  
Since all parameters $\alpha_{in} V_{in} \ge 0$, $W_n(S)$ is a nondecreasing set function; that is, 
for any $A \subseteq B \subseteq [m]$, we have $W_n(A) \le W_n(B)$.  
We analyze the two functions in turn.

\paragraph{(i) Submodularity of $K_n(S)$.}
For $\sigma_n \in (0,1]$, the mapping $t \mapsto t^{\sigma_n}$ is concave and nondecreasing over $\mathbb{R}_+$.  
By a standard result in submodular analysis (see, e.g., Bach 2013),  
the composition of a nondecreasing concave function with a nondecreasing modular (additive) function 
yields a \emph{submodular} set function.  
Hence $K_n(S) = f(W_n(S))$ with $f(t) = t^{\sigma_n}$ is submodular.  
Moreover, since $f$ is nondecreasing, $K_n(S)$ is also monotonically increasing in $S$; 
that is, adding any element $i \notin S$ never decreases the value of $K_n(S)$.

\paragraph{(ii) Supermodularity of $H_n(S)$.}
For $\sigma_n \in (0,1)$, the mapping $t \mapsto t^{\sigma_n - 1}$ is convex and nonincreasing over $\mathbb{R}_+$.  
Again using composition results for set functions,  
the composition of a nonincreasing convex function with a nondecreasing modular function 
produces a \emph{supermodular} set function (the sign of the marginal difference reverses).  
Therefore, $H_n(S)$ is supermodular.  
Furthermore, since the function $t^{\sigma_n - 1}$ is nonincreasing, $H_n(S)$ is monotonically decreasing:  
adding more elements to $S$ reduces its value.  
When $\sigma_n = 1$, $H_n(S) = 1$ is constant, and thus trivially both supermodular and submodular.

Formally, the supermodularity of $H_n$ can also be verified via the diminishing-increase property:  
for any $A \subseteq B \subseteq [m]$ and any $i \notin B$:
$H_n(A \cup \{i\}) - H_n(A) 
\le 
H_n(B \cup \{i\}) - H_n(B),$ since the marginal decrease in $H_n(S)$ becomes smaller as $S$ grows.  
An analogous argument applies for $K_n(S)$, 
where the concavity of $t^{\sigma_n}$ ensures diminishing returns, i.e., 
$
K_n(A \cup \{i\}) - K_n(A)
\ge 
K_n(B \cup \{i\}) - K_n(B),
$ confirming submodularity.
\end{proof}

\subsection{Proof of Proposition 4:}
\begin{proof}
We recall the standard set-function definitions. A set function $f:2^{[m]}\to\mathbb{R}$ is
\emph{submodular} if for all $A\subseteq B\subseteq [m]$ and $j\notin B$,
\[
\Delta_f(A,j)\;\equiv\;f(A\cup\{j\})-f(A)\;\ge\; f(B\cup\{j\})-f(B)\;=\;\Delta_f(B,j),
\]
i.e., it has \emph{diminishing returns}. It is \emph{supermodular} if $-f$ is submodular, i.e., if
the above inequality is reversed. Monotonicity is with respect to set inclusion.

\medskip
\noindent\textbf{(i) Supper-modularity of $Y_n$.}
Fix a nest $n$ and write $W(S)=W_n(S)$ for brevity. Note that $W(S)$ is \emph{modular} (additive) in $S$
with nonnegative weights because $\alpha_{in}V_{in}\ge 0$ and $V_{0n}\ge 0$. Since $\sigma_n\in(0,1]$, the scalar function: 
$
\phi(t)\;=\;(\sigma_n-1)\log t, t>0,
$ is \emph{convex} and \emph{nonincreasing} (indeed, $\phi'(t)=(\sigma_n-1)/t\le 0$ and
$\phi''(t)=(1-\sigma_n)/t^2\ge 0$). Hence $Y_n(S)=\phi(W(S))$ is the composition of a
nonincreasing convex function with a nondecreasing modular function.

\emph{Monotonicity.} If $S\subseteq T$ then $W(S)\le W(T)$, so because $\phi$ is nonincreasing,
$Y_n(S)=\phi\big(W(S)\big)\ge \phi\big(W(T)\big)=Y_n(T)$. Thus $Y_n$ is monotonically decreasing.

\emph{Supermodularity.}
Fix $A\subseteq B$ and $j\notin B$. Let $\delta=\alpha_{jn}v_j\ge 0$. Consider the marginal change
\[
\Delta_{Y_n}(S,j)\;=\;Y_n(S\cup\{j\})-Y_n(S)
\;=\;(\sigma_n-1)\Big[\log\!\big(W(S)+\delta\big)-\log W(S)\Big].
\]
Define $g(w):=\log(w+\delta)-\log w=\log\!\big(1+\delta/w\big)$ for $w>0$. Then
$g'(w)=\frac{1}{w+\delta}-\frac{1}{w}=-\frac{\delta}{w(w+\delta)}<0$, so $g$ is strictly
\emph{decreasing} in $w$. Since $W(A)\le W(B)$, we have $g\big(W(A)\big)\ge g\big(W(B)\big)$.
Multiplying by the negative constant $(\sigma_n-1)\le 0$ reverses the inequality:
\[
\Delta_{Y_n}(A,j) \;=\; (\sigma_n-1)g\big(W(A)\big)
\;\le\; (\sigma_n-1)g\big(W(B)\big) \;=\; \Delta_{Y_n}(B,j).
\]
Therefore $\Delta_{Y_n}(A,j)\le \Delta_{Y_n}(B,j)$, i.e., $Y_n$ is \emph{supermodular}.

\noindent\textbf{(ii) Submodularity of $Z(S)$.}
Let $S\subseteq [m]$ be a given assortment. For each nest $n\in[N]$, define 
$\delta_{nj} := \alpha_{jn} v_j \ge 0,$ for each product $j\in [m]$ and nest $n\in [N]$.
Then the nonlinear component of interest is: 
$
Z(S) = \log\!\left(\sum_{n=1}^N W_n(S)^{\sigma_n}\right).
$

Define the shorthand:
$
U_n(S) := W_n(S)^{\sigma_n}, 
\text{ and } 
H(S) := \sum_{n=1}^N U_n(S).
$
For any $j \notin S$, adding product $j$ increases each nest’s term by: 
$
U_n(S \cup \{j\}) = (W_n(S) + \delta_{nj})^{\sigma_n}.
$
Hence, the marginal gain of adding $j$ is
\[
\Delta_Z(S,j)
  := Z(S \cup \{j\}) - Z(S)
  = \log\!\big(H(S) + \Delta_H(S,j)\big)
    - \log\!\big(H(S)\big),
\]
where $\Delta_H(S,j)
     = \sum_{n=1}^N \big[(W_n(S) + \delta_{nj})^{\sigma_n} - W_n(S)^{\sigma_n}\big]$.

\smallskip
\noindent\emph{Monotonicity.}
Since $\sigma_n \in (0,1]$ and each $\delta_{nj} \ge 0$, we have
$(W_n(S) + \delta_{nj})^{\sigma_n} \ge W_n(S)^{\sigma_n}$ for all $n$.
Thus $\Delta_H(S,j) \ge 0$ and consequently $\Delta_Z(S,j) \ge 0$,
which establishes that $Z(S)$ is monotonically increasing in $S$.

\smallskip
\noindent\emph{Submodularity (diminishing returns).}
Let $A \subseteq B \subseteq [m]$ and $j \notin B$.  
Define vectors 
\[
\mathbf{w}(S) := (W_1(S), \dots, W_N(S))^\top, \qquad
\boldsymbol{\delta}_j := (\delta_{1j}, \dots, \delta_{Nj})^\top.
\]
We can express $Z(S)$ as a function of $\mathbf{w}(S)$: 
$
Z(S) = f(\mathbf{w}(S)),\text{ where } f(\mathbf{x}) = \log\!\Big(\sum_{n=1}^N x_n^{\sigma_n}\Big)
$.

Then, for any $\mathbf{x} > 0$, the marginal effect of adding product $j$ is
\[
\Delta_f(\mathbf{x}, \boldsymbol{\delta}_j)
    := f(\mathbf{x} + \boldsymbol{\delta}_j) - f(\mathbf{x})
    = \log\!\left(\sum_{n=1}^N (x_n + \delta_{nj})^{\sigma_n}\right)
     - \log\!\left(\sum_{n=1}^N x_n^{\sigma_n}\right).
\]

To prove submodularity, it suffices to show that
$\Delta_f(\mathbf{x}, \boldsymbol{\delta}_j)$ is nonincreasing in $\mathbf{x}$ componentwise.
We compute the partial derivative with respect to $x_k$:
\[
\frac{\partial \Delta_f(\mathbf{x}, \boldsymbol{\delta}_j)}{\partial x_k}
= \frac{\sigma_k (x_k + \delta_{kj})^{\sigma_k - 1}}
       {\sum_{n=1}^N (x_n + \delta_{nj})^{\sigma_n}}
  - \frac{\sigma_k x_k^{\sigma_k - 1}}
       {\sum_{n=1}^N x_n^{\sigma_n}}.
\]

Since $\sigma_k \in (0,1]$, the function $t \mapsto t^{\sigma_k - 1}$ is strictly decreasing on
$(0, \infty)$. Thus, $(x_k + \delta_{kj})^{\sigma_k - 1} \le x_k^{\sigma_k - 1}$, and moreover
$
\sum_{n=1}^N (x_n + \delta_{nj})^{\sigma_n} \ge \sum_{n=1}^N x_n^{\sigma_n}
$. 

Both effects cause the first term above to be less than or equal to the second, implying that:
$
\frac{\partial \Delta_f(\mathbf{x}, \boldsymbol{\delta}_j)}{\partial x_k} \le 0
\quad \forall k
$. Hence $\Delta_f(\mathbf{x}, \boldsymbol{\delta}_j)$ is nonincreasing in $\mathbf{x}$ componentwise.

Because $\mathbf{w}(A) \le \mathbf{w}(B)$ elementwise whenever $A \subseteq B$, we obtain
\[
\Delta_Z(A,j) = \Delta_f(\mathbf{w}(A), \boldsymbol{\delta}_j)
   \ge \Delta_f(\mathbf{w}(B), \boldsymbol{\delta}_j)
   = \Delta_Z(B,j),
\]
which verifies the diminishing returns property. Therefore, $Z$ is submodular.
\end{proof}

\subsection{Proof of Proposition 5:}
\begin{proof} Recall that, at optimality, we have $y_n = (\sigma_n - 1)\log(W_n)$ for each $n \in [N]$, where $W_n > 0$. From the convex reformulation, it follows that: 
$
z \;\geq\; \log\!\left(\sum_{n\in [N]} W_n^{\sigma_n}\right)
$. 
Substituting the expression for $y_n$ into the right-hand side, we obtain: 
$
e^{\frac{\sigma_n}{\sigma_n - 1} y_n} 
= e^{\frac{\sigma_n}{\sigma_n - 1}(\sigma_n - 1)\log W_n}
= e^{\sigma_n \log W_n} 
= W_n^{\sigma_n}
$.\\
Consequently,
$
\sum_{n\in [N]} e^{\frac{\sigma_n}{\sigma_n - 1} y_n}
= \sum_{n\in [N]} W_n^{\sigma_n}
$.\\
Taking logarithms on both sides gives
$
\log\!\left(\sum_{n\in [N]} e^{\frac{\sigma_n}{\sigma_n - 1} y_n}\right)
= \log\!\left(\sum_{n\in [N]} W_n^{\sigma_n}\right)
\;\leq\; z
$.\\
Therefore, the inequality 
$
z \;\geq\; \log\!\left(\sum_{n\in [N]} e^{\frac{\sigma_n}{\sigma_n - 1} y_n}\right)
$ is valid and holds with equality at the optimal solution of the original nonlinear formulation.  

\smallskip
Moreover, note that the right-hand side of the inequality defines a convex function of $\by = (y_1,\dots,y_N)$, because it is the composition of the convex and nondecreasing function $\log(\cdot)$ with the \emph{log-sum-exp} form
$\sum_{n} \exp\!\big(\tfrac{\sigma_n}{\sigma_n - 1}y_n\big)$, which is convex in $\by$. Consequently, the constraint
\[
z \;\geq\; \log\!\left(\sum_{n\in [N]} e^{\frac{\sigma_n}{\sigma_n - 1} y_n}\right)
\]
defines a convex feasible set. It can therefore be safely incorporated into the reformulation to strengthen the relaxation without sacrificing tractability or convergence guarantees.
\end{proof}

\subsection{Proof of Proposition 6:}
\begin{proof}
For clarity, define the numerators and denominators of the original and approximate objectives as:
$
F(\bx,\by) = \frac{N(\bx,\by)}{D(\bx,\by)}
$ and $
\widehat F(\bx,\by) = \frac{\widehat N(\bx,\by)}{\widehat D(\bx,\by)},
$\\
where:
\begin{align*}
N(\bx,\by)
&:=
\sum_{n\in[N]}
W_n^{\sigma_n-1}
\left(
\beta V_{0n}
+
\sum_{i\in[m]}
\alpha_{in} x_i (\beta-y_i) t_{in}
\right),\\
\widehat N(\bx,\by)
&:=
\sum_{n\in[N]}
\widehat W_n^{\sigma_n-1}
\left(
\beta V_{0n}
+
\sum_{i\in[m]}
\alpha_{in} x_i (\beta-y_i) \widehat t_{in}
\right),\\
D(\bx,\by)
&:= \sum_{n\in[N]} W_n^{\sigma_n},
\qquad
\widehat D(\bx,\by) := \sum_{n\in[N]} \widehat W_n^{\sigma_n}.
\end{align*}
Here,
$
t_{in} = \exp\!\left(\frac{-y_i\eta_i+\kappa_i}{\sigma_n}\right),
\widehat t_{in} = \cPA(t_{in}),
$
and
$
W_n = V_{0n} + \sum_{i\in[m]} \alpha_{in} x_i t_{in},
\widehat W_n = V_{0n} + \sum_{i\in[m]} \alpha_{in} x_i \widehat t_{in}.
$
By construction of the PWLA, we have the uniform bound: 
$
|t_{in} - \widehat t_{in}| \le \epsilon,
\forall i,n.
$

Now, using a standard inequality for ratios with positive denominators, we have
\[
\left|
\frac{N}{D} - \frac{\widehat N}{\widehat D}
\right|
\le
\frac{|N-\widehat N|}{D}
+
\frac{\widehat N}{D\,\widehat D}\,|D-\widehat D|.
\tag{A.1}
\]
For each nest $n$,
\[
|W_n-\widehat W_n|
\le
\sum_{i\in[m]} \alpha_{in} x_i |t_{in}-\widehat t_{in}|
\le
\epsilon \sum_{i\in[m]} \alpha_{in}.
\tag{A.2}
\]
Moreover, since $\sigma_n\in(0,1]$, the function $w\mapsto w^{\sigma_n}$ is Lipschitz continuous on
$[\underline W,\infty)$, where $\underline W := \min_n W_n > 0$.
Thus, for some constant $\delta_n>0$,
$
|W_n^{\sigma_n}-\widehat W_n^{\sigma_n}|
\le
\delta_n |W_n-\widehat W_n|.
$
Combining with~(A.2) yields
\[
|D-\widehat D|
\le
\epsilon
\sum_{n\in[N]} \delta_n \sum_{i\in[m]} \alpha_{in}
=
\mathcal{O}(\epsilon).
\tag{A.3}
\]
{We now bound the gap between $|N-\widehat N|$. To this end,}
we decompose
\[
N-\widehat N
=
\sum_{n\in[N]}
\Bigl[
W_n^{\sigma_n-1} A_n
-
\widehat W_n^{\sigma_n-1} \widehat A_n
\Bigr],
\]
where: 
$
A_n := \beta V_{0n} + \sum_{i\in[m]} \alpha_{in} x_i (\beta-y_i) t_{in},
\quad
\widehat A_n := \beta V_{0n} + \sum_{i\in[m]} \alpha_{in} x_i (\beta-y_i) \widehat t_{in}.
$
Adding and subtracting $W_n^{\sigma_n-1}\widehat A_n$ and applying the triangle inequality,
\[
|N-\widehat N|
\le
\sum_{n\in[N]}
\Bigl(
|W_n^{\sigma_n-1}-\widehat W_n^{\sigma_n-1}|\,|\widehat A_n|
+
W_n^{\sigma_n-1}|A_n-\widehat A_n|
\Bigr).
\]
Using the PWLA error bound, the boundedness of $(\beta-y_i)$, and the Lipschitz continuity of
$w^{\sigma_n-1}$ on $[\underline W,\infty)$, we also obtain
\[
|N-\widehat N| = \mathcal{O}(\epsilon).
\tag{A.4}
\]
By assumption,
$
D \ge N \underline W^{\sigma_{\min}},
\widehat D \ge N \underline W^{\sigma_{\min}} - \mathcal{O}(\epsilon),
$ 
which are strictly positive for sufficiently small $\epsilon$. Substituting the bounds~(A.3)–(A.4) into~(A.1) yields:
$
|F(\bx,\by)-\widehat F(\bx,\by)|
\le
C\,\epsilon,
$
for some constant $C>0$ independent of $\epsilon$.
Hence, the PWLA-based formulation provides an $\mathcal{O}(\epsilon)$-accurate approximation
of the original JAP objective.
\end{proof}

\section{Assortment Optimization under the MGNL Model}\label{ec:mixed_gnl_model}

While the GNL model captures flexible correlation structures among products, it assumes a homogeneous customer population. In practice, consumers are heterogeneous: different customer segments may perceive product utilities differently or follow distinct correlation patterns across nests. To capture such heterogeneity, we consider the MGNL model, which represents a mixture of multiple GNL models with segment-specific parameters, each corresponding to a distinct customer segment. %This extension is natural and widely used in discrete choice modeling to represent a population as a convex combination of multiple choice models, each presents a set of customers.

\subsection{Model Setup}

Suppose there are $T$ customer segments (or types), indexed by $t \in [T] := \{1,2,\ldots,T\}$. Segment $t$ arrives with probability (or weight) $\theta_t > 0$, where $\sum_{t=1}^T \theta_t = 1$. Each segment $t$ follows its own GNL structure:
\begin{itemize}
    \item The set of products is the same across all segments, indexed by $[m]$.
    \item Segment $t$ has utility parameters $V_{tin}$ and revenues $r_{ti}$ for each product $i \in [m]$.
    \item The nests for segment $t$ are $\{S^t_1,\ldots,S^t_N\}$, with membership weights $\alpha_{tin}$ satisfying $\sum_{n=1}^N \alpha_{tin} = 1$ for each product $i$.
    \item Each nest $n$ has a dissimilarity parameter $\sigma_{tn} \in (0,1]$ and an opt-out utility $V_{t0n} > 0$.
\end{itemize}
Similar to the case of homogeneous customers, we assume that the opt-out alternative has a positive weight in each nest. This assumption is reasonable, as it reflects the realistic scenario in which every customer type, when considering a particular nest, retains the possibility of not purchasing any item and instead choosing to opt out entirely.  For completeness, we provide in the Appendix~\ref{ec:zero_opt_out} a discussion on how our approach can be extended to the case where this assumption does not hold, i.e., when the opt-out alternative has zero weight in some nests.

Given an assortment $\bx \in \{0,1\}^m$, the total preference weight that customer segment $t$ assigns to nest $n$ is expressed as
$W_{tn}(\bx) = V_{t0n} + \sum_{i \in [m]} \alpha_{tin} V_{tin} x_i,$
where $V_{t0n}$ denotes the utility weight of the opt-out alternative in nest $n$.
Following the standard two-stage structure of the GNL model, for each segment $t$, a customer will first select a nest according to its aggregate attractiveness and then chooses a specific product within the selected nest. Accordingly, the probability that a customer from segment $t$ purchases product $i \in [m]$ under assortment $\bx$ is given by
\[
P_t(i \mid \bx) 
= \sum_{n=1}^N 
\frac{W_{tn}(\bx)^{\sigma_{tn}-1}\, \alpha_{tin} V_{tin} x_i}
     {\sum_{n'=1}^N W_{tn'}(\bx)^{\sigma_{tn'}}},
\]
where $\sigma_{tn} \in (0,1]$ is the dissimilarity parameter of nest $n$ for segment $t$.
The inner fraction corresponds to the conditional probability of choosing product $i$ given nest $n$, while the outer sum aggregates over all nests that product $i$ belongs to. The expected revenue from segment $t$ under assortment $\bx$ is then
\[
F_t(\bx)
= \frac{\sum_{i \in [m]} \sum_{n \in [N]} 
     W_{tn}(\bx)^{\sigma_{tn}-1}\, \alpha_{tin} V_{tin} r_{ti} x_i}
     {\sum_{n \in [N]} W_{tn}(\bx)^{\sigma_{tn}}},
\]
 Intuitively, the numerator represents the expected revenue weighted by both the attractiveness of each product and its contribution to each nest, while the denominator normalizes across all nests.

Assuming that customer segments arrive according to probabilities $\theta_t$, the overall expected revenue under assortment $\bx$ is
$F(\bx) = \sum_{t=1}^T \theta_t F_t(\bx).$ Thus, the assortment optimization problem under the MGNL model can be formulated as
\begin{align}
\max_{\bx \in \cX} \quad 
& F(\bx)
= \sum_{t=1}^T \theta_t 
\frac{\sum_{i \in [m]} \sum_{n \in [N]} 
      W_{tn}^{\sigma_{tn}-1} \, \alpha_{tin} V_{tin} r_{ti} x_i}
     {\sum_{n \in [N]} W_{tn}^{\sigma_{tn}}}\label{prob:assort-MGNL}
\tag{\sf Assort-MGNL} \\
\text{s.t.} \quad 
& W_{tn} = V_{t0n} + \sum_{i \in [m]} \alpha_{tin} V_{tin} x_i,
\quad \forall\, t \in [T],\, n \in [N]. \nonumber
\end{align}
This formulation generalizes the standard GNL assortment problem to account for heterogeneous customer segments, allowing each segment to differ in its sensitivity parameters, nest memberships, and opt-out preferences. As a result, the model can flexibly capture diverse behavioral patterns while maintaining analytical tractability through the MGNL framework.

Problem~\eqref{prob:assort-MGNL} for a single GNL model already leads to a highly nonlinear and non-concave objective function. The MGNL formulation introduces an additional layer of complexity: the objective becomes a weighted sum of multiple fractional terms, each corresponding to a distinct customer segment. This structure destroys the convenient bisection-type reformulations that are applicable to single-ratio objectives and significantly complicates both the theoretical analysis and algorithmic design. 
In fact, even under simplified settings, the problem remains computationally intractable. It is NP-hard to approximate within any constant factor, even when $T = 1$ (i.e., a single segment) \citep{Cuong2024}, or when $T > 1$ but each segment’s choice model reduces to the standard MNL model \citep{Desiretal2022}. To the best of our knowledge, no existing method can solve \eqref{prob:assort-MGNL} under the MGNL model to global optimality. Approximation approaches based on PWLA and MILP reformulations have been explored in the literature, but they scale poorly with the number of segments $T$ and the desired accuracy \citep{Cuong2024}. In the following sections, we propose new formulations and cutting-plane algorithms that can handle the MGNL setting effectively and provide exact optimality guarantees.

\subsection{Convexification via Log-Transformation}
We show below that our convexification approach developed for the single-customer segment in Section~\ref{sec:convexification_gnl} can be extended naturally to the MGNL model. 
This convex reformulation transforms the multi-segment fractional objective into a MICP that can be solved to global optimality using a single B\&C procedure. 

To start, recall that the MGNL assortment optimization problem can be stated as:
\begin{align}
\max_{\bx \in \cX} \quad &
F(\bx) =
\sum_{t\in[T]} \theta_t 
\frac{\sum_{i\in[m]}\sum_{n\in[N]} 
W_{tn}(\bx)^{\sigma_{tn}-1}\alpha_{tin}V_{tin}r_{ti}x_i}
{\sum_{n\in[N]} W_{tn}(\bx)^{\sigma_{tn}}}, \nonumber \\
\text{s.t.} \quad 
& W_{tn}(\bx) = V_{t0n} + \sum_{i\in[m]}\alpha_{tin}V_{tin}x_i,
\quad \forall t \in [T],\, n \in [N] \nonumber.
\end{align}
Our first step is to convert the maximization problem into a minimization problem to facilitate the analysis of convexity and submodularity properties.
Let $\beta > \max_{t,i} r_{ti}$ (the maximum revenue over all customer segments) and define $r'_{ti} = \beta - r_{ti} > 0$. 
By applying the maximization–minimization equivalence as in Section~\ref{sec:convexification_gnl}, the problem can be equivalently rewritten as:
\begin{equation}
\min_{\bx \in \cX} 
\sum_{t\in[T]} \theta_t 
\frac{\sum_{n\in[N]} W_{tn}(\bx)^{\sigma_{tn}-1}
\left(\beta V_{t0n} + \sum_{i\in[m]}\alpha_{tin}V_{tin}r'_{ti}x_i\right)}
{\sum_{n\in[N]} W_{tn}(\bx)^{\sigma_{tn}}} \nonumber.
\end{equation}

The objective contains multiple ratio terms, one per segment $t$, which destroys the convenient single-ratio structure exploited by the bisection scheme. 
To convexify these coupled fractional components simultaneously, we introduce logarithmic auxiliary variables for each segment and nest:
\begin{equation}\label{eq:MGNL-eq}
y_{tn} = \log\left(W_{tn}^{\sigma_{tn}-1}\right), 
\quad 
z_t = \log\left(\sum_{n\in[N]} W_{tn}^{\sigma_{tn}}\right).
\end{equation}

Substituting these definitions and applying the exponential transformation $t_{tn} \ge e^{y_{tn}-z_t}$, we obtain the following equivalent MINLP:
\begin{equation}\label{prob: Assort--MGNL--Convex}
\begin{aligned}
\min_{\bx,\,\bW,\,\by,\,\bz,\,\bt}
&\quad
\sum_{t\in[T]} \theta_t 
\Bigg[
\sum_{n\in[N]} 
\left(\beta V_{t0n}\, t_{tn}
+ \sum_{i\in[m]} \alpha_{tin}V_{tin}r'_{ti}\,x_i t_{tn}\right)
\Bigg] \\
\text{s.t. } 
&\quad W_{tn} = V_{t0n} + \sum_{i\in[m]} \alpha_{tin}V_{tin}x_i, \quad \forall t,n,\\
&\quad t_{tn} \ge e^{y_{tn}-z_t}, \quad \forall t,n,\\
&\quad y_{tn} \ge (\sigma_{tn}-1)\log(W_{tn}), \quad \forall t,n,\\
&\quad z_t \le \log\!\left(\sum_{n\in[N]} W_{tn}^{\sigma_{tn}}\right), \quad \forall t,\\
&\quad x \in \{0,1\}^m,\;
W,y,t \ge 0.
\end{aligned}
\tag{\sf MGNL-Convex}
\end{equation}

The equalities in~\eqref{eq:MGNL-eq} can be safely relaxed into inequality constraints because, at optimality, these inequalities always bind. 
Specifically, for each $t,n$, decreasing $y_{tn}$ or increasing $z_t$ would increase the objective value, while tightening the inequalities to equality never worsens the solution. 
Hence, the inequalities preserve the exactness of the formulation while simplifying the convex analysis.

Following the reformulation approach detailed in Section~\ref{sec:convexification_gnl}, each constraint in the aforementioned problem is convex with respect to the continuous variables. Specifically, the exponential epigraph constraint $t_{tn} \ge e^{y_{tn}-z_t}$ is convex because the exponential function is both convex and nondecreasing. The inequality $y_{tn} \ge (\sigma_{tn}-1)\log(W_{tn})$ defines a convex set because the logarithm is concave and the term $\sigma_{tn}-1 \le 0$ is non-positive. Furthermore, the constraint $z_t \le \log(\sum_n W_{tn}^{\sigma_{tn}})$ is convex as it represents the hypograph of a concave function, and the affine equalities defining $W_{tn}$ serve to preserve this convexity.

Thus, for any fixed binary vector $x$, the continuous relaxation is convex, and the overall model constitutes a MICP. 
The bilinear terms $x_i t_{tn}$ in the objective are exactly linearized using \emph{McCormick envelopes} without introducing any relaxation gap since $x_i$ is binary.
% :
% \begin{align*}
% s_{tin} &\ge t^L_{tn}x_i, \quad s_{tin} \le t^U_{tn}x_i,\\
% s_{tin} &\ge t_{tn} - (1-x_i)t^U_{tn}, \quad s_{tin} \le t_{tn} - (1-x_i)t^L_{tn},
% \end{align*}
% where $t^L_{tn}$ and $t^U_{tn}$ are valid lower and upper bounds on $t_{tn}$, respectively. 
% Since $x_i$ is binary, these linear inequalities provide an exact representation of the bilinear term without introducing any relaxation gap.

Besides the convexity properties, analogous to the single-segment GNL case, we can also exploit the (super/sub)modularity of the transformed components. 
Let us define
\[
Y_{tn}(S) = (\sigma_{tn}-1)\log(W_{tn}(S)), \quad
Z_t(S) = \log\!\Bigg(\sum_{n\in[N]} W_{tn}(S)^{\sigma_{tn}}\Bigg).
\]
Proposition~\ref{prop:submod-yn-z} naturally extends to the multi-segment setting. 
Specifically, we can show that $Y_{tn}(S)$ is \emph{supermodular} and monotonically decreasing in $S$, whereas $Z_t(S)$ is \emph{submodular} and monotonically increasing. 
These structural properties enable the construction of valid SCs in addition to the standard OA cuts derived from convexity, thereby forming a hybrid OA+SC cutting-plane framework for the \eqref{prob: Assort--MGNL--Convex}. 
In the cutting-plane procedure, at each iteration, violated OA cuts and SCs are iteratively added to tighten the relaxation. 
Because all continuous components are convex and the McCormick linearization is exact, the overall algorithm is guaranteed to converge finitely to the global optimum.

\subsection{Bilinear-Convex Reformulation}
% - ExpCone Reformulation

The MICP derived via log-transformations in the previous section is practically tractable through the integration of OA and SC techniques. 
In this section, we introduce an alternative yet equally practical reformulation that preserves convexity and can be efficiently solved to optimality using modern solvers that support bilinear or MINLP. 
In particular, we show that the fractional structure in the MGNL objective can be equivalently represented through bilinear terms, resulting in a \emph{bilinear-convex} formulation.

Specifically, we consider the equivalent minimization problem under the MGNL model:
\begin{align}
\min_{\bx \in \cX } \quad 
&\left\{
\sum_{t\in [T]} \theta_t 
\frac{ \sum_{n\in [N]}{W_{tn}}^{\sigma_{tn}-1} 
\left(\beta V_{t0n} + \sum_{i\in [m]}\alpha_{tin}x_ir'_{ti}V_{tin}\right)}
{\sum_{n\in [N]} W_{tn}^{ \sigma_{tn}}}
\right\}
\label{prob:GNL-assort-2}\tag{\sf Assort-2}\\
\text{s.t.} \quad 
&W_{tn} = V_{t0n} + \sum_{i\in  [m]}\alpha_{tin}x_iV_{tin}, \quad \forall\, t \in [T], n \in [N]. \nonumber
\end{align}
Here, $\beta$ is chosen such that $\beta > \max_{t,i} r_{ti}$ is an upper bound on the product revenues across all customer segments, introduced to ensure nonnegativity in the reformulated objective. 
The transformed revenue parameter $r'_{ti} = \beta - r_{ti}$ facilitates the conversion of the maximization into a minimization problem while preserving optimality. 
The decision variables $x_i$ represent product inclusion decisions, and $W_{tn}$ denotes the inclusive value (total preference weight) of nest $n$ for customer segment $t$ under the assortment $\bx$. 

The fractional structure of~\eqref{prob:GNL-assort-2} couples nonlinear numerator and denominator terms across all nests within each customer segment. 
To simplify this dependency, we write the problem as
\begin{align}
\min_{\bx \in \cX } \quad 
&\left\{
\sum_{t\in [T]} \theta_t \delta_t
\right\}
\label{prob:GNL-assort-2}\tag{\sf Assort-2}\\
\text{s.t.} \quad& \delta_t \geq \frac{ \sum_{n\in [N]}{W_{tn}}^{\sigma_{tn}-1} 
\left(\beta V_{t0n} + \sum_{i\in [m]}\alpha_{tin}x_ir'_{ti}V_{tin}\right)}
{\sum_{n\in [N]} W_{tn}^{ \sigma_{tn}}}\nonumber \\
&W_{tn} = V_{t0n} + \sum_{i\in  [m]}\alpha_{tin}x_iV_{tin}, \quad \forall\, t \in [T], n \in [N]. \nonumber
\end{align}
By denoting $H_{tn}(\bx) = W_{tn}(\bx)^{\sigma_{tn}-1}$  and 
$K_{tn}(\bx) = W_{tn}(\bx)^{\sigma_{tn}}$ and introducing auxiliary variables $(h_{tn}, k_{tn})$ to represent these nonlinear functions,  we can rewrite~\eqref{prob:GNL-assort-2} as the following mixed-integer nonlinear problem:
\begin{align}
\min_{\bx\in \cX} \quad 
&\sum_{t \in [T]}\theta_t \delta_t
\label{prob:M-Bilinear}\tag{\sf MGNL-Bi}\\
\text{s.t.} \quad 
& \sum_{n \in [N]} h_{tn}\left(\beta V_{t0n} + \sum_{i \in [m]} \alpha_{tin} x_i r'_{ti} V_{tin}\right) 
= \delta_t \left(\sum_{n \in [N]} k_{tn}\right), 
\quad \forall\, t \in [T], \nonumber\\
& W_{tn} = V_{t0n} + \sum_{i\in  [m]}\alpha_{tin}x_iV_{tin}, 
\quad \forall\, t \in [T], n \in [N], \nonumber\\
& k_{tn} = W_{tn} h_{tn}, 
\quad \forall\, t \in [T], n \in [N], \label{ctr:bi_h_k}\\
& h_{tn} \geq  H_{tn}(\bx),
\quad \forall\, t \in [T], n \in [N], \label{ctr:htn}\\
& k_{tn} \leq K_{tn}(\bx), 
\quad \forall\, t \in [T], n \in [N], \label{ctr:ktn}\\
%& \bA \bx + \bB \leq \bC, \nonumber\\
& \bx \in \{0,1\}^{[m]},\;
\bW , \bh , \bk \in \mathbb{R}_+^{[T]\times[N]},\;
\delta \in \mathbb{R}_+^{[T]}.\nonumber
\end{align}
We note that the equalities $h_{tn} = H_{tn}(\bx)$ and $k_{tn} = K_{tn}(\bx)$ can be safely relaxed into nonlinear inequality constraints because, at optimality, these inequalities always bind - any feasible solution $(h_{tn}, k_{tn})$ satisfying $h_{tn} \ge H_{tn}(\bx)$ and $k_{tn} \le K_{tn}(\bx)$ can be tightened to equality without loss of optimality. 
We also introduce the bilinear constraints \eqref{ctr:bi_h_k} to explicitly capture the functional relationship between the two nonlinear components $H_{tn}(\bx)$ and $K_{tn}(\bx)$. 
Although this bilinear constraint is theoretically redundant in the full MINLP model—since both sides coincide at the true solution—it can  strengthen the relaxation in practice when the nonlinear equalities are approximated by OA cuts and SCs.  The inclusion of this bilinear link improves the tightness of the continuous relaxation and accelerates convergence within B\&C frameworks.

It is worth emphasizing that, apart from the bilinear terms, all other nonlinear constraints in~\eqref{prob:M-Bilinear} are either convex or admit submodular representations.  
Specifically, for each $(t,n)$ pair, and analogous to Proposition~\ref{pro:convex_H_K}, the function $H_{tn}(\bx)$ can be shown to be convex in $\bx$, while $K_{tn}(\bx)$ is concave.  
Moreover, when these functions are defined over a product set $S \subseteq [m]$, the corresponding set functions 
\[
H_{tn}(S) = \left(V_{t0n} + \sum_{i\in S}\alpha_{tin}V_{tin}\right)^{\sigma_{tn}-1}, 
\qquad
K_{tn}(S) = \left(V_{t0n} + \sum_{i\in S}\alpha_{tin}V_{tin}\right)^{\sigma_{tn}},
\]
inherit supermodular and submodular properties, respectively (see Proposition~\ref{prop:submod-yn-z}).  
These curvature and set-function properties enable the use of both OA cuts and SCs to construct tight linear approximations of the nonlinear constraints~\eqref{ctr:htn} and~\eqref{ctr:ktn}.  
Such cuts iteratively refine the relaxation of~\eqref{prob:M-Bilinear}, ensuring convergence to the global optimum while maintaining computational tractability.

In summary, the overall problem~\eqref{prob:M-Bilinear} is \emph{bilinear-convex}: all nonconvexities arise solely from bilinear terms, which admit tight convex relaxations.  
Compared with the log-transformed MICP formulation, this bilinear-convex reformulation is more compact and can be directly handled by modern solvers that support bilinear programming.  Specifically, one can design a B\&C procedure in which, at each iteration, linear cuts are generated to approximate the nonlinear constraints~\eqref{ctr:htn} and~\eqref{ctr:ktn}.  
The resulting master problem, which is a bilinear program, is then solved to obtain a new candidate solution.  
If the candidate solution satisfies the original nonlinear constraints (within a prescribed tolerance~$\epsilon$), the algorithm terminates; otherwise, new cuts derived from the current solution are added to the master problem, and the process repeats.  
Thanks to the convexity and submodularity properties discussed above, all generated cuts are valid and progressively tighter approximations of the true feasible region.  
Moreover, if the master problem is solved exactly to optimality at each iteration, the procedure is guaranteed to converge to a globally optimal solution in a finite number of iterations.

It is important to note that the OA cuts and SCs generated in our algorithm play distinct geometric roles. 
Each OA corresponds to a supporting hyperplane of the convex epigraph of the nonlinear constraints and therefore provides a valid but generally non-facet-defining approximation of the convex envelope.  
In contrast, the SCs, which exploit the submodularity of the functions $H_{tn}(S)$ and $K_{tn}(S)$, correspond to facet-defining inequalities of the associated submodular polyhedron \citep{Edmonds2003}.  
Hence, while OA cuts ensure global convex validity, SCs enrich the relaxation with facet-defining inequalities that capture the combinatorial structure of the assortment decisions.  
Together, these two families of cuts progressively reconstruct the facet structure of the convex hull of feasible solutions to the MGNL assortment problem.

\section{JAP with Discrete Prices}\label{ec:JAP_DP}
The case with discrete prices represents a common and practically relevant setting widely adopted in retail and revenue management applications.  
In practice, firms often consider only a limited number of price points for each product—typically those following standard psychological or marketing conventions such as prices ending in \emph{.99}~\$—rather than treating prices as fully continuous decision variables.  
This discretization not only reflects real-world managerial practice but also offers significant computational advantages.  
From a modeling perspective, the discrete-price setting allows the JAP problem to be conveniently reformulated as a standard assortment optimization problem defined over an \emph{extended product set}.  
In this reformulation, each product--price combination is treated as a distinct ``virtual'' item characterized by its own utility, revenue, and nest membership. Thus, this transformation preserves the original structure of the assortment optimization problem, allowing the techniques developed in the previous sections to be directly applied.  
In the following, we briefly present the formulation of the JAP problem with discrete prices for the single customer segment case, while noting that the extension to multiple customer segments can be carried out in a straightforward manner.  
% A detailed discussion on how this formulation can be reformulated as a convex or bilinear-convex program is provided in the Appendix.

We now present the mathematical formulation of the JAP problem with discrete prices for a single customer segment.  
The extension to multiple customer segments can be carried out straightforwardly by introducing segment-specific parameters and aggregating expected revenues as shown in the previous sections.  
Let each product $i \in [m]$ be offered at a finite set of admissible discrete price levels denoted by 
$\mathcal{P}_i = \{p_{i1}, p_{i2}, \ldots, p_{iL_i}\},$
where $p_{il}$ represents the $l$-th feasible price point for product $i$ (e.g., price levels ending in \emph{.99}~\$). 
Each combination of product and price $(i,l)$ is treated as a distinct ``virtual'' item characterized by its own utility, price, and revenue.  We define the deterministic utility component of product $i$ offered at price $p_{il}$ as
${v}_{il} = \kappa_i - \eta_i p_{il},$
where $\kappa_i$ represents the intrinsic attractiveness of product $i$, and $\eta_i > 0$ denotes the price sensitivity of the customer segment.  The corresponding revenue for this virtual item is $r_{il} = p_{il}$.  
Let $x_{il} \in \{0,1\}$ be the binary decision variable that equals~1 if product~$i$ is offered at price~$p_{il}$ and~0 otherwise.  
To ensure that each product can be offered at most one price level, we impose
\[
\sum_{l \in \mathcal{P}_i} x_{il} \le 1, \quad \forall i \in [m].
\]
Now, under the GNL model, the inclusive value of nest $n \in [N]$ is given by
\[
W_{n}(\bx) = V_{0n} + \sum_{i\in[m]} \sum_{l\in\mathcal{P}_i} \alpha_{in} V_{iln}  x_{il},
\]
where $V_{iln} = e^{v_{il}/\sigma_{n}}$.
The corresponding choice probability of the product--price pair $(i,l)$ is then given by
\[
P_{il}(x) = 
\sum_{n\in[N]}
\frac{\alpha_{in} V_{iln} W_{n}(\bx)^{\sigma_{n}-1}}
{\sum_{n'\in[N]} W_{n'}(\bx)^{\sigma_{n'}}}.
\]
The expected revenue under assortment and pricing decisions $\bx$ can thus be written as
\[
F(\bx) = 
\sum_{i\in[m]}\sum_{l\in\mathcal{P}_i} r_{il} P_{il}(\bx)
=
\frac{
\sum_{i\in[m]} \sum_{l\in\mathcal{P}_i} \sum_{n\in[N]} 
\alpha_{in} V_{iln} r_{il} W_n(x)^{\sigma_n - 1} x_{il}
}{
\sum_{n\in[N]} W_n(\bx)^{\sigma_n}
}.
\]
Hence, the JAP problem with discrete prices for a single customer segment can be formulated as follows:
\begin{align}
\max_{\bx \in \mathcal{X}} \quad &
F(\bx)
= 
\frac{
\sum_{i\in[m]} \sum_{l\in\mathcal{P}_i} \sum_{n\in[N]} 
\alpha_{in} V_{iln} r_{il} W_n(\bx)^{\sigma_n - 1} x_{il}
}{
\sum_{n\in[N]} W_n(\bx)^{\sigma_n}
} \nonumber\\
% \label{prob:JAP-Discrete-Single}\\
\text{s.t.} \quad &
W_{n}(\bx) = V_{0n} + \sum_{i\in[m]} \sum_{l\in\mathcal{P}_i} \alpha_{in} V_{iln} x_{il}, \quad \forall n\in[N], \nonumber\\
& \sum_{l\in\mathcal{P}_i} x_{il} \le 1, \quad \forall i \in [m], \nonumber\\
& x_{il} \in \{0,1\}, \quad \forall i\in[m],\, l\in \mathcal{P}_i. \nonumber
\end{align}
This discrete-price formulation preserves the structural characteristics of the standard assortment problem, with each product--price pair $(i,l)$ treated as a separate item in an extended item set
$
\tilde{\mathcal{M}} = \{(i,l): i\in[m],\, l\in\mathcal{P}_i\}.
$
Consequently, the convexification and bilinear-convex reformulation techniques developed in Sections~\ref{sec:convexification_gnl} and ~\ref{ec:mixed_gnl_model} can be directly applied to this discrete-price setting after appropriate indexing and parameter substitution.  
% A detailed derivation of the convex and bilinear-convex reformulations for~\eqref{prob:JAP-Discrete-Single} is provided in Appendix~\ref{sec:assort-reformulation}.

\section{Construction of Optimal Breakpoints for PWLA in JAP with Continuous Prices}
\label{ec:PWLA}
To achieve a desired approximation accuracy $\epsilon$ while minimizing the number of segments $H$, 
we determine the optimal breakpoint vector $\bq_{in}$ such that
\[
\max_{w_{in} \in [\underline{w}_{in}, \overline{w}_{in}]}
\big|
\cPA(\exp(w_{in}) \,|\, \underline{w}_{in}, \overline{w}_{in}, \bq_{in})
- e^{w_{in}}
\big| \le \epsilon.
\]
To this end, we can perform the following search procedure. Starting from each breakpoint $q_h$, the next breakpoint $q_{h+1}$ is chosen to 
maximize the interval width while ensuring that the local approximation error 
does not exceed $\epsilon$.  
The approximation gap within $[q_h, q_{h+1}]$ is given by:
\[
\phi(w_{in}) = e^{q_h} 
+ (w_{in} - q_h)
  \frac{e^{q_{h+1}} - e^{q_h}}{q_{h+1} - q_h}
- e^{w_{in}},
\]
and the worst-case deviation is obtained from
\[
\max_{w_{in} \in [q_h, q_{h+1}]} \phi(w_{in}).
\]
The function $\phi(w_{in})$ is concave in $w_{in}$; therefore, its maximum is attained 
by setting the derivative of $\phi(w_{in})$ with respect to $w_{in}$ to zero, yielding:
\[
w^* = \text{argmax}_{w_{in} \in [q_h, q_{h+1}]} \phi(w_{in})
= \ln\!\left(
\frac{e^{q_{h+1}} - e^{q_h}}{q_{h+1} - q_h}
\right).
\]
and the corresponding maximum gap within the interval $[q_h, q_{h+1}]$ can be computed as:
\[
\theta(q_{h+1}) 
= e^{q_h} 
+ \left(
  \ln\!\left(
  \frac{e^{q_{h+1}} - e^{q_h}}{q_{h+1} - q_h}
  \right) - q_h - 1
  \right)
  \frac{e^{q_{h+1}} - e^{q_h}}{q_{h+1} - q_h}.
\]
The next breakpoint $q_{h+1}$ is thus determined as:
\[
q_{h+1} = 
\arg\max_{t > q_h}
\left\{
t ~\big|~ \theta(t) \le \epsilon
\right\}.
\]
We can now observe that the function 
$\theta(t) = \max_{w_{in} \in [q_h,\, t]} \phi(w_{in})$ 
is monotonically increasing in $t$. 
Therefore, the next breakpoint $q_{h+1}$ can be efficiently determined 
using a binary search procedure, as described below:

\begin{mdframed}[linewidth=1pt, roundcorner=5pt, backgroundcolor=gray!10]
\textbf{Binary Search Procedure:}
\begin{itemize}
    \item \textbf{Step 1:} Initialize $l = q_h$, $u = \overline{w}_{in}$, and tolerance $\tau > 0$.
    \item \textbf{Step 2:} If $\theta(u) \le \epsilon$, set $q_{h+1} = \overline{w}_{in}$ and terminate.
    \item \textbf{Step 3:} Compute $w = (u + l)/2$. 
          If $\theta(w) \le \epsilon$, set $l = w$; otherwise, set $u = w$.
    \item \textbf{Step 4:} If $|u - l| \le \tau$, set $q_{h+1} = l$ and terminate; otherwise, repeat Step~3.
\end{itemize}
\end{mdframed}

This algorithm converges exponentially to the optimal breakpoint 
satisfying $\theta(q_{h+1}) \le \epsilon$ within tolerance $\tau$.
The complete PWLA construction for $\exp(w_{in})$ proceeds iteratively as follows:

\begin{mdframed}[linewidth=1pt, roundcorner=5pt, backgroundcolor=gray!10]
\textbf{Constructing Breakpoints for $\cPA(\exp(w_{in}) \,|\, \underline{w}_{in}, \overline{w}_{in}, \bq_{in})$:}
\begin{itemize}
    \item Initialize $q_1 = \underline{w}_{in}$.
    \item Use the binary search algorithm to compute $q_{h+1}$ from $q_h$.
    \item Terminate when $q_{h+1} = \overline{w}_{in}$.
\end{itemize}
\end{mdframed}

Since the maximum approximation error is optimized in each sub-interval, 
this procedure yields the \emph{minimum number of breakpoints} $H$ 
required to achieve the target accuracy $\epsilon$.  
In practice, selecting a sufficiently small tolerance $\tau$ 
ensures near-optimal precision, and due to the exponential convergence 
of the binary search, only a few iterations are typically required 
even for very tight error bounds.
\begin{figure}[htb]
    \centering
    \includegraphics[width=0.55\textwidth]{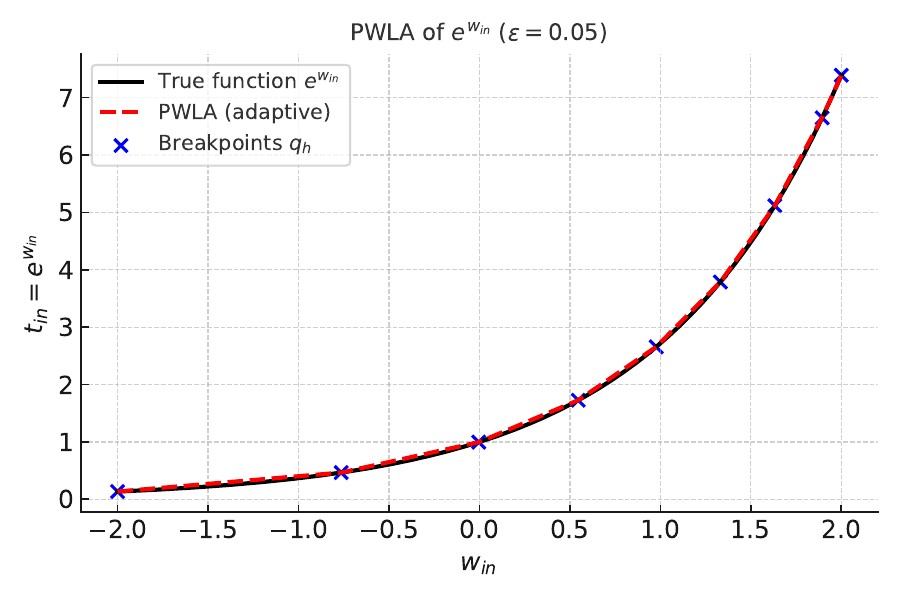}
    \caption{PWLA of the exponential function 
    $e^{w_{in}}$ for an error tolerance of $\epsilon = 0.05$.}
    \label{fig:pwla_exp_adaptive}
\end{figure}
In the following, we establish an explicit upper bound on the number of breakpoints 
generated by the adaptive PWLA construction procedure.

\begin{proposition}\label{prop:breakpoint}
Let $\epsilon > 0$ be the desired approximation accuracy, the number of breakpoints $H$ generated by the adaptive 
PWLA procedure described above satisfies
\[
    H 
    \;\le\;
    \Biggl\lceil
        \frac{
        e^{ \frac{-L_i\eta_i + \kappa_i}{2\sigma_n}}\, \eta_i (U_i - L_i)
        }{
        2\sqrt{2\epsilon}\, \sigma_n
        }
    \Biggr\rceil.
\]
Consequently, $H = \mathcal{O}\!\left( \frac{
        e^{ \frac{-L_i\eta_i + \kappa_i}{2\sigma_n}}\, \eta_i (U_i - L_i)
        }{
        2\sqrt{2\epsilon}\, \sigma_n
        }\right)$,
which shows that the number of breakpoints grows at most at the rate 
$1/\sqrt{\epsilon}$ and depends linearly on the effective scaled price range 
$\eta_i(U_i - L_i)/\sigma_n$.
\end{proposition}

% The above result provides an explicit dependence between the approximation accuracy 
% $\epsilon$ and the number of required breakpoints $H$.
% In particular, the quadratic relationship between the local segment length 
% and the error ($E_h \propto (q_{h+1} - q_h)^2$) 
% implies that halving $\epsilon$ increases $H$ only by a factor of $\sqrt{2}$.
% Thus, the adaptive PWLA achieves high accuracy with only a modest increase 
% in the number of linear pieces, making it computationally efficient even for tight error tolerances.
\begin{proof}
Consider the exponential function $f(w) = e^{w}$ 
and its secant approximation $\widehat{f}(w)$ constructed between two consecutive breakpoints 
$q_h$ and $q_{h+1}$.  
The maximum error between $f(w)$ and $\widehat{f}(w)$ over the interval 
$[q_h, q_{h+1}]$ is given by: 
$
    E_h = 
    \max_{w \in [q_h, q_{h+1}]} 
    \big| f(w) - \widehat{f}(w) \big|
$.

By Taylor’s theorem, for some $\xi \in (q_h, q_{h+1})$, 
the deviation between a convex twice-differentiable function and its secant line 
satisfies
\[
    f(w) - \widehat{f}(w) 
    = \tfrac{f''(\xi)}{8} (q_{h+1} - q_h)^2.
\]
Since $f''(w) = e^{w}$ for the exponential function, 
we have
\[
    E_h = \frac{e^{\xi}}{8}(q_{h+1} - q_h)^2,
    \qquad \text{for some } \xi \in [q_h, q_{h+1}].
\]
To ensure that the local approximation error does not exceed $\epsilon$, 
we require $E_h \le \epsilon$, which gives
\[
    q_{h+1} - q_h 
    \;\le\;
    2\sqrt{2\epsilon\, e^{-\xi}}
    \;\le\;
    2\sqrt{2\epsilon\, e^{-q_h}},
\]
where the last inequality uses $\xi \ge q_h$.  

We can now bound the total number of subintervals needed to cover $[L, U]$, where $L = \underline{w}_{in}$ and $U = \overline{w}_{in}$.
To obtain an upper bound on the total number of breakpoints, 
we sum the lengths of all subintervals generated by the adaptive procedure.
Since the union of all subintervals covers the entire range $[L, U]$, we have: 
$
    \sum_{h=1}^{H} (q_{h+1} - q_h) = U - L.
$

From the local error condition established above,
each interval length satisfies 
$(q_{h+1} - q_h) \le 2\sqrt{2\epsilon\, e^{-q_h}}$.
Taking the reciprocal and summing over all $h$ gives
\[
    H = \sum_{h=1}^{H} 1 
    \;=\;
    \sum_{h=1}^{H} 
    \frac{(q_{h+1} - q_h)}{(q_{h+1} - q_h)}
    \;\le\;
    \sum_{h=1}^{H}
    \frac{(q_{h+1} - q_h)}{2\sqrt{2\epsilon\, e^{-q_h}}}.
\]
Because the breakpoints are dense and ordered, 
the sum on the right-hand side can be approximated by the Riemann integral
\[
    H 
    \;\lesssim\;
    \int_{L}^{U} 
    \frac{dw}{2\sqrt{2\epsilon\, e^{-w}}}
    =
    \frac{1}{2\sqrt{2\epsilon}} 
    \int_{L}^{U} e^{w/2}\, dw.
\]
Evaluating the integral yields
\[
    H 
    \;\le\;
    \frac{e^{U/2} - e^{L/2}}{2\sqrt{2\epsilon}}.
\]
Finally, noting that $e^{U/2} - e^{L/2} \le e^{U/2}(U-L)$ for bounded $U$, 
we obtain a simpler and conservative bound:
\[
    H 
    \;\le\;
    \frac{e^{U/2}(U-L)}{2\sqrt{2\epsilon}}.
\]
Since $e^{U/2} - e^{L/2} \le e^{U/2}(U-L)$ for bounded $U$, 
we obtain the simplified bound
\[
    H 
    \;\le\;
    \frac{e^{U/2}(U-L)}{2\sqrt{2\epsilon}}.
\]
Substituting 
$L = \underline{w}_{in} = \tfrac{-U_i\eta_i + \kappa_i}{\sigma_n}$ 
and 
$U = \overline{w}_{in} = \tfrac{-L_i\eta_i + \kappa_i}{\sigma_n}$ 
into the bound above yields the desired expression.  
Finally, taking the ceiling ensures that the number of breakpoints $H$ is an integer, 
which completes the proof.
\end{proof}

\section{Reformulation for Models with Zero Opt-Out Values}\label{ec:zero_opt_out}
In the main paper, we  assume that $V_{0n} > 0$ for all $n \in [N]$, which means that an opt-out option is always available in each nest. This assumption aligns with the standard two-stage choice process in the GNL model: in the first stage, the customer selects a nest, and in the second stage, the customer either chooses one of the products in that nest or decides not to purchase anything. For example, in an online retail platform, after browsing a category such as ``electronics,'' a customer may still leave without purchasing any item, which is naturally captured by an opt-out option within each nest.

In contrast, some studies in the literature adopt a different setup in which no opt-out option is available within a nest \citep{Zhang2024}. In this case, once a customer selects a nest, they are required to purchase one of the offered products. Under our formulation, this corresponds to allowing $V_{0n} = 0$ for some (or even all) nests. When $V_{0n} = 0$, numerical issues may arise because $W_n$ can take the value zero, which makes $W_n^{\sigma_n - 1}$ undefined whenever $\sigma_n < 1$. This occurs precisely when a nest has no opt-out option and none of its products are included in the assortment, so that the total preference weight $W_n$ collapses to zero.

Under our formulation in  \eqref{prob:GNL-assort}, two remedies are possible. First, for any nest $n$ with $V_{0n} = 0$, one may introduce an artificial opt-out option with a very small utility value, which resolves the issue without substantially altering the objective function. 

Alternatively, the optimization problem can be equivalently reformulated 
to ensure that the inclusive values $W_n$ remain strictly positive. 
In the following, we establish a proposition that provides an alternative 
formulation of the assortment optimization problem, which eliminates the 
degeneracy caused by zero-valued $W_n$.
\begin{proposition}
Problem~\eqref{prob:GNL-assort} can be equivalently reformulated as the following:
\begin{align}
    \max_{\bx \in \cX} \quad 
    &F(\bx) =\frac{\sum_{n\in[N]} W_n^{\sigma_n-1}\Big(\sum_{i\in[m]}\alpha_{in}x_i r_i V_{in}\Big)}
{\sum_{n\in[N]} U_n^{\sigma_n}}
    \label{prob:assort-v0}\tag{\sf Assort-2} \\
    \text{subject to} \quad 
    & W_n \ge V_{0n} + \sum_{i \in [m]} \alpha_{in} x_i V_{in}, \quad \forall n \in [N], \nonumber \\
    & W_n \ge \delta_n, \quad \forall n \in [N] \text{ such that } V_{0n} = 0, \nonumber\\
    & U_n \ge V_{0n} + \sum_{i \in [m]} \alpha_{in} x_i V_{in}, \quad \forall n \in [N], \nonumber 
\end{align}
where each $\delta_n$ is chosen such that $\delta_n < \min_{i:\,\alpha_{in}>0} \alpha_{in} V_{in}$. This guarantees $W_n>0$ for all $n \in [N]$.
\end{proposition}
\begin{proof}
For clarity,  we define the following shorthand for any $\bx\in \cX$:
\begin{align*}
A_n(\bx) &:= \sum_{i \in [m]} \alpha_{in} x_i r_i V_{in} \;\;\; (\geq 0), \\
N(\bx,\bW) &:= \sum_{n \in [N]} W_n^{\sigma_n - 1} A_n(\bx), \\
D(\bU) &:= \sum_{n \in [N]} U_n^{\sigma_n}.
\end{align*}
With this notation, the expected revenue in the reformulated problem can be written as
\[
F(\bx,\bW, \bU) = \frac{N(\bx,\bW)}{D(\bU)}.
\]
First we see that, for each nest $n$ and parameter $\sigma_n \in (0,1]$, the following properties hold:
\begin{itemize}
\item The mapping $W_n \mapsto W_n^{\sigma_n-1}$ is nonincreasing (strictly decreasing if $\sigma_n<1$, constant if $\sigma_n=1$). Intuitively, larger $W_n$ values dilute the contribution of nest $n$ to the numerator.
\item The mapping $U_n \mapsto U_n^{\sigma_n}$ is strictly increasing. Thus larger $U_n$ values increase the denominator, making the objective smaller.
\end{itemize}
This simple monotonicity observation tells us that, to maximize $F(\bx)$, we should always keep $W_n$ as small as possible (since that increases the numerator or leaves it unchanged), and also keep $U_n$ as small as possible (since that decreases the denominator).

Now, let $(\bx,\bW,\bU)$ be any feasible solution of the reformulated problem. Define the ``tightened'' values
\[
\widehat W_n := \max\!\left\{V_{0n} + \sum_i \alpha_{in} x_i V_{in}, \;\; \delta_n \cdot \mathbf{1}_{\{V_{0n}=0\}}\right\}, 
\qquad
\widehat U_n := V_{0n} + \sum_i \alpha_{in} x_i V_{in}, 
\quad \forall n \in [N].
\]
By feasibility, we always have $\widehat W_n \leq W_n$ and $\widehat U_n \leq U_n$. By monotonicity, this tightening never decreases the numerator and never increases the denominator, so
\[
N(\bx,\widehat\bW) \geq N(\bx,\bW), 
\qquad
D(\widehat\bU) \leq D(\bU).
\]
Hence the tightened choice yields
\[
F(\bx,\widehat\bW,\widehat\bU) \;\geq\; F(\bx,\bW,\bU).
\]
Therefore, at any optimum, the inequalities can be assumed to be tight, i.e.,
\begin{equation}
W_n = \widehat W_n, 
\qquad 
U_n = V_{0n} + \sum_i \alpha_{in} x_i V_{in}, 
\quad \forall n \in [N].
\tag{$\ast$}
\end{equation}
Under $(\ast)$ we have two cases:
\begin{itemize}
\item If either $V_{0n}>0$ or at least one product $i$ with $\alpha_{in}>0$ is offered ($x_i=1$), then 
\[
W_n = U_n = V_{0n} + \sum_i \alpha_{in} x_i V_{in}.
\] 
In this case, the contributions of nest $n$ to both the numerator and denominator are exactly the same as in the original problem \eqref{prob:GNL-assort}.
\item If $V_{0n}=0$ and no product in nest $n$ is offered, then $U_n=0$ as in the original problem. In addition, $A_n(\bx)=0$, so nest $n$ contributes nothing to the numerator. In the reformulated problem we enforce $W_n \geq \delta_n > 0$ to avoid the undefined term $0 \cdot W_n^{\sigma_n-1}$ when $W_n=0$. Since the contribution is zero anyway, the objective value is unaffected.
\end{itemize}
Thus, for any $\bx$, the tightened $(\bW,\bU)$ reproduces exactly the same objective value as the original formulation.
Now we are ready to show the two problems have the same optimal value:
\begin{itemize}
\item Given any feasible $\bx$ in the original problem \eqref{prob:GNL-assort}, define 
\[
U_n = V_{0n} + \sum_i \alpha_{in} x_i V_{in},
\qquad
W_n = \max\!\{\,V_{0n} + \sum_i \alpha_{in} x_i V_{in},\; \delta_n \mathbf{1}_{\{V_{0n}=0\}} \}.
\]
This provides a feasible solution to the reformulated problem with the same objective value. Therefore, the reformulated problem is at least as good as the original.
\item Conversely, take any optimal solution $(\bx^\star,\bW^\star,\bU^\star)$ of the reformulated problem. Replace it by the tightened version $(\widehat\bW,\widehat\bU)$ as in $(\ast)$ without decreasing the objective. The resulting value coincides with that of the original problem for the same $\bx^\star$. Hence the original problem is at least as good as the reformulated one.
\end{itemize}
Combining both directions, the two problems are equivalent in terms of optimal value and optimal assortments.

Finally, note that if a product with $\alpha_{in}>0$ is included in nest $n$, then 
\[
V_{0n} + \sum_i \alpha_{in} x_i V_{in} \;\geq\; \alpha_{in} V_{in} \;>\; \delta_n,
\]
so the artificial floor $\delta_n$ never binds for nonempty nests. It only matters for empty nests with $V_{0n}=0$, where it ensures $W_n>0$ and the objective remains well-defined. Thus $\delta_n$ does not affect the solution but guarantees numerical safety.
\end{proof}

In the next step, in order to apply the techniques developed in the main body of the paper
(e.g., bilinear--convex or convex reformulations),
we need to convert the maximization problem~\eqref{prob:assort-v0}
into an equivalent minimization problem.
This conversion is not immediate, since the numerator and denominator in~\eqref{prob:assort-v0}
depend on different auxiliary variables.
Nevertheless, the transformation remains valid due to the following observation:
at optimality, $U_n = W_n$ whenever
\[
V_{0n} + \sum_{i} \alpha_{in} x_i V_{in} > 0,
\]
whereas $U_n < W_n$ may occur only when the inclusive value of nest $n$ is zero.
In the former case, $U_n$ can be replaced by $W_n$ without affecting the optimal solution.
We formalize this argument in the following proposition.
\begin{proposition}
\label{prop:assort-v0-min}
Problem~\eqref{prob:assort-v0} can be equivalently reformulated as the following minimization problem:
\begin{align}
\min_{\bx \in \cX} \quad 
&\frac{\sum_{n\in[N]} W_n^{\sigma_n-1}
\Big(\beta V_{0n}+\sum_{i\in[m]}\alpha_{in}x_i r'_i V_{in}\Big)}
{\sum_{n\in[N]} U_n^{\sigma_n}}
\label{prob:assort-v0-min}\tag{\sf Assort-2} \\
\text{s.t.} \quad 
& W_n \ge V_{0n} + \sum_{i \in [m]} \alpha_{in} x_i V_{in},
\quad \forall n \in [N], \nonumber \\
& W_n \ge \delta_n,
\quad \forall n \in [N] \text{ such that } V_{0n} = 0, \nonumber\\
& U_n \ge V_{0n} + \sum_{i \in [m]} \alpha_{in} x_i V_{in},
\quad \forall n \in [N], \nonumber
\end{align}
where $\beta>\max_{i\in[m]} r_i$ and $r'_i := \beta - r_i > 0$ for all $i\in[m]$.
\end{proposition}
\begin{proof}
Fix any assortment decision $\bx \in \cX$ and define
\[
W_n^\star := V_{0n} + \sum_{i\in[m]} \alpha_{in} x_i V_{in},
\quad \forall n \in [N].
\]
We analyze the objective value of the reformulated problem for this fixed $\bx$.

First, note that for each $n$ and $\sigma_n \in (0,1]$,
the mapping $W_n \mapsto W_n^{\sigma_n-1}$ is nonincreasing,
while the mapping $U_n \mapsto U_n^{\sigma_n}$ is strictly increasing.
Therefore, for a fixed $\bx$, the objective function is minimized by choosing
$W_n$ and $U_n$ as small as possible subject to feasibility.

We consider the following two cases.

\emph{Case 1: $W_n^\star > 0$.}
In this case, since $W_n^\star \ge \delta_n$, both constraints
$W_n \ge W_n^\star$ and $U_n \ge W_n^\star$
are tight at optimality, implying
\[
W_n = U_n = W_n^\star.
\]
Hence, the contribution of nest $n$ to both the numerator and denominator
coincides exactly with that in the original formulation~\eqref{prob:GNL-assort}.

\emph{Case 2: $W_n^\star = 0$.}
This can occur only when $V_{0n}=0$ and no product assigned to nest $n$ is selected.
In this case, the numerator contribution of nest $n$ is zero regardless of the value of $W_n$.
The additional constraint $W_n \ge \delta_n > 0$ ensures that $W_n^{\sigma_n-1}$ is well defined, 
while $U_n$ only appears in the denominator and $U_n = 0$. 
Thus, the contribution of such a nest does not affect the objective value, 
and we can replace $U_n^{\sigma_n}$ by 
$W_n^{\sigma_n-1}\big(V_{0n} + \sum_{i\in[m]}\alpha_{in} x_i r_i V_{in}\big)$ 
without affecting the objective function.

With the above observations, we can rewrite the objective of the maximization problem as
\begin{align*}
F(\bx)
&= \beta -
\frac{
\sum_{n\in[N]}
\beta U_n^{\sigma_n}
-
\sum_{n\in[N]}
W_n^{\sigma_n-1}
\Big(\sum_{i\in[m]}\alpha_{in}x_i r_i V_{in}\Big)
}{
\sum_{n\in[N]} U_n^{\sigma_n}
}\\
&=
\beta -
\frac{
\sum_{n\in[N]}
W_n^{\sigma_n-1}
\Big(\beta V_{0n} + \sum_{i\in[m]}\alpha_{in}x_i r'_i V_{in}\Big)
}{
\sum_{n\in[N]} U_n^{\sigma_n}
}.
\end{align*}
Since $\beta$ is a constant, maximizing $F(\bx)$ is equivalent to minimizing
\[
\frac{
\sum_{n\in[N]}
W_n^{\sigma_n-1}
\Big(\beta V_{0n} + \sum_{i\in[m]}\alpha_{in}x_i r'_i V_{in}\Big)
}{
\sum_{n\in[N]} U_n^{\sigma_n}
},
\]
which completes the proof.
\end{proof}
\subsection{Bilinear--Convex and Convex Reformulations}
\label{sec:assort-reformulation}

In this section, we develop tractable reformulations for the minimization problem
\eqref{prob:assort-v0-min}, which arises from the equivalent transformation of the
original assortment optimization problem.
Our goal is to cast the problem into forms that can be solved to global optimality
using modern mixed-integer optimization techniques.
Specifically, we present (i) a bilinear--convex reformulation based on bilinear
auxiliary variables, and (ii) a convex reformulation obtained via logarithmic
transformations, extending the techniques developed in the main body of the paper.

\paragraph{Bilinear--Convex Reformulation}
\label{subsec:assort-biconvex}

Recall the minimization problem
\eqref{prob:assort-v0-min}:
\begin{align*}
\min_{\bx \in \cX} \quad 
&\frac{\sum_{n\in[N]} W_n^{\sigma_n-1}
\Big(\beta V_{0n}+\sum_{i\in[m]}\alpha_{in}x_i r'_i V_{in}\Big)}
{\sum_{n\in[N]} U_n^{\sigma_n}}.
\end{align*}
The main difficulty lies in the coupled nonlinear terms
$W_n^{\sigma_n-1}$ and $U_n^{\sigma_n}$, as well as their interaction with the
binary assortment variables $\bx$. To decouple the numerator and denominator, we introduce auxiliary variables $h_n$ and $k_n$ to present  the terms  $W_n^{\sigma_n-1}$ and  $U_n^{\sigma_n}$, respectively.
We then consider the following equivalent formulation:
\begin{align}
\min_{\bx,\bW,\bU,\bh,\bk,\delta} \quad 
& \delta \label{prob:assort-biconvex} \\
\text{s.t.}\quad 
& \sum_{n\in[N]} h_n
\Big(\beta V_{0n}+\sum_{i\in[m]}\alpha_{in}x_i r'_i V_{in}\Big)
\;\le\;
\delta \sum_{n\in[N]} k_n, \nonumber\\
& W_n \ge V_{0n} + \sum_{i\in[m]}\alpha_{in}x_i V_{in}, \quad \forall n\in[N], \nonumber\\
& U_n \ge V_{0n} + \sum_{i\in[m]}\alpha_{in}x_i V_{in}, \quad \forall n\in[N], \nonumber\\
& W_n \ge \delta_n,
\qquad \forall n\in[N]\text{ such that }V_{0n}=0, \nonumber\\
& h_n \ge W_n^{\sigma_n-1}, \quad \forall n\in[N], \nonumber\\
& k_n \le U_n^{\sigma_n}, \quad \forall n\in[N], \nonumber\\
& \bx\in\{0,1\}^m,\;
\bW,\bU,\bh,\bk \in \mathbb{R}_+^N,\;
\delta\in\mathbb{R}_+.
\nonumber
\end{align}
We note that the inequalities $h_n \ge W_n^{\sigma_n-1}$ and $k_n \le U_n^{\sigma_n}$
constitute valid relaxations, since both constraints are tight at optimality.
Moreover, because $W_n^{\sigma_n-1}$ is convex and nonincreasing, while
$U_n^{\sigma_n}$ is concave and nondecreasing for $\sigma_n \in (0,1]$,
the constraints $h_n \ge W_n^{\sigma_n-1}$ and $k_n \le U_n^{\sigma_n}$
define convex feasible sets and can be efficiently handled within a
cutting-plane framework.

Importantly, all nonconvexities in
\eqref{prob:assort-biconvex} arise only through bilinear terms involving the
binary variables $\bx$ and continuous auxiliary variables.
Such bilinearities can be handled either by exact McCormick linearizations
or directly by solvers that support mixed-integer bilinear programming.
As a result, \eqref{prob:assort-biconvex} constitutes a
\emph{bilinear--convex} optimization problem.

\paragraph{Convex Reformulation via Log Transformation}
\label{subsec:assort-logconvex}
We now show that, analogous to Section~\ref{sec:convexification_gnl}, the minimization problem
\eqref{prob:assort-v0-min} can be transformed into a MICP via a logarithmic transformation. The resulting formulation can be solved to global optimality using standard
B\&C methods. Recall the minimization problem:
\begin{align*}
\min_{\bx \in \cX} \quad 
&\frac{\sum_{n\in[N]} W_n^{\sigma_n-1}
\Big(\beta V_{0n}+\sum_{i\in[m]}\alpha_{in}x_i r'_i V_{in}\Big)}
{\sum_{n\in[N]} U_n^{\sigma_n}}.
\end{align*}
To convexify this structure, we introduce logarithmic auxiliary variables. Specifcally, 
for each nest $n\in[N]$, define:
$
y_n := \log\!\big(W_n^{\sigma_n-1}\big)$ and 
$
z := \log\!\Big(\sum_{n\in[N]} U_n^{\sigma_n}\Big).
$
Using these definitions, the ratio term can be written as
\[
\frac{W_n^{\sigma_n-1}}{\sum_{n'} U_{n'}^{\sigma_{n'}}}
\;=\;
\exp(y_n - z).
\]
Similar to Section~\ref{sec:convexification_gnl}, 
the assortment problem can be equivalently written as:
\begin{equation}\label{prob:logConvex}
\begin{aligned}
\min_{\bx,\bt, \bW,\bU,\by,z,\delta}
\quad &   \sum_{n\in[N]} \Big(\beta V_{0n}+\sum_{i\in[m]}\alpha_{in}x_i r'_i V_{in}\Big)
t_n \\[4pt]
\text{s.t.}\quad
&t_n \geq
\exp(y_n - z), \\[4pt]
& y_n \ge (\sigma_n-1)\log(W_n),
\qquad \forall n\in[N], \\[4pt]
& z \le \log\!\Big(\sum_{n\in[N]} U_n^{\sigma_n}\Big), \\[4pt]
& W_n \ge V_{0n} + \sum_{i\in[m]}\alpha_{in}x_i V_{in},
\qquad \forall n\in[N], \\[4pt]
& W_n \ge \delta_n,
\qquad \forall n\in[N]\text{ such that }V_{0n}=0, \\[4pt]
& U_n \ge V_{0n} + \sum_{i\in[m]}\alpha_{in}x_i V_{in},
\qquad \forall n\in[N], \\[4pt]
& \bx \in \{0,1\}^m,\;
\bW,\bU \in \mathbb{R}_+^N,\;
\by \in \mathbb{R}^N,\;
z\in \mathbb{R}.
\end{aligned}
\tag{\sf Assort--LogConvex}
\end{equation}
In the above formulation, the variables $t_n$ represent the normalized
fractional terms
\(
W_n^{\sigma_n-1}/\sum_{n'} U_{n'}^{\sigma_{n'}}
\),
while the exponential constraint
$t_n \ge \exp(y_n - z)$
enforces this relationship in epigraph form.
The bilinear terms $x_i t_n$ appearing in the objective function can be
exactly linearized using standard McCormick inequalities, since $x_i$ is
binary and $t_n$ is continuous and bounded.

After this linearization, all remaining constraints and objective components
are convex.
Specifically, the exponential epigraph constraints are convex because
$\exp(\cdot)$ is convex and nondecreasing;
the constraints
$y_n \ge (\sigma_n-1)\log(W_n)$
define convex feasible sets since $\log(\cdot)$ is concave and
$\sigma_n-1 \le 0$;
and the constraint
$z \le \log(\sum_n U_n^{\sigma_n})$
is convex as it represents the hypograph of a concave function.
All other constraints are linear. Consequently, the resulting formulation is a MICP which can be solved to global optimality using a single B\&C framework, where violated convex constraints are
iteratively approximated by linear cuts and integrality is enforced
through branching on the assortment variables $\bx$.

\section{Additional Experiments}
\subsection{Assortment Optimization under GNL Model: Bisection, Convex and Bilinear--Convex Reformulations}\label{ec:assort-GNL}
\begin{table}[htb]\footnotesize
\centering
\resizebox{\textwidth}{!}{%
\begin{tabular}{rrrrrrrrrrrrrrrrrrrrr}
\toprule
 &  & \multicolumn{7}{c}{MILP-based approximation} &  & \multicolumn{2}{c}{\multirow{2}{*}{\texttt{SCIP}}} &  & \multicolumn{8}{c}{B\&C} \\ \cmidrule{3-9} \cmidrule{14-21} 
 &  & \multicolumn{3}{c}{Bisection (99\%)} &  & \multicolumn{3}{c}{MILP (99\%)} &  & \multicolumn{2}{c}{} &  & \multicolumn{2}{c}{\texttt{BIS}} &  & \multicolumn{2}{c}{\texttt{Convex}} &  & \multicolumn{2}{c}{\texttt{BiCo}} \\ \cmidrule{3-5} \cmidrule{7-9} \cmidrule{11-12} \cmidrule{14-15} \cmidrule{17-18} \cmidrule{20-21} 
$(N,m)$ &  & \#O/\#S & Time(s) & Gap(\%) &  & \#O/\#S & Time(s) & Gap(\%) &  & \#O & Time(s) &  & \#O & Time(s) &  & \#O & Time(s) &  & \#O & Time(s) \\ \midrule
(2,20) &  & 2/5 & 0.86 & -0.60 &  & 2/5 & 0.41 & -0.60 &  & 5 & 0.46 &  & 5 & 0.68 &  & 5 & 0.13 &  & 5 & \textbf{0.07} \\
(2,30) &  & 0/5 & 2.40 & -1.53 &  & 0/5 & 0.67 & -1.38 &  & 5 & 1.13 &  & 5 & 2.51 &  & 5 & 0.49 &  & 5 & \textbf{0.08} \\
(2,50) &  & 1/5 & 14.04 & -0.67 &  & 1/5 & 2.00 & -0.58 &  & 5 & 27.31 &  & 5 & 37.97 &  & 5 & 96.48 &  & 5 & \textbf{0.09} \\
(2,80) &  & 0/5 & 40.84 & -1.74 &  & 0/5 & 3.04 & -2.17 &  & 5 & 380.59 &  & 4 & 1075.29 &  & 4 & 1152.81 &  & 5 & \textbf{0.13} \\
(2,100) &  & 0/5 & 65.84 & -1.78 &  & 0/5 & 8.38 & -1.29 &  & 4 & 1719.55 &  & 5 & 1064.17 &  & 1 & 2923.49 &  & 5 & \textbf{0.18} \\ \midrule
(5,20) &  & 4/5 & 0.47 & -0.08 &  & 4/5 & 0.45 & -0.08 &  & 5 & 0.44 &  & 5 & 0.49 &  & 5 & 0.22 &  & 5 & \textbf{0.07} \\
(5,30) &  & 3/5 & 1.47 & -0.35 &  & 3/5 & 0.95 & -0.35 &  & 5 & 2.36 &  & 5 & 1.38 &  & 5 & 2.41 &  & 5 & \textbf{0.10} \\
(5,50) &  & 2/5 & 5.65 & -0.58 &  & 2/5 & 2.36 & -0.58 &  & 5 & 133.32 &  & 5 & 99.96 &  & 5 & 550.22 &  & 5 & \textbf{0.11} \\
(5,80) &  & 3/5 & 42.46 & -0.14 &  & 1/5 & 5.77 & -0.21 &  & 3 & 1805.76 &  & 5 & 975.28 &  & 1 & 2966.57 &  & 5 & \textbf{0.14} \\
(5,100) &  & 1/5 & 37.54 & -0.28 &  & 1/5 & 9.88 & -0.48 &  & 0 & - &  & 3 & 1860.33 &  & 0 & - &  & 5 & \textbf{0.18} \\ \midrule
(10,20) &  & 4/5 & 0.25 & -0.01 &  & 5/5 & 0.40 & 0.00 &  & 5 & 0.49 &  & 5 & 0.56 &  & 5 & 3.75 &  & 5 & \textbf{0.09} \\
(10,30) &  & 3/5 & 0.46 & -0.14 &  & 3/5 & 1.13 & -0.14 &  & 5 & 1.01 &  & 5 & 1.25 &  & 5 & 115.32 &  & 5 & \textbf{0.16} \\
(10,50) &  & 4/5 & 1.68 & 0.00 &  & 3/5 & 12.12 & -0.03 &  & 5 & 11.00 &  & 5 & 24.06 &  & 3 & 1859.90 &  & 5 & \textbf{0.21} \\
(10,80) &  & 2/5 & 10.97 & -0.09 &  & 2/5 & 49.87 & -0.09 &  & 5 & 544.19 &  & 5 & 424.89 &  & 0 & - &  & 5 & \textbf{0.41} \\
(10,100) &  & 3/5 & 19.76 & -0.07 &  & 4/5 & 112.71 & -0.01 &  & 0 & - &  & 3 & 1531.24 &  & 0 & - &  & 5 & \textbf{0.18} \\ \midrule
Summary: &  & 32/75 &  &  &  & 31/75 &  &  &  & 62 &  &  & 70 &  &  & 49 &  &  & \textbf{75} &  \\ \bottomrule
\end{tabular}%
}
\caption{Results of all methods on the small GNL instances.}
\label{tab:GNL_all}
\end{table}
Table~\ref{tab:GNL_all} presents a comprehensive comparison highlighting the relative strengths of the three proposed approaches for solving the assortment optimization problem under the GNL model. Specifically, we compare the bisection-based method for solving~\eqref{prob:Bisection} (denoted as \texttt{BIS}), the convex reformulation~\eqref{prob:Expcone} solved via a B\&C framework (denoted as \texttt{Convex}), and the bilinear--convex reformulation solved by B\&C with OA cuts and SCs (denoted as \texttt{BiCo}). For benchmarking purposes, we also report results obtained using the general-purpose solver SCIP, as well as the two MILP-based approximation approaches proposed by \citet{Cuong2024}.

Among the three exact methods, \texttt{BiCo} exhibits the strongest performance, solving all instances to optimality with the average runtime below 1 second. \texttt{BIS} shows strong overall performance, consistently solving most instances to optimality with competitive runtimes, especially in smaller and moderate-size problems. It achieves 70 out of 75 optimal solutions, indicating high robustness and reliability. However, its runtime increases noticeably with problem size, occasionally exceeding 1000 seconds for larger instances such as (2,80), (2,100), (5,100) and (10,100). On the other hand, the \texttt{Convex} approach  demonstrates mixed results: while it is effective on small and medium instances, solving them in acceptable time, its performance degrades significantly for larger instances, both in terms of runtime and number of optimally solved cases. Notably, it solves only 49 out of 75 instances, with several timeouts or failures on larger problems. This contrast suggests that while both formulations are viable, the Bisection approach offers a better balance between robustness and scalability.

\subsection{Assortment Optimization under MGNL Model: Convex and Bilinear--Convex Reformulations}\label{ec:assort-MGNL}
We report comparative results for the assortment optimization problem with multiple customer segments, corresponding to the  MGNL model. In this setting, we compare our approaches based on the convex and bilinear--convex reformulations embedded in B\&C procedure with OA cuts and SCs, as described in Appendix~\ref{ec:mixed_gnl_model}, against \texttt{SCIP} solver and the MILP-based approximation approach proposed by \citet{Cuong2024}. We denote the convex reformulation via logarithmic transformation \eqref{prob: Assort--MGNL--Convex} as \texttt{Convex} and the bilinear--convex reformulation \eqref{prob:M-Bilinear} as \texttt{BiCo}.

The comparison is conducted on 27 MGNL datasets, generated by varying the number of customer segments, nests, and products, with
$T \in \{2,5,10\}$, $N \in \{5,10,20\}$, and $m \in \{20,50,100\}$. For each dataset, we generate five independent instances. In each instance, the revenue of product~$i$ is given by $r_i = u_i^2 X_i$, and the preference parameter is defined as $V_{tin} = (1 - u_i) Y_{ti}$ for all $t \in [T]$ and $i \in [m]$, where $u_i \sim (0,1]$, $X_i \sim [0.1,10]$, and $Y_{ti} \sim [0.1,10]$ are independently sampled. The arrival probabilities $\theta_t$ for customer segment~$t$ are drawn independently from $(0,1)$ and then normalized so that different customer types have heterogeneous arrival rates and the total probability satisfies $\sum_{t \in [T]} \theta_t = 1$.
% \subsection{Comparison of the Bisection, Convex and Bilinear--Convex Reformulations on the GNL and MGNL instances}\label{appd:Exp-Bis}

\begin{table}[htb]\footnotesize
\centering
\resizebox{0.67\textwidth}{!}{%
\begin{tabular}{rlrrrlrrlrrlrr}
\toprule
 &  & \multicolumn{3}{c}{MILP-based approximation} & \multicolumn{1}{c}{} & \multicolumn{2}{c}{\multirow{2}{*}{\texttt{SCIP}}} & \multicolumn{1}{c}{} & \multicolumn{5}{c}{B\&C} \\ \cmidrule{3-5} \cmidrule{10-14} 
 &  & \multicolumn{3}{c}{MILP (99\%)} & \multicolumn{1}{c}{} & \multicolumn{2}{c}{} & \multicolumn{1}{c}{} & \multicolumn{2}{c}{\texttt{Convex}} & \multicolumn{1}{c}{} & \multicolumn{2}{c}{\texttt{BiCo}} \\ \cmidrule{3-5} \cmidrule{7-8} \cmidrule{10-11} \cmidrule{13-14} 
$(T,N,m)$ &  & \#O/\#S & Time(s) & Gap(\%) &  & \#O & Time(s) &  & \#O & Time(s) &  & \#O & Time(s) \\ \midrule
(2,2,20) &  & 1/5 & 0.85 & -1.85 &  & 5 & 1.4 &  & 5 & 0.45 &  & 5 & \textbf{0.21} \\
(2,2,30) &  & 3/5 & 1.54 & -0.22 &  & 5 & 6.29 &  & 5 & 1.62 &  & 5 & \textbf{0.73} \\
(2,2,50) &  & 0/5 & 4.55 & -1.24 &  & 5 & 399.00 &  & 5 & 82.14 &  & 5 & \textbf{2.26} \\
(2,2,80) &  & 1/5 & 14.38 & -0.61 &  & 0 & - &  & 3 & 1736.11 &  & 5 & \textbf{6.14} \\
(2,2,100) &  & 0/5 & 29.22 & -0.95 &  & 0 & - &  & 0 & - &  & 5 & \textbf{7.45} \\ \midrule
(2,5,20) &  & 2/5 & 0.99 & -0.32 &  & 5 & 1.36 &  & 5 & \textbf{0.18} &  & 5 & 0.35 \\
(2,5,30) &  & 4/5 & 2.72 & -0.25 &  & 5 & 4.82 &  & 5 & 1.93 &  & 5 & \textbf{1.35} \\
(2,5,50) &  & 3/5 & 28.28 & -0.23 &  & 5 & 586.69 &  & 5 & 141.95 &  & 5 & \textbf{3.08} \\
(2,5,80) &  & 2/5 & 88.68 & -0.27 &  & 0 & - &  & 0 & - &  & 5 & \textbf{17.30} \\
(2,5,100) &  & 2/5 & 500.13 & -0.21 &  & 0 & - &  & 0 & - &  & 5 & \textbf{7.54} \\ \midrule
(5,2,20) &  & 1/5 & 8.59 & -1.18 &  & 5 & 3.61 &  & 5 & \textbf{0.23} &  & 5 & 1.05 \\
(5,2,30) &  & 4/5 & 29.25 & -1.01 &  & 5 & 24.69 &  & 5 & 3.81 &  & 5 & \textbf{4.78} \\
(5,2,50) &  & 0/5 & 211.89 & -0.81 &  & 0 & - &  & 5 & 266.35 &  & 5 & \textbf{11.93} \\
(5,2,80) &  & 0/0 & - & - &  & 0 & - &  & 0 & - &  & 5 & \textbf{35.40} \\
(5,2,100) &  & 0/0 & - & - &  & 0 & - &  & 0 & - &  & 5 & \textbf{28.78} \\ \midrule
(5,5,20) &  & 3/5 & 6.72 & -0.17 &  & 5 & 3.72 &  & 5 & \textbf{1.75} &  & 5 & 1.99 \\
(5,5,30) &  & 2/5 & 36.31 & -0.34 &  & 5 & 40.83 &  & 5 & 31.52 &  & 5 & \textbf{2.89} \\
(5,5,50) &  & 2/4 & 1870.77 & -0.21 &  & 0 & - &  & 3 & 2355.26 &  & 5 & \textbf{4.97} \\
(5,5,80) &  & 0/0 & - & - &  & 0 & - &  & 0 & - &  & 5 & \textbf{17.87} \\
(5,5,100) &  & 0/0 & - & - &  & 0 & - &  & 0 & - &  & 5 & \textbf{23.24} \\ \midrule
Summary: &  & 30/79 &  &  &  & 50 &  &  & 61 &  &  & \textbf{100} &  \\ \bottomrule
\end{tabular}%
}
\caption{Results of all methods on the small MGNL instances.}
\label{tab:MGNL_all}
\end{table}

Table \ref{tab:MGNL_all} presents a comparative evaluation of solution methods for some small MGNL instances. The \texttt{Convex}  shows moderately strong performance, solving to optimality 61 out of 100 instances. It performs reliably on small problem sizes—for example, all instances with $m = \{20,30,50\}$ are solved quickly and efficiently. However, its effectiveness decreases significantly for larger problems, as seen in the frequent timeouts or failures for instances such as $(2,2,100)$, $(2,5,100)$, and all configurations with $(T,N) = (5,2)$ or $(5,5)$ and $m \geq 80$. This decline in scalability is also evident in the substantial runtime spikes, with several instances taking over 1000 seconds or remaining unsolved. In contrast, \texttt{BiCo} demonstrates superior performance in terms of both solution quality and computational time. This method is able to solve all instances to optimality, with the average runtime not exceeding 36 seconds.

\begin{table}[htb]\footnotesize
\centering
\footnotesize
\resizebox{0.7\textwidth}{!}{%
\begin{tabular}{rrrrrrrrrrrrrr}
\toprule
            &  & \multicolumn{3}{c}{MILP-based approximation} &  & \multicolumn{8}{c}{\texttt{B\&C-BiCo} (ours)}                                     \\ \cmidrule{3-5} \cmidrule{7-14} 
           &  & \multicolumn{3}{c}{MILP (99\%)} &  & \multicolumn{2}{c}{OA} &  & \multicolumn{2}{c}{SC} &  & \multicolumn{2}{c}{OA + SC} \\ \cmidrule{3-5} \cmidrule{7-8} \cmidrule{10-11} \cmidrule{13-14} 
$(T,N,m)$    &  & \#O/\#S  & Time(s) & Gap(\%) &  & \#O       & Time(s)             &  & \#O       & Time(s)             &  & \#O          & Time(s)               \\ \midrule
(2,5,20)    &  & 2/5        & 0.99           & -0.32       &  & 5 & 0.35            &  & 5 & 0.38           &  & 5 & \textbf{0.34}    \\
(2,5,50)    &  & 3/5        & 28.28          & -0.23       &  & 5 & \textbf{3.08}   &  & 5 & 5.71           &  & 5 & 5.35             \\
(2,5,100)   &  & 2/5        & 500.13         & -0.21       &  & 5 & \textbf{7.54}   &  & 5 & 20.96          &  & 5 & 21.89            \\
(2,10,20)   &  & 5/5        & 0.71           & 0.00           &  & 5 & 0.29            &  & 5 & \textbf{0.27}  &  & 5 & \textbf{0.27}    \\
(2,10,50)   &  & 4/5        & 19.23          & -0.08       &  & 5 & 1.96            &  & 5 & \textbf{1.72}  &  & 5 & 2.01             \\
(2,10,100)  &  & 0/0        & -              & -           &  & 5 & 7.35            &  & 5 & \textbf{6.04}  &  & 5 & 6.72             \\
(2,20,20)   &  & 4/5        & 0.40           & -0.06       &  & 5 & \textbf{0.18}   &  & 5 & 0.19           &  & 5 & 0.19             \\
(2,20,50)   &  & 5/5        & 7.28           & 0.00           &  & 5 & 2.65            &  & 5 & \textbf{2.14}  &  & 5 & 2.22             \\
(2,20,100)  &  & 0/0        & -              & -           &  & 5 & \textbf{9.26}   &  & 5 & 14.00          &  & 5 & 17.22            \\ \midrule
(5,5,20)    &  & 3/5        & 6.72           & -0.17       &  & 5 & \textbf{1.99}   &  & 5 & 3.47           &  & 5 & 3.99             \\
(5,5,50)    &  & 2/4        & 1870.77        & -0.21       &  & 5 & \textbf{4.97}   &  & 5 & 7.28           &  & 5 & 7.14             \\
(5,5,100)   &  & 0/0        & -              & -           &  & 5 & \textbf{23.24}  &  & 5 & 23.94          &  & 5 & 28.66            \\
(5,10,20)   &  & 4/5        & 3.45           & 0.00           &  & 5 & 2.71            &  & 5 & \textbf{2.44}  &  & 5 & 2.95             \\
(5,10,50)   &  & 5/5        & 501.37         & 0.00           &  & 5 & \textbf{15.24}  &  & 5 & 25.02          &  & 5 & 28.72            \\
(5,10,100)  &  & 0/0        & -              & -           &  & 5 & \textbf{24.33}  &  & 5 & 53.35          &  & 5 & 52.95            \\
(5,20,20)   &  & 5/5        & 1.32           & 0.00           &  & 5 & 1.03            &  & 5 & \textbf{0.83}  &  & 5 & 0.88             \\
(5,20,50)   &  & 5/5        & 32.13          & 0.00           &  & 5 & 12.91           &  & 5 & \textbf{11.35} &  & 5 & 12.92            \\
(5,20,100)  &  & 0/0        & -              & -           &  & 5 & \textbf{57.44}  &  & 5 & 112.88         &  & 5 & 110.56           \\ \midrule
(10,5,20)   &  & 4/5        & 17.58          & -0.12       &  & 5 & \textbf{8.71}   &  & 5 & 11.96          &  & 5 & 17.11            \\
(10,5,50)   &  & 0/0        & -              & -           &  & 5 & \textbf{18.44}  &  & 5 & 19.4           &  & 5 & 29.03            \\
(10,5,100)  &  & 0/0        & -              & -           &  & 5 & 643.1           &  & 5 & 294.71         &  & 5 & \textbf{95.13}   \\
(10,10,20)  &  & 4/5        & 11.14          & -0.02       &  & 5 & 8.49            &  & 5 & \textbf{7.58}  &  & 5 & 7.72             \\
(10,10,50)  &  & 3/3        & 1830.77        & 0.00           &  & 5 & 49.76           &  & 5 & \textbf{41.79} &  & 5 & 60.79            \\
(10,10,100) &  & 0/0        & -              & -           &  & 5 & \textbf{392.75} &  &  5 & 442.68         &  & 5 & 426.32           \\
(10,20,20) &  & 5/5          & 3.96    & 0.00      &  & 5           & \textbf{1.66}       &  & 5           & \textbf{1.66}       &  & 5              & \textbf{1.66}         \\
(10,20,50)  &  & 4/5        & 135.13         & -0.01       &  & 5 & 48.55           &  & 5 & 44.69          &  & 5 & \textbf{40.25}   \\
(10,20,100) &  & 0/0        & -              & -           &  & 3 & 2242.14         &  & 2 & 3076.12        &  & 4 & \textbf{2176.14} \\ \midrule
Summary: &  & 69/82 &  &  &  & 133 &  &  & 132 &  &  & \textbf{134} & \\ 
\bottomrule
\end{tabular}%
}
\caption{Comparison results on the medium and large MGNL datasets.}
\label{tab:MGNL}
\end{table}

Table~\ref{tab:MGNL} presents a comparative evaluation between our \texttt{B\&C-BiCo} method and the MILP-based approximation approach from \cite{Cuong2024} on medium- and large-scale instances. The results clearly demonstrate the advantage of \texttt{B\&C-BiCo} in terms of both solution quality and computational efficiency. While the MILP-based approximation method fails to solve many instances, particularly as the problem size increases, the \texttt{B\&C-BiCo} variants—OA, SC, and OA+SC—consistently solve nearly all instances. Among them, the OA+SC variant delivers the best overall performance, successfully solving 134 out of 135 instances with competitive runtimes across all instances. Notably, even in challenging scenarios with large product sets (e.g., 100 products), where MILP approximation fails to solve, OA+SC remains both reliable and efficient. These findings underscore the robustness and scalability of the proposed \texttt{B\&C-BiCo} approach in addressing the complexity of MGNL-based assortment optimization problems.

% \subsection{Performance of Billinear Formulations}
% \begin{figure}[htb]
%     \centering
%     \begin{subfigure}[t]{0.48\textwidth}
%         \centering
%         \includegraphics[width=\textwidth]{GNL_Combined.pdf}
%         \caption{GNL instances}
%         \label{fig:GNL}
%     \end{subfigure}
%     \hfill
%     \begin{subfigure}[t]{0.48\textwidth}
%         \centering
%         \includegraphics[width=\textwidth]{MGNL_Combined.pdf}
%         \caption{MGNL instances}
%         \label{fig:MGNL}
%     \end{subfigure}
%     \caption{Comparison of runtime and objective gap on GNL and MGNL datasets}
%     \label{fig:GNL_MGNL}
% \end{figure}

% \begin{figure}[htb]
%     \centering
%     \begin{subfigure}[t]{0.48\textwidth}
%         \centering
%         \includegraphics[width=\textwidth]{AP_Combined.pdf}
%         \caption{A\&DP instances}
%         \label{fig:AP}
%     \end{subfigure}
%     \hfill
%     \begin{subfigure}[t]{0.48\textwidth}
%         \centering
%         \includegraphics[width=\textwidth]{CP_Combined.pdf}
%         \caption{A\&CP instances}
%         \label{fig:CP}
%     \end{subfigure}
%     \caption{Comparison of runtime and objective gap on the joint assortment optimization and pricing datasets}
%     \label{fig:AP_CP}
% \end{figure}

\subsection{Discrete-Price JAP under GNL Model (A\&DP): Bisection, Convex and Bilinear--Convex Reformulations}\label{ec:jap_dp}
We present comparative results for the JAP problem with discrete price variables. As discussed earlier, this problem can be equivalently reformulated as an assortment optimization problem over an extended set of products corresponding to discrete price levels. Consequently, all methods developed for the assortment optimization problem can be directly applied in this setting, and we benchmark our proposed approaches against the MILP-based approximation method proposed by \citet{Cuong2024}.

We generate 27 datasets to test our exact method on the joint assortment optimization and discrete pricing problem. The number of nests $N$, number of products $m$ and number of price levels $L$ vary over $\{5,10,20\}$, $\{20,50,100\}$ and $\{2,5,10\}$, respectively. In this experiment, we assume that the price $p_{il}$ and the utility $v_{il}$ follow parametric relationship suggested by \cite{LiHuh2011} and \cite{gallego2014multiproduct}, i.e., $v_{il}=\exp(\mu_i-\eta_i p_{il})$, where: $\mu_i \sim [-1,1]$ and $\eta_i \sim (0,1)$. The discrete price set $P_i$ are generated by using random number $\gamma_i \sim (0,1)$ and calculating the prices as $p_{il}=l\gamma_i+0.5, \forall l\in [L]$.

\begin{table}[htb]\footnotesize
\centering
\resizebox{\textwidth}{!}{%
\begin{tabular}{rrrrrrrrrrrrrrrrrrrrr}
\toprule
 &  & \multicolumn{7}{c}{MILP-based approximation} &  & \multicolumn{2}{c}{\multirow{2}{*}{\texttt{SCIP}}} &  & \multicolumn{8}{c}{B\&C} \\ \cmidrule{3-9} \cmidrule{14-21} 
 &  & \multicolumn{3}{c}{BIS (99\%)} &  & \multicolumn{3}{c}{MILP (99\%)} &  & \multicolumn{2}{c}{} &  & \multicolumn{2}{c}{\texttt{BIS}} &  & \multicolumn{2}{c}{\texttt{Convex}} &  & \multicolumn{2}{c}{\texttt{BiCo}} \\ \cmidrule{3-5} \cmidrule{7-9} \cmidrule{11-12} \cmidrule{14-15} \cmidrule{17-18} \cmidrule{20-21} 
$(N,m,L)$ &  & \#O/\#S & Time(s) & Gap(\%) &  & \#O/\#S & Time(s) & Gap(\%) &  & \#O & Time(s) &  & \#O & Time(s) &  & \#O & Time(s) &  & \#O & Time(s) \\ \midrule
(2,20,2) &  & 2/5 & 3.73 & -0.25 &  & 1/5 & 1.17 & -0.35 &  & 5 & 14.95 &  & 5 & 0.60 &  & 5 & 0.17 &  & 5 & \textbf{0.04} \\
(2,20,5) &  & 1/5 & 17.72 & -0.51 &  & 1/5 & 5.96 & -0.44 &  & 5 & 469.25 &  & 5 & 17.36 &  & 5 & 3.73 &  & 5 & \textbf{0.14} \\
(2,30,2) &  & 0/5 & 13.9 & -0.54 &  & 1/5 & 2.60 & -0.50 &  & 4 & 1879.51 &  & 5 & 1.64 &  & 5 & 0.96 &  & 5 & \textbf{0.03} \\
(2,30,5) &  & 1/5 & 17.72 & -0.51 &  & 1/5 & 5.96 & -0.44 &  & 0 & - &  & 5 & 852.92 &  & 4 & 2232.34 &  & 5 & \textbf{0.22} \\
(2,50,2) &  & 0/5 & 30.60 & -0.86 &  & 0/5 & 7.11 & -0.79 &  & 0 & - &  & 5 & 7.07 &  & 5 & 1.82 &  & 5 & \textbf{0.06} \\
(2,50,5) &  & 1/5 & 109.52 & -0.80 &  & 1/5 & 78.82 & -0.93 &  & 0 & - &  & 5 & 348.63 &  & 4 & 1517.02 &  & 5 & \textbf{0.52} \\
(2,80,2) &  & 0/5 & 101.09 & -0.96 &  & 0/5 & 26.98 & -0.68 &  & 0 & - &  & 5 & 14.54 &  & 4 & 773.31 &  & 5 & \textbf{0.08} \\
(2,80,5) &  & 0/5 & 432.31 & -0.61 &  & 0/5 & 209.22 & -0.59 &  & 0 & - &  & 3 & 2052.00 &  & 0 & - &  & 5 & \textbf{0.59} \\
(2,100,2) &  & 0/5 & 146.22 & -0.90 &  & 0/5 & 25.56 & -0.83 &  & 0 & - &  & 5 & 28.03 &  & 3 & 2167.81 &  & 5 & \textbf{0.38} \\
(2,100,5) &  & 0/5 & 629.11 & -0.80 &  & 0/5 & 383.56 & -0.46 &  & 0 & - &  & 2 & 2640.93 &  & 0 & - &  & 5 & \textbf{1.19} \\ \midrule
(5,20,2) &  & 3/5 & 2.77 & -0.09 &  & 3/5 & 1.52 & -0.09 &  & 5 & 10.02 &  & 5 & 0.85 &  & 5 & \textbf{0.06} &  & 5 & 0.09 \\
(5,20,5) &  & 4/5 & 10.57 & -0.01 &  & 4/5 & 13.5 & -0.01 &  & 5 & 465.62 &  & 5 & 37.22 &  & 5 & 7.96 &  & 5 & \textbf{0.20} \\
(5,30,2) &  & 3/5 & 6.44 & -0.04 &  & 5/5 & 4.18 & 0.00 &  & 5 & 210.24 &  & 5 & 1.20 &  & 5 & 0.37 &  & 5 & \textbf{0.06} \\
(5,30,5) &  & 2/5 & 64.21 & -0.07 &  & 3/5 & 85.83 & -0.05 &  & 0 & - &  & 3 & 1974.29 &  & 0 & - &  & 5 & \textbf{0.58} \\
(5,50,2) &  & 1/5 & 79.77 & -0.08 &  & 1/5 & 38.01 & -0.09 &  & 0 & - &  & 5 & 9.10 &  & 5 & 2.37 &  & 5 & \textbf{0.13} \\
(5,50,5) &  & 0/5 & 192.87 & -0.09 &  & 2/5 & 506.38 & -0.04 &  & 0 & - &  & 3 & 1867.07 &  & 0 & - &  & 5 & \textbf{0.59} \\
(5,80,2) &  & 1/5 & 186.60 & -0.13 &  & 0/5 & 66.45 & -0.14 &  & 0 & - &  & 5 & 14.62 &  & 5 & 15.04 &  & 5 & \textbf{0.08} \\
(5,80,5) &  & 0/5 & 704.62 & -0.16 &  & 0/5 & 787.52 & -0.13 &  & 0 & - &  & 2 & 2443.00 &  & 0 & - &  & 5 & \textbf{0.69} \\
(5,100,2) &  & 0/5 & 264.25 & -0.18 &  & 0/5 & 112.02 & -0.23 &  & 0 & - &  & 5 & 35.62 &  & 5 & 776.58 &  & 5 & \textbf{0.18} \\
(5,100,5) &  & 0/5 & 1277.56 & -0.31 &  & 0/5 & 1530.6 & -0.28 &  & 0 & - &  & 0 & - &  & 0 & - &  & 5 & \textbf{2.01} \\ \midrule
Summary: &  & 19/100 &  &  &  & 23/100 &  &  &  & 29 &  &  & 83 &  &  & 65 &  &  & \textbf{100} &  \\ \bottomrule
\end{tabular}%
}
\caption{Results of all methods on the small A\&DP instances.}
\label{tab:ADP-all}
\end{table}
Table \ref{tab:ADP-all} presents a detailed comparison of various methods on small A\&DP instances, with particular attention to the performance of the \eqref{prob:Bisection}, \eqref{prob:Expcone} and bilinear--convex formulations under the B\&C framework with OA cuts and SCs, denoted as \texttt{BIS}, \texttt{Convex} and \texttt{BiCo}, respectively. Once again, \texttt{BiCo} confirms itself as the best-performing method among the three exact approaches considered, as it is able to solve all instances to optimality with the average computational time of approximately 2 seconds or less. The \texttt{BIS} demonstrates strong overall performance, solving 83 out of 100 instances optimally, with especially low solution times for smaller and moderately sized problems. However, as the instance size increases (e.g., for large $N$, $m$, and $L$), its runtime rises significantly—occasionally exceeding 2000 seconds—highlighting some scalability limitations. Nonetheless, it remains far more robust and effective than the benchmark methods from \cite{Cuong2024} and \texttt{SCIP}, both of which solve fewer than 30 instances each. The \texttt{Convex} formulation, while less robust than its Bisection counterpart, still performs competitively, solving 65 out of 100 instances. It tends to perform well on small and medium-sized instances but struggles on the larger ones, where timeouts or failures are common. Additionally, the runtimes of \texttt{Convex} are often significantly higher than those of \texttt{BIS} on the same instances, primarily due to the wide range of the variable that presents the denominator of the objective function.

\begin{table}[htb]\footnotesize
\centering
\footnotesize
\resizebox{0.9\textwidth}{!}{%
\begin{tabular}{rlrrrlrrrlrrlrrlrr}
\toprule
            &  & \multicolumn{7}{c}{MILP-based approximation}        &  & \multicolumn{8}{c}{\texttt{B\&C-BiCo} (ours)}                               \\ \cmidrule{3-9} \cmidrule{11-18} 
 &
   &
  \multicolumn{3}{c}{BIS (99\%)} &
   &
  \multicolumn{3}{c}{MILP (99\%)} &
   &
  \multicolumn{2}{c}{OA} &
   &
  \multicolumn{2}{c}{SC} &
   &
  \multicolumn{2}{c}{OA + SC} \\ \cmidrule{3-5} \cmidrule{7-9} \cmidrule{11-12} \cmidrule{14-15} \cmidrule{17-18} 
$(N,m,L)$ &
   &
  \#O/\#S &
  Time(s) &
  Gap(\%) &
   &
  \#O/\#S &
  Time(s) &
  Gap(\%) &
   &
  \#O &
  Time(s) &
   &
  \#O &
  Time(s) &
   &
  \#O &
  Time(s) \\ \midrule
% (5,20,2)    &  & 3/5 & 2.77    & -0.09 &  & 3/5 & 1.52    & -0.09 &  & 5 & \textbf{0.09} &  & 5 & \textbf{0.09} &  & 5 & \textbf{0.09} \\
% (5,20,5)    &  & 4/5 & 10.57   & -0.01 &  & 4/5 & 13.50   & -0.01 &  & 5 & \textbf{0.20} &  & 5 & 0.34          &  & 5 & 0.26          \\
% (5,20,10)   &  & 1/5 & 61.96   & -0.06 &  & 1/5 & 151.04  & -0.06 &  & 5 & \textbf{0.35} &  & 5 & 1.16          &  & 5 & 0.38          \\
% (5,50,2)    &  & 1/5 & 79.77   & -0.08 &  & 1/5 & 38.01   & -0.09 &  & 5 & \textbf{0.13} &  & 5 & 0.19          &  & 5 & 0.14          \\
% (5,50,5)    &  & 0/5 & 192.87  & -0.09 &  & 2/5 & 506.38  & -0.04 &  & 5 & \textbf{0.59} &  & 5 & 2.50          &  & 5 & 0.75          \\
% (5,50,10)   &  & 0/5 & 879.88  & -0.13 &  & 0/0 & -       & -     &  & 5 & \textbf{2.17} &  & 5 & 575.88        &  & 5 & 3.10          \\
% (5,100,2)   &  & 0/5 & 264.25  & -0.18 &  & 0/5 & 112.02  & -0.23 &  & 5 & \textbf{0.18} &  & 5 & 0.24          &  & 5 & \textbf{0.18} \\
% (5,100,5)   &  & 0/5 & 1277.56 & -0.31 &  & 0/5 & 1530.60 & -0.28 &  & 5 & \textbf{2.01} &  & 5 & 7.84          &  & 5 & 2.72          \\
% (5,100,10)  &  & 0/0 & -       & -     &  & 0/0 & -       & -     &  & 5 & \textbf{9.34} &  & 0 & -             &  & 5 & 22.58         \\ \midrule
(10,20,2)   &  & 3/5 & 0.96    & -0.07 &  & 3/5 & 2.96    & -0.03 &  & 5 & 0.13          &  & 5 & 0.09          &  & 5 & \textbf{0.07} \\
(10,20,5)   &  & 2/5 & 2.79    & -0.03 &  & 4/5 & 39.69   & 0.00  &  & 5 & \textbf{0.12} &  & 5 & 0.16          &  & 5 & 0.13          \\
(10,20,10)  &  & 4/5 & 8.03    & -0.01 &  & 5/5 & 101.26  & 0.00  &  & 5 & \textbf{0.17} &  & 5 & 0.24          &  & 5 & \textbf{0.17} \\
(10,50,2)   &  & 2/5 & 41.9    & -0.01 &  & 1/5 & 240.76  & -0.02 &  & 5 & 0.17          &  & 5 & \textbf{0.12} &  & 5 & \textbf{0.12} \\
(10,50,5)   &  & 2/5 & 164.73  & -0.03 &  & 0/1 & 3583.72 & -0.06 &  & 5 & \textbf{0.42} &  & 5 & 1.09          &  & 5 & 0.53          \\
(10,50,10)  &  & 0/3 & 2534.71 & -0.12 &  & 0/0 & -       & -     &  & 5 & \textbf{1.22} &  & 5 & 22.28         &  & 5 & 1.78          \\
(10,100,2)  &  & 0/4 & 2536.04 & -0.05 &  & 0/0 & -       & -     &  & 5 & \textbf{0.11} &  & 5 & 0.19          &  & 5 & 0.12          \\
(10,100,5)  &  & 0/3 & 2559.40 & -0.10 &  & 0/0 & -       & -     &  & 5 & \textbf{1.21} &  & 5 & 3.63          &  & 5 & 1.75          \\
(10,100,10) &  & 0/0 & -       & -     &  & 0/0 & -       & -     &  & 5 & \textbf{8.21} &  & 0 & -             &  & 5 & 10.45         \\ \midrule
(20,20,2)   &  & 5/5 & 0.26    & 0.00  &  & 5/5 & 0.47    & 0.00  &  & 5 & \textbf{0.09} &  & 5 & \textbf{0.09} &  & 5 & 0.10          \\
(20,20,5)   &  & 5/5 & 0.77    & 0.00  &  & 5/5 & 10.55   & 0.00  &  & 5 & 0.20          &  & 5 & 0.16          &  & 5 & \textbf{0.11} \\
(20,20,10)  &  & 4/5 & 1.12    & -0.01 &  & 5/5 & 66.20   & 0.00  &  & 5 & 0.27          &  & 5 & \textbf{0.25} &  & 5 & \textbf{0.25} \\
(20,50,2)   &  & 4/5 & 4.60    & 0.00  &  & 5/5 & 110.59  & 0.00  &  & 5 & 0.21          &  & 5 & 0.19          &  & 5 & \textbf{0.18} \\
(20,50,5)   &  & 3/5 & 29.95   & -0.01 &  & 0/0 & -       & -     &  & 5 & \textbf{0.30} &  & 5 & 0.43          &  & 5 & 0.33          \\
(20,50,10)  &  & 1/5 & 102.78  & -0.03 &  & 0/0 & -       & -     &  & 5 & \textbf{0.66} &  & 5 & 1.79          &  & 5 & 0.69          \\
(20,100,2)  &  & 1/4 & 1805.12 & -0.03 &  & 0/0 & -       & -     &  & 5 & 0.45          &  & 5 & 0.46          &  & 5 & \textbf{0.37} \\
(20,100,5)  &  & 0/0 & -       & -     &  & 0/0 & -       & -     &  & 5 & \textbf{0.85} &  & 5 & 2.63          &  & 5 & 0.97          \\
(20,100,10) &  & 0/0 & -       & -     &  & 0/0 & -       & -     &  & 5 & \textbf{7.71} &  & 0 & -             &  & 5 & 10.50         \\ \midrule
Summary: &
   &
  36/69 &
  \multicolumn{1}{l}{} &
  \multicolumn{1}{l}{} &
   &
  33/41 &
  \multicolumn{1}{l}{} &
  \multicolumn{1}{l}{} &
   &
  \textbf{90} &
  \multicolumn{1}{l}{} &
   &
  80 &
  \multicolumn{1}{l}{} &
   &
  \textbf{90} &
  \multicolumn{1}{l}{} \\ \bottomrule
\end{tabular}%
}
\caption{Comparison results on the medium and large A\&DP instances.}
\label{tab:ADP}
\end{table}

Table~\ref{tab:ADP} provides a detailed comparison between the MILP-based approximation method and our proposed formulation (in Appendix~\ref{ec:JAP_DP}) embedded in the B\&C procedure. Across all instance configurations, the \texttt{B\&C-BiCo} variants—OA, SC, and OA+SC—consistently outperform BIS and MILP in both solution quality and runtime. \texttt{B\&C-BiCo} with OA cuts and OA+SC are able to solve all instances significantly faster, often by several orders of magnitude. Among them, OA typically yields the fastest runtimes, while OA+SC also demonstrates strong performance stability across different problem sizes. In contrast, BIS and MILP struggle with larger instances, frequently failing to find solutions within time limit. These results highlight the superior scalability and effectiveness of the B\&C framework for solving complex A\&DP instances.

\subsection{Comparative analysis of the GNL and MGNL models}\label{ec:gnl_mgnl}
% \begin{figure}[htb]
%     \centering
%     \includegraphics[width=\linewidth]{gap_plot.pdf}
%     \caption{Objective value gap between MGNL and GNL models across problem instances.}
%     \label{fig:gap_plot}
% \end{figure}

% \usepackage{booktabs}
% \usepackage{graphicx}
\begin{table}[htb]
\centering
\resizebox{\textwidth}{!}{%
\begin{tabular}{@{}rrrrrrrrrrrrrrrrr@{}}
\toprule
         &  & \multicolumn{3}{c}{$N=5$} &  &          &  & \multicolumn{3}{c}{$N=10$} &  &          &  & \multicolumn{3}{c}{$N=20$} \\ \cmidrule(lr){3-5} \cmidrule(lr){9-11} \cmidrule(l){15-17} 
$(m,T)$ &  & MinGap(\%) & AveGap(\%) & MaxGap(\%) &  & $(m,T)$ &  & MinGap(\%) & AveGap(\%) & MaxGap(\%) &  & $(m,T)$ &  & MinGap(\%) & AveGap(\%) & MaxGap(\%) \\ \cmidrule(r){1-5} \cmidrule(lr){7-11} \cmidrule(l){13-17} 
(20,2)   &  & 0.00    & 1.36   & 5.16   &  & (20,2)   &  & 0.00    & 0.09    & 0.47   &  & (20,2)   &  & 0.00    & 0.61    & 1.80   \\
(20,5)   &  & 0.99    & 9.19   & 21.70  &  & (20,5)   &  & 1.39    & 5.93    & 10.98  &  & (20,5)   &  & 0.00    & 1.12    & 2.54   \\
(20,10)  &  & 9.38    & 17.04  & 24.32  &  & (20,10)  &  & 1.93    & 11.38   & 24.13  &  & (20,10)  &  & 0.00    & 4.36    & 12.16  \\ \cmidrule(r){1-5} \cmidrule(lr){7-11} \cmidrule(l){13-17} 
(50,2)   &  & 1.28    & 3.16   & 6.66   &  & (50,2)   &  & 0.31    & 1.57    & 2.82   &  & (50,2)   &  & 0.00    & 2.20    & 3.93   \\
(50,5)   &  & 7.98    & 13.02  & 20.43  &  & (50,5)   &  & 7.35    & 15.15   & 20.44  &  & (50,5)   &  & 5.61    & 9.81    & 15.11  \\
(50,10)  &  & 23.08   & 29.97  & 34.97  &  & (50,10)  &  & 12.28   & 22.47   & 30.87  &  & (50,10)  &  & 4.22    & 14.07   & 21.56  \\ \cmidrule(r){1-5} \cmidrule(lr){7-11} \cmidrule(l){13-17} 
(100,2)  &  & 1.17    & 2.27   & 3.39   &  & (100,2)  &  & 1.28    & 2.78    & 4.24   &  & (100,2)  &  & 0.47    & 2.00    & 3.05   \\
(100,5)  &  & 7.59    & 13.32  & 20.09  &  & (100,5)  &  & 10.48   & 13.07   & 16.54  &  & (100,5)  &  & 10.14   & 11.86   & 14.23  \\
(100,10) &  & 18.91   & 26.33  & 33.61  &  & (100,10) &  & 20.42   & 25.90   & 30.49  &  & (100,10) &  & 16.71   & 22.27   & 26.87  \\ \bottomrule
\end{tabular}%
}
\caption{Objective value gap between MGNL and GNL models across problem instances.}
\label{tab:gap_obj_gnl_mgnl}
\end{table}

The GNL model can be viewed as a special case of the MGNL model when the number of customer segments satisfies $T=1$. In this sense, the MGNL model extends the GNL framework by explicitly capturing heterogeneity across multiple customer segments, each characterized by distinct preference parameters and arrival probabilities. While the GNL model aggregates demand into a single representative population, the MGNL model allows the firm to tailor assortment decisions to segment-level differences, which is particularly important in markets with diverse customer behaviors.

From a modeling perspective, applying the GNL formulation to an MGNL instance effectively amounts to ignoring customer heterogeneity. Specifically, product utilities within each customer segment are replaced by their average values across all segments, yielding a single-population approximation of the underlying heterogeneous market. Despite this simplification, the GNL model can serve as a useful baseline for comparison, as it reflects a common practice in which firms rely on aggregate demand estimates when segment-level information is unavailable or computationally difficult to exploit.

To quantify the performance loss induced by this aggregation, we solve the GNL model constructed from averaged utilities and then evaluate the resulting assortment decisions under the original MGNL objective function to obtain ground-truth objective values. %Figure~\ref{fig:gap_plot} 
Table~\ref{tab:gap_obj_gnl_mgnl} reports the percentage gap in objective values between the MGNL and GNL models under various configurations of $m$, $N$, and $T$. The results show that the performance gap increases monotonically with $T$, indicating that the benefits of the MGNL model become more pronounced as customer heterogeneity intensifies. When the number of segments is small $(T=2)$, both models yield comparable results, with average gaps typically below $4\%$. However, for larger values of $T$, the MGNL model consistently achieves higher objective values, reflecting its superior ability to exploit segment-specific preferences. For $T=5$, the average gap ranges between 1.12\% and 15.15\%, whereas for $T=10$, it increases to a range of 4.36\% to 29.97\%. These findings highlight the limitations of aggregate-choice models and provide strong empirical motivation for adopting the MGNL framework in multi-segment environments, where accounting for heterogeneity can lead to substantial improvements in revenue performance.

\subsection{Objective gap comparison for continuous-price JAP}\label{ec:obj_jap_cp}
% \begin{figure}[htb]
%     \centering
%     \includegraphics[width=\linewidth]{gap_continuous.pdf}
%     \caption{Objective gap comparison.}
%     \label{fig:gap_continuous}
% \end{figure}
We provide a comparison in terms of objective value gaps among three main solution approaches—\texttt{B\&C-BiCo-PWLA}, \texttt{B\&C-BiCo-DP} $(\#Price = 2L)$, and the \texttt{SCIP} solver—with the motivation of quantifying the trade-off between modeling accuracy and computational tractability when solving the JAP problem with continuous prices. In particular, this comparison allows us to assess the performance loss induced by price discretization and to benchmark the quality of the proposed PWLA-based reformulation against a general-purpose MINLP solver. %The objective gap corresponding to each B\&C-based method and each instance is calculated as follows: $\frac{Objective_{\texttt{B\&C-BiCo}}-Objective_{\texttt{SCIP}}}{Objective_{\texttt{SCIP}}}\times 100(\%)$.

% \usepackage{booktabs}
% \usepackage{graphicx}
\begin{table}[htb]
\centering
\resizebox{\textwidth}{!}{%
\begin{tabular}{@{}rrrrrrrrrrrrrrrrrrr@{}}
\toprule
 &
   &
  \multicolumn{3}{c}{\texttt{B\&C-BiCo-DP}} &
   &
  \multicolumn{3}{c}{\texttt{B\&C-BiCo-PWLA}} &
   &
   &
   &
  \multicolumn{3}{c}{\texttt{B\&C-BiCo-DP}} &
   &
  \multicolumn{3}{c}{\texttt{B\&C-BiCo-PWLA}} \\ \cmidrule(lr){3-5} \cmidrule(lr){7-9} \cmidrule(lr){13-15} \cmidrule(l){17-19} 
$(N,m,L)$ &
   &
  MinGap(\%) &
  AveGap(\%) &
  MaxGap(\%) &
   &
  MinGap(\%) &
  AveGap(\%) &
  MaxGap(\%) &
   &
  $(N,m,L)$ &
   &
  MinGap(\%) &
  AveGap(\%) &
  MaxGap(\%) &
   &
  MinGap(\%) &
  AveGap(\%) &
  MaxGap(\%) \\ \cmidrule(r){1-9} \cmidrule(lr){11-15} \cmidrule(l){17-19} 
(2,10,5)  &  & -0.029 & -0.018 & 0.000  &  & 0.000 & \textbf{0.000}  & 0.000  &  & (5,10,5)  &  & -0.022 & -0.011 & 0.000  &  & 0.000 & \textbf{0.000} & 0.000  \\
(2,20,5)  &  & -0.085 & 0.059  & 0.440  &  & 0.000 & \textbf{0.088}  & 0.441  &  & (5,20,5)  &  & -0.063 & 0.313  & 1.671  &  & 0.000 & \textbf{0.336} & 1.677  \\
(2,30,5)  &  & 0.191  & 2.090  & 5.536  &  & 0.197 & \textbf{2.105}  & 5.548  &  & (5,30,5)  &  & -0.002  & 1.836  & 6.750  &  & 0.000 & \textbf{1.849} & 6.751  \\
(2,40,5)  &  & 0.219  & 5.126  & 9.369  &  & 0.232 & \textbf{5.133}  & 9.375  &  & (5,40,5)  &  & 0.424  & 2.660  & 6.331  &  & 0.417 & \textbf{2.669} & 6.333  \\
(2,50,5)  &  & 2.088  & 8.365  & 21.547 &  & 2.108 & \textbf{8.383}  & 21.587 &  & (5,50,5)  &  & 1.217  & \textbf{2.390}  & 4.960  &  & 1.158 & 2.387 & 4.962  \\ \cmidrule(r){1-9} \cmidrule(l){11-19} 
(2,10,10) &  & -0.026 & -0.010 & -0.001  &  & 0.000 & \textbf{0.000}  & 0.000  &  & (5,10,10) &  & -0.047 & -0.030 & -0.015 &  & -0.001 & \textbf{0.000} & 0.000  \\
(2,20,10) &  & -0.002  & 0.723  & 2.110  &  & 0.000 & \textbf{0.737}  & 2.120  &  & (5,20,10) &  & -0.005 & 1.635  & 5.141  &  & 0.000 & \textbf{1.655} & 5.161  \\
(2,30,10) &  & 0.194  & 1.899  & 3.463  &  & 0.208 & \textbf{1.909}  & 3.464  &  & (5,30,10) &  & 0.598  & 5.553  & 20.515 &  & 0.614 & \textbf{5.573} & 20.532 \\
(2,40,10) &  & 0.601  & 2.042  & 4.693  &  & 0.605 & \textbf{2.049}  & 4.693  &  & (5,40,10) &  & 2.171  & \textbf{5.442}  & 13.853 &  & 2.149 & 5.424 & 13.855 \\
(2,50,10) &  & 3.471  & \textbf{10.615} & 20.323 &  & 3.445 & 10.596 & 20.341 &  & (5,50,10) &  & 0.106  & \textbf{3.550}  & 7.550  &  & 0.011 & 3.523 & 7.554  \\ \bottomrule
\end{tabular}%
}
\caption{Objective value gap between \texttt{B\&C-BiCo-DP} and \texttt{B\&C-BiCo-PWLA} models compared to \texttt{SCIP}.}
\label{tab:gap_obj_jap}
\end{table}

%Figure~\ref{fig:gap_continuous}
Table~\ref{tab:gap_obj_jap} reports the minimum, average, and maximum objective value gaps achieved by the three approaches on each set of 5 instances. Notably, \texttt{SCIP} is able to solve all instances in the smallest datasets (with $m=10$) to optimality. Moreover, the minimum and maximum gap between the objective values obtained by \texttt{B\&C-BiCo-PWLA} and those provided by \texttt{SCIP} is $0.000\%$, indicating that \texttt{B\&C-BiCo-PWLA} attains almost-optimal solutions for small-scale JAP instances. For larger problem sizes, where \texttt{SCIP} fails to consistently reach optimality within the time limit, we observe that the two approximation-based methods remain competitive. In particular, \texttt{B\&C-BiCo-PWLA} slightly outperforms \texttt{B\&C-BiCo-DP}, achieving higher average gaps in 16 out of 20 instance sets with only marginal differences. These results highlight the effectiveness of the \texttt{B\&C-BiCo-PWLA} approach in preserving solution quality while avoiding the limitations of price discretization.

\subsection{Impacts of Cardinality Constraints and Cross Rate}\label{ec:cardinality_cross_rate}
This section investigates the effects of the cardinality parameters and the cross rate on three datasets: GNL, MGNL, and A\&DP. The A\&CP instances are excluded from our experiments, as they can be reformulated into A\&DP instances, which are more efficiently solvable. The experimental results are obtained using the B\&C approach with OA cuts. For each model, two datasets are evaluated, each consisting of 20 instances.

\begin{figure}[htb]
    \centering
    \includegraphics[width=\linewidth]{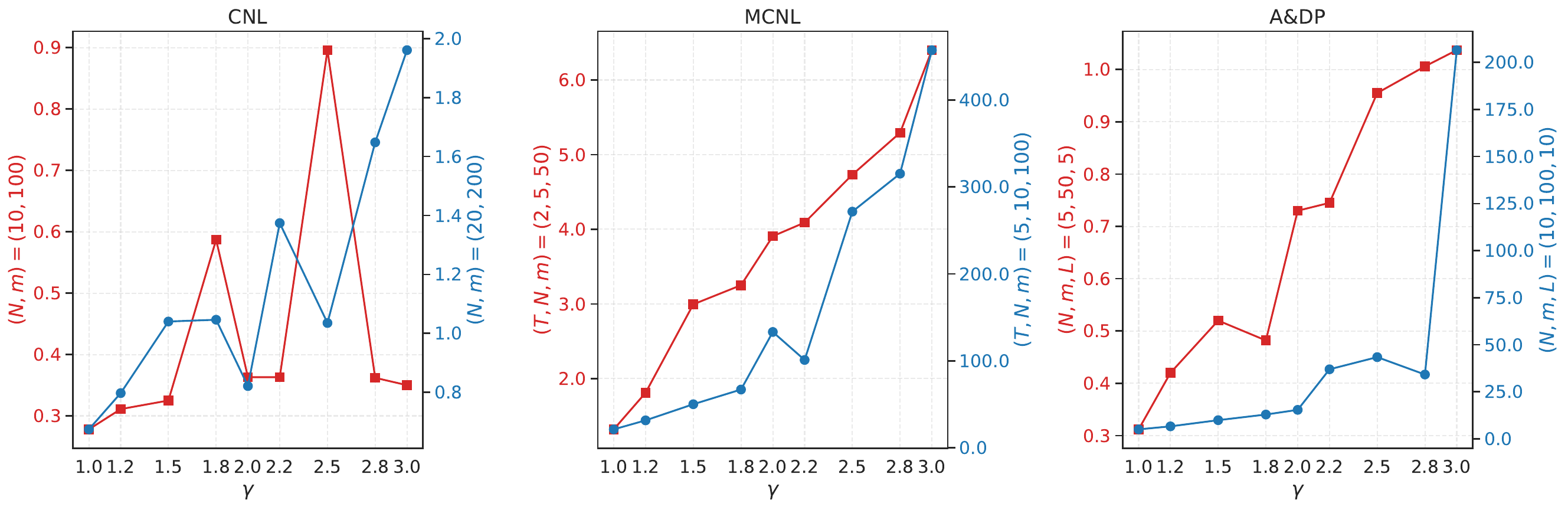}
    \caption{Runtime provided by B\&C approach (with OA cuts) when varying cross rate}
    \label{fig:cross-rate}
\end{figure}
\paragraph{Impact of Cross Rate $\gamma$.}
In the first subplot of Figure \ref{fig:cross-rate}, the computation time exhibits a generally increasing trend with respect to the rate parameter for both problem sizes \((N,m) = (10,100)\) and \((20,200)\). For the smaller instance \((10,100)\), the time remains relatively stable up to rate 2.2, with a spike at rate 2.5, while for the larger instance \((20,200)\), a more consistent increase is observed, from 0.67 seconds at rate 1.0 to 1.96 seconds at rate 3.0. This indicates that computation time tends to grow with both rate and problem size, though with occasional fluctuations.

A similar trend is found in both the second and third subplots: time increases with rate and scales significantly with problem size. For the MGNL model, time grows from 1.32 seconds to 6.40 seconds at \((2,5,50)\), and sharply from 21.24 seconds to 457.38 seconds at \((5,10,100)\). The A\&DP model shows a moderate increase from 0.31 seconds to 1.04 seconds at \((5,50,5)\), but a substantial rise from 5.04 seconds to 206.52 seconds at \((10,100,10)\). In both models, the most dramatic jumps occur at rate 3.0, highlighting a common pattern: computational burden grows notably at higher rates and larger scales.

\begin{figure}[htb]
    \centering
    \begin{subfigure}{0.48\textwidth}
        \centering
        \includegraphics[width=\textwidth]{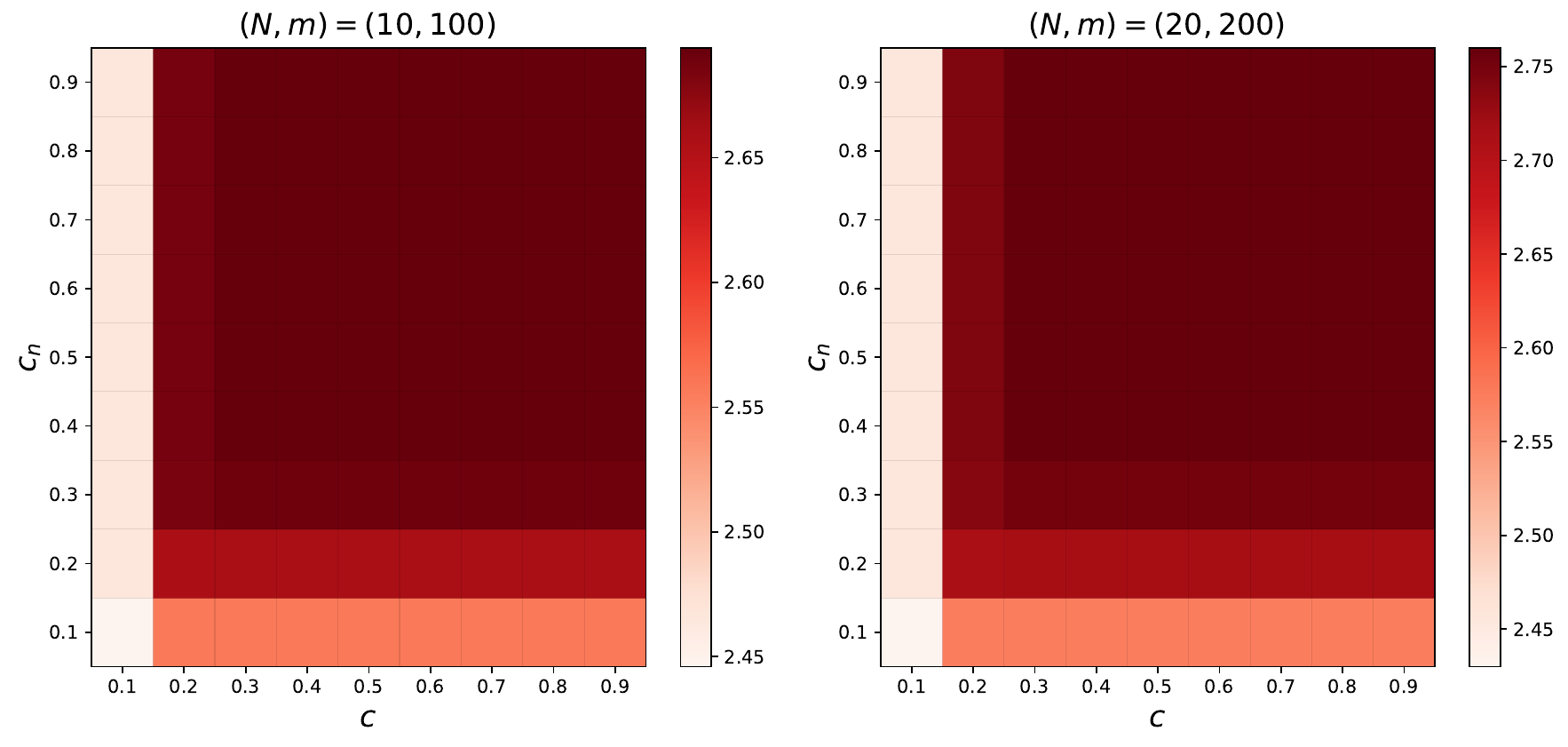}
        \caption{Objective value}
        \label{fig:GNL-Cap-Obj}
    \end{subfigure}
    \hfill
    \begin{subfigure}{0.48\textwidth}
        \centering
        \includegraphics[width=\textwidth]{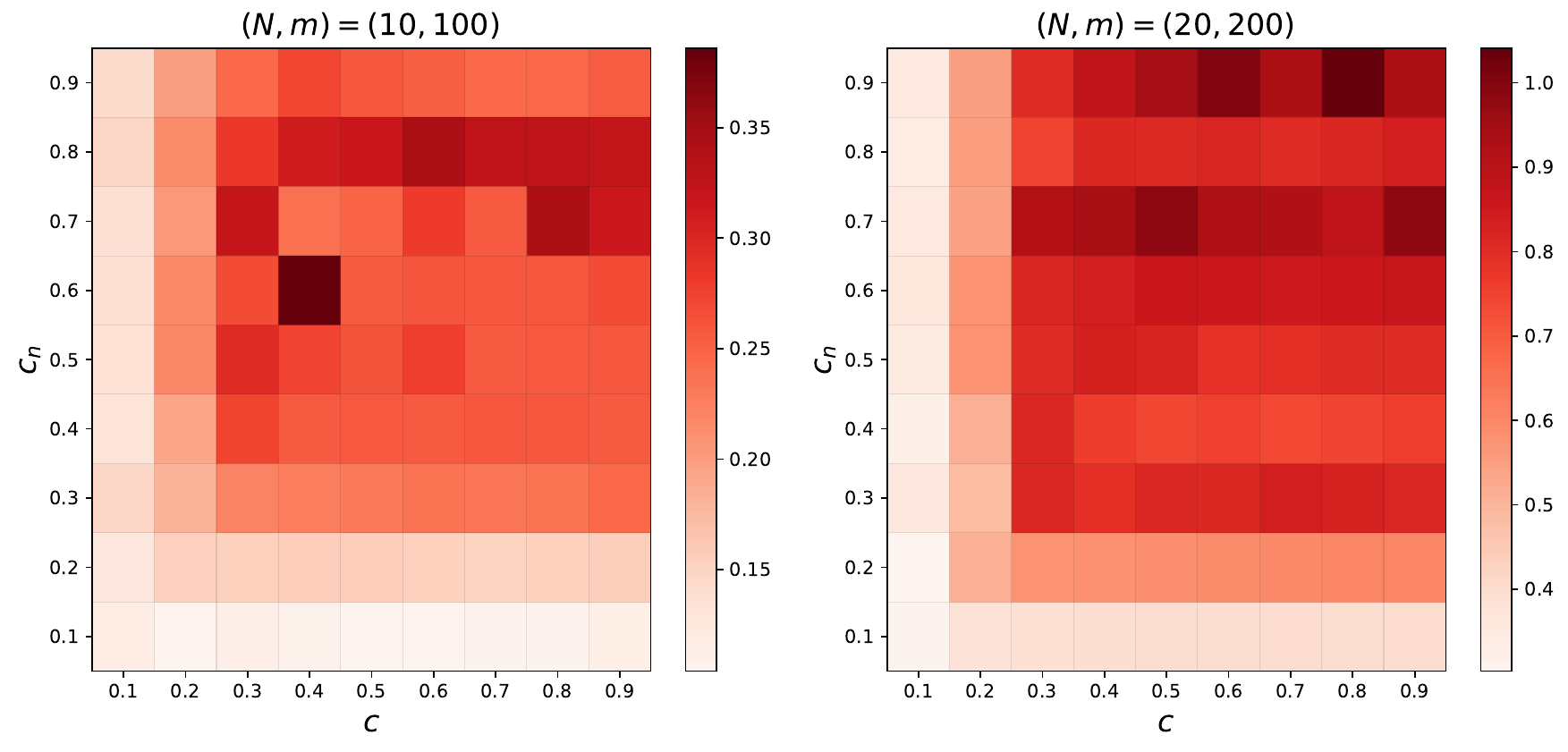}
        \caption{Time(s)}
        \label{fig:GNL-Cap-Time}
    \end{subfigure}

    \caption{Objective and runtime provided by B\&C approach (with OA cuts) when varying assortment capacity on the GNL instances}
    \label{fig:GNL-Cap}
\end{figure}

\paragraph{Impact of Cardinality Parameters $c$ and $c_n$.}
Figures \ref{fig:GNL-Cap-Obj} shows that objective values increase with $c$ and $c_n$ but quickly plateau. For $(N,m)=(10,100)$, the objective rises from 2.446 to about 2.694 as $c_n$ increases to 0.4, with little change thereafter. A similar saturation occurs for $(20,200)$, where values stabilize around 2.760 for $c_n \ge 0.5 $. The modest difference between problem sizes indicates stable and predictable scaling. The same phenomenon can be observed in the Figure \ref{fig:MGNL-Cap-Obj}.

Figure \ref{fig:GNL-Cap-Time} reports runtimes under varying $c$ and $c_n$. For the smaller instance, runtimes remain low (mostly below 0.35 seconds), with a maximum of 0.49 seconds. In contrast, the larger instance exhibits a sharper increase, with runtimes exceeding 1 second for several configurations and peaking at 1.04 seconds when $c_n=0.9$ and $c=0.8$. This shows that higher capacity parameters and larger problem sizes substantially increase computational effort.

\begin{figure}[htb]
    \centering
    \begin{subfigure}{0.48\textwidth}
        \centering
        \includegraphics[width=\textwidth]{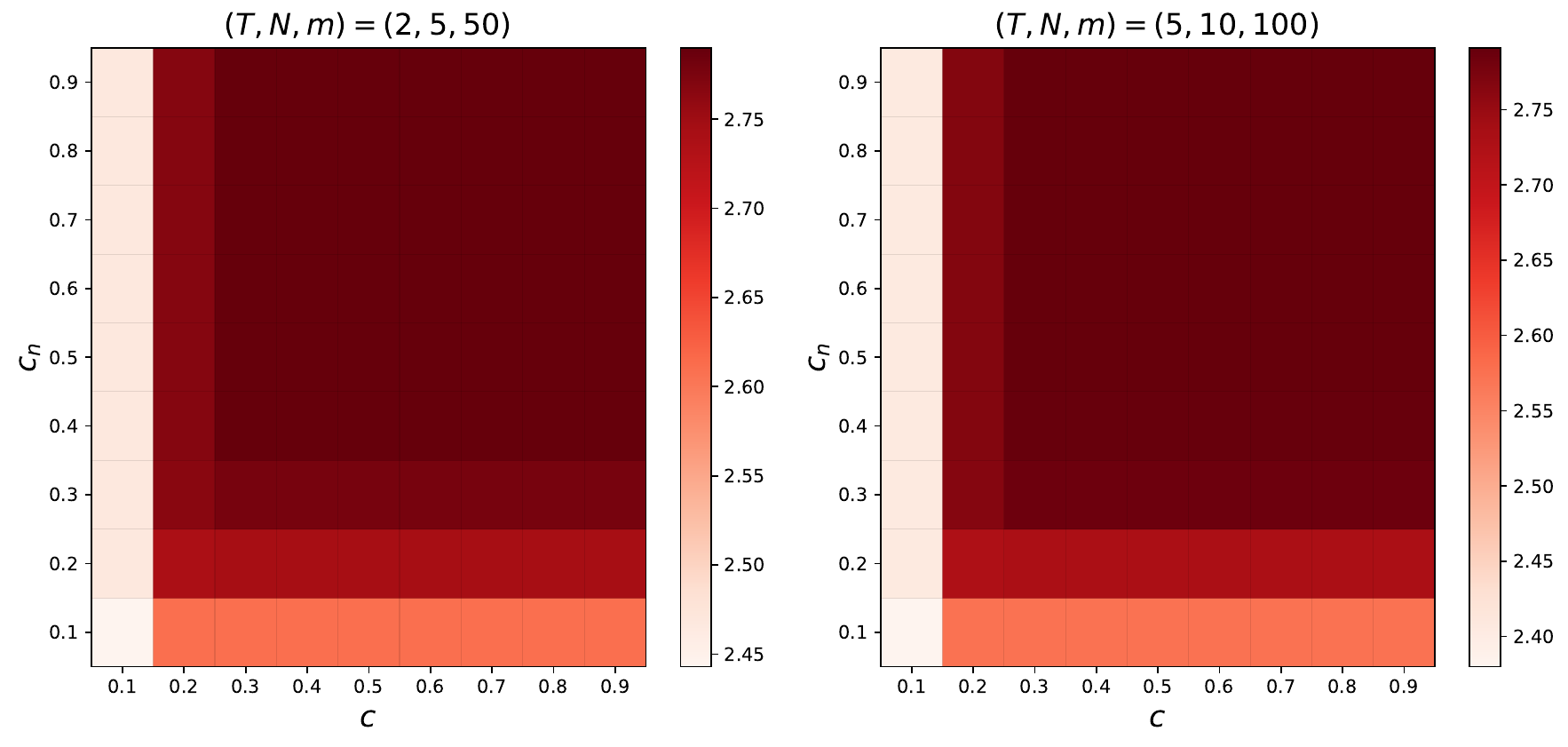}
        \caption{Objective value}
        \label{fig:MGNL-Cap-Obj}
    \end{subfigure}
    \hfill
    \begin{subfigure}{0.48\textwidth}
        \centering
        \includegraphics[width=\textwidth]{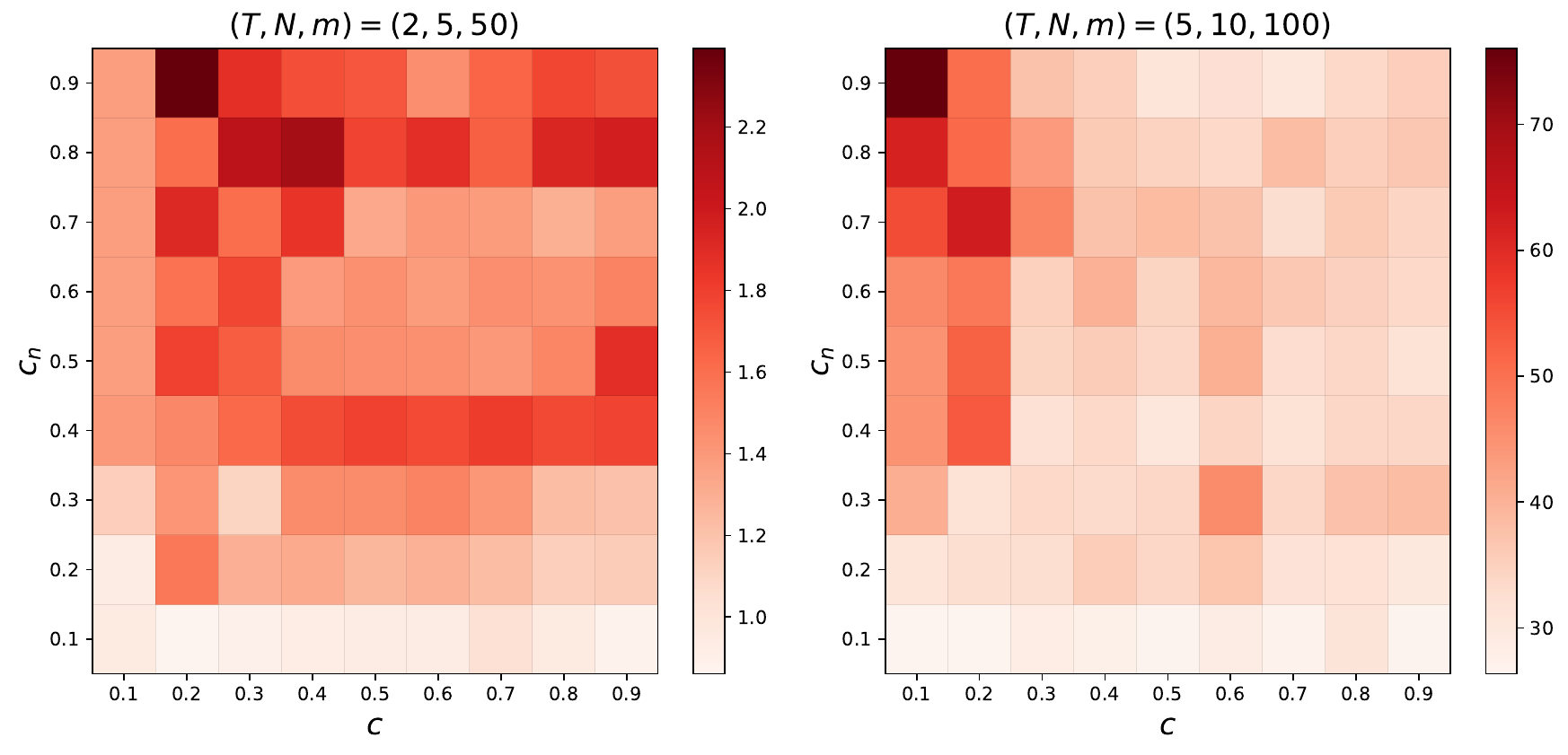}
        \caption{Time(s)}
        \label{fig:MGNL-Cap-Time}
    \end{subfigure}

    \caption{Objective and runtime provided by B\&C approach (with OA cuts) when varying assortment capacity on the MGNL instances}
    \label{fig:MGNL-Cap}
\end{figure}
 
Figure~\ref{fig:MGNL-Cap-Time} reports runtimes for two MGNL datasets: 
$(T,N,m) = (2,5,50)$ and $(5,10,100)$.
For the smaller instance, runtimes remain low (mostly under 2 seconds), with a modest increase as $c$ and $c_n$ approach 0.9. In contrast, the larger instance shows substantially higher and more variable runtimes, exceeding 40 seconds for $c_n \ge 0.7$ and peaking at 76 seconds. Overall, runtime increases with problem size and is more sensitive to $c_n$ than to $c$, reflecting the growing complexity of the solution space.

\begin{figure}[htb]
    \centering
    \begin{subfigure}{0.48\textwidth}
        \centering
        \includegraphics[width=\textwidth]{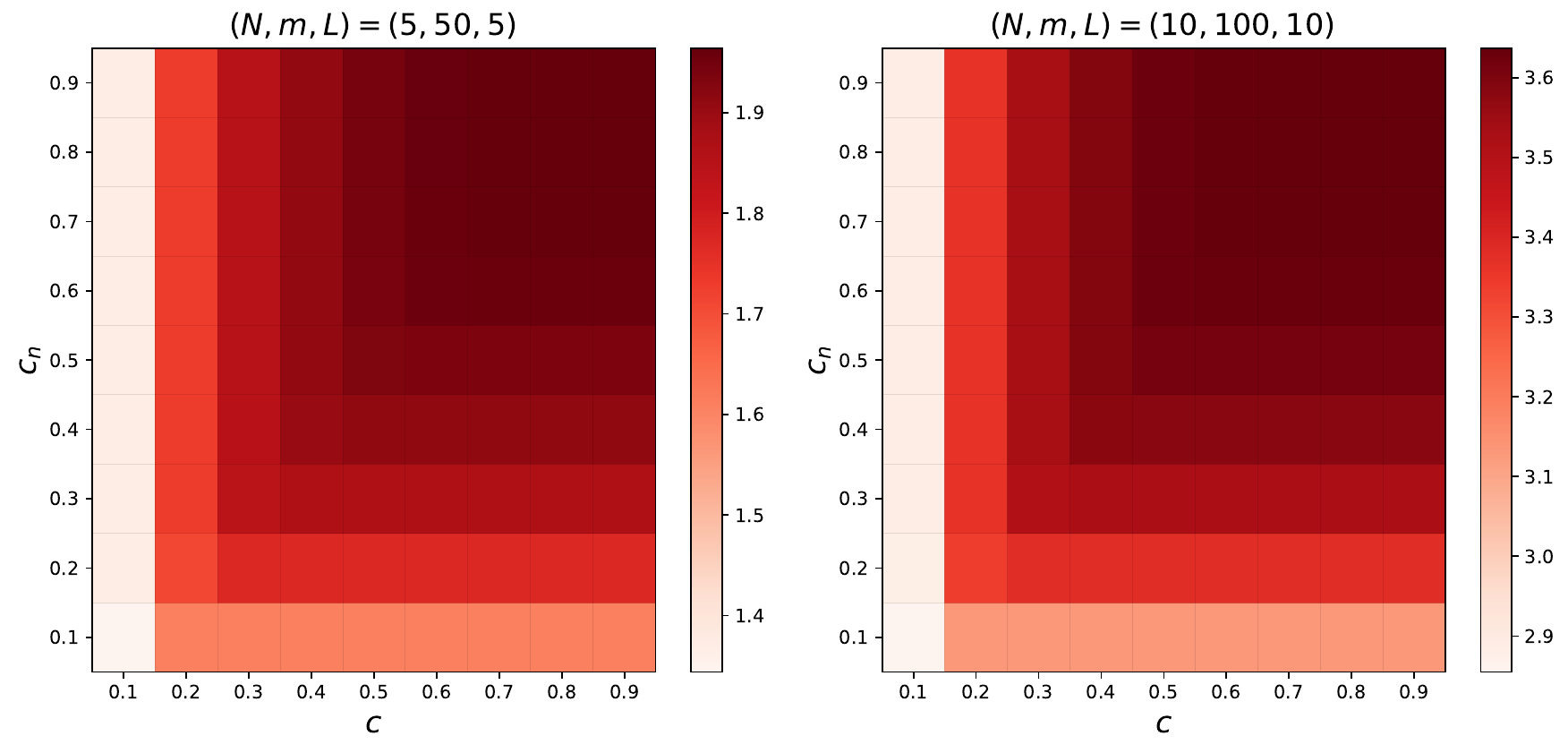}
        \caption{Objective value}
        \label{fig:DP-Cap-Obj}
    \end{subfigure}
    \hfill
    \begin{subfigure}{0.48\textwidth}
        \centering
        \includegraphics[width=\textwidth]{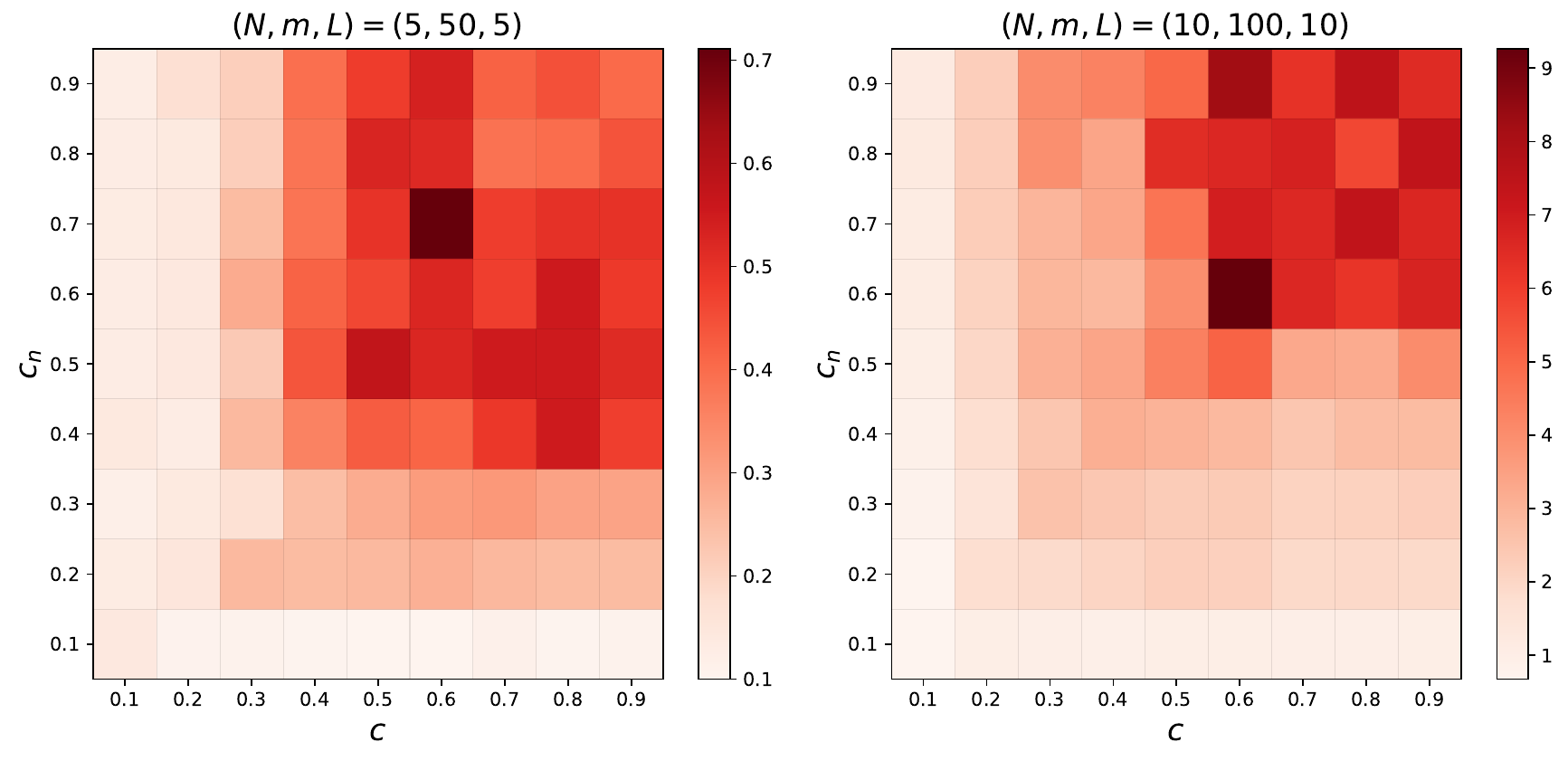}
        \caption{Time(s)}
        \label{fig:DP-Cap-Time}
    \end{subfigure}

    \caption{Objective and runtime provided by B\&C approach (with OA cuts) when varying assortment capacity on the A\&DP instances}
    \label{fig:DP-Cap}
\end{figure}
Figure~\ref{fig:DP-Cap-Obj} shows a smooth and predictable response to changes in the cardinality parameters. For the smaller instance, the objective increases moderately with $c$ and $c_n$ and stabilizes around 1.964 for $c_n \ge 0.7$. The larger instance attains substantially higher values and rises more rapidly before plateauing near 3.636 once $c_n > 0.6$. This saturation behavior indicates good scalability with respect to problem size and diminishing returns from increasing capacity parameters beyond moderate levels. Figure~\ref{fig:DP-Cap-Time} reports runtimes under varying $c$ and $c_n$. For the smaller instance, runtimes remain low (mostly below 0.6 seconds), with only slight increases as parameters grow. In contrast, the larger instance shows a marked rise in computation time for higher parameter values, exceeding 9 seconds in some cases (e.g., $c = 0.6, c_n = 0.6$). Overall, the method scales reasonably, but higher capacity parameters introduce noticeable computational overhead.

\end{document}